\documentclass[a4paper,12pt,reqno]{amsart}
\usepackage{amsmath,amsthm,amssymb,amsfonts}

\usepackage[utf8]{inputenc}
\usepackage[T1]{fontenc}

\usepackage[usenames,dvipsnames]{xcolor}
\usepackage{hyperref}
\hypersetup{
    colorlinks=true,
    linkcolor=cyan!80!black,
    citecolor=MidnightBlue,
    urlcolor=magenta,
}
\usepackage{url}
\usepackage{booktabs}
\usepackage{nicefrac}
\usepackage{microtype}
\usepackage[margin=0.6in]{geometry}
\usepackage{natbib}
\usepackage{graphicx}

\theoremstyle{plain}

\newtheorem{theorem}{Theorem}[section]
\newtheorem{lemma}[theorem]{Lemma}
\newtheorem{corollary}[theorem]{Corollary}

\theoremstyle{remark}

\newtheorem{remark}[theorem]{Remark}
\newtheorem*{example}{Example}

\newcommand{\A}{\mathcal{A}}
\newcommand{\C}{\mathcal{C}}
\newcommand{\I}{\mathbb{I}}
\newcommand{\E}{\mathbb{E}}
\newcommand{\RR}{\mathbb{R}}
\newcommand{\NN}{\mathbb{N}}
\renewcommand{\P}{\mathbb{P}}
\newcommand{\cond}{\mathrm{cond}}
\newcommand{\cov}{\mathrm{Cov}}
\newcommand{\corr}{\mathrm{Cor}}
\newcommand{\var}{\mathrm{Var}}
\newcommand{\pois}{\mathsf{Poisson}}
\newcommand{\bern}{\mathsf{Bernoulli}}
\newcommand{\bone}{\mathbf{1}}
\newcommand{\CI}{\mathrm{CI}}
\newcommand{\DC}{\mathrm{DC}}
\newcommand{\CC}{\mathrm{CC}}
\newcommand{\bZ}{\mathbf{Z}}
\newcommand{\btheta}{\boldsymbol{\theta}}

\let\hat\widehat

\title{A dynamic mean-field statistical model of academic collaboration}

\author[S. S. Mukherjee]{Soumendu Sundar Mukherjee}
\author[T. Sadhukhan]{Tamojit Sadhukhan}
\address{
    Statistics and Mathematics Unit (SMU) \\
    Indian Statistical Institute \\
    203 B. T. Road, Kolkata 700108 \\
    West Bengal, India
}
\email{ssmukherjee@isical.ac.in}
\email{tamojit96sadhukhan@gmail.com}
\author[S. Chatterjee]{Shirshendu Chatterjee}
\address{
    Department of Mathematics\\
    City University of New York \\
    North Academic Center 8/133 \\
    160 Convent Ave \\
    New York, NY, 10031 \\
}
\email{shirshendu@ccny.cuny.edu}

\begin{document}
\maketitle

\begin{abstract}
There is empirical evidence that collaboration in academia has increased significantly during the past few decades, perhaps due to the breathtaking advancements in communication and technology during this period. Multi-author articles have become more frequent than single-author ones. Interdisciplinary collaboration is also on the rise. Although there have been several studies on the dynamical aspects of collaboration networks, systematic statistical models which theoretically explain various empirically observed features of such networks have been lacking. In this work, we propose a dynamic mean-field model and an associated estimation framework for academic collaboration networks. We primarily focus on how the degree of collaboration of a typical author, rather than the local structure of her collaboration network, changes over time. We consider several popular indices of collaboration from the literature and study their dynamics under the proposed model. In particular, we obtain exact formulae for the expectations and temporal rates of change of these indices. Through extensive simulation experiments, we demonstrate that the proposed model has enough flexibility to capture various phenomena characteristic of real-world collaboration networks. Using metadata on papers from the arXiv repository, we empirically study the mean-field collaboration dynamics in disciplines such as Computer Science, Mathematics and Physics.
\end{abstract}

\section{Introduction}\label{sec:intro}
There is empirical evidence that collaboration in academia has significantly increased during the last two decades. Advancements in communication and technology may be one of the main reasons. Multi-author articles have become more frequent than single-author articles. Collaboration between different disciplines is also on the rise \citep{porter2009science}. Quantitative study of collaboration has therefore attracted a great deal of attention and there is already a substantial body of (mostly empirical) work analyzing various aspects of collaboration, both intra- and inter-disciplinary. An incomplete list of relevant works include \cite{lawani1980quality,subramanyam1983collaborative,ajiferuke1988collaborative,de1997modelling, newman2001structure, bettencourt2008population, bettencourt2009scientific, porter2009science, tomasello2017data, abramo2018comparison, liang2018modeling, lalli2020dynamics, hurtado2021analysis, ebrahimi2021developing}.

Despite the large body of existing works, there is a shortage of flexible statistical generative models which are capable of capturing various empirically observed features of real-world collaboration networks and are accompanied by statistical estimation/inference frameworks that enjoy rigorous theoretical guarantees. In this work, we attempt to address this gap and propose a dynamic \emph{mean-field} model which takes into consideration several aspects of scientific collaboration, such as the temporal intensity of producing research articles and the role of past collaboration while selecting new collaborators. The proposed model is mean-field in the sense that we only consider the collaboration dynamics of a typical author in a pool of researchers. We develop both parametric and non-parametric frameworks for estimating the model parameters and establish  consistency and asymptotic normality results for the proposed estimators.

As a first step, we primarily focus on how the degree of collaboration of a typical author, rather than the local structure of her collaboration network, changes over time. We consider several popular indices of collaboration from the literature \citep{lawani1980quality,subramanyam1983collaborative,ajiferuke1988collaborative} and study their dynamics under the proposed model. In particular, we obtain exact formulae for the expectations and temporal rates of change of these indices. Through extensive simulation experiments, we demonstrate that the proposed model has enough flexibility to capture various phenomena characteristic of real-world collaboration networks. Using metadata on papers from the arXiv repository \citep{clement2019arxiv}, we empirically study the mean-field collaboration dynamics in disciplines such as Computer Science, Mathematics and Physics.

The rest of the paper is organised as follows. In Section~\ref{sec:model} we describe the proposed mean-field model. Section~\ref{sec:est} then describes our estimation framework. In Section~\ref{sec:indices}, we discuss several metrics of collaboration. In Section~\ref{sec:theory} we describe our main theoretical findings. The proofs are deferred to the appendix. In Section~\ref{sec:simulations} we describe some simulation experiments. Then in Section~\ref{sec:real-data}, we describe an analysis of Computer Science, Mathematics and Physics papers from the arXiv repository. Section~\ref{sec:conc} is the final section with some concluding remarks and directions for future work.

\section{Set-up and methodology}\label{sec:set-up}
\subsection{The proposed model}\label{sec:model}
Let us denote the authors by $\A = \{a_0, \ldots, a_L\}$. Fix an individual author, say $a_0$. We assume that $a_0$ writes papers at events of an inhomogeneous Poisson process with intensity functional $\lambda(t)$. We denote the event-times of this process by $E_n, n \ge 1$. Further, we assume that she writes exactly one paper at each event. Recall that the total number of events $N[s, t]$ during interval $[s, t]$ is then a $\pois(\int_s^t \lambda(u) \, du)$ random variable. This is simply the total number of papers that $a_0$ has written during the interval $[s, t]$. We denote by $\C_n$ the set of co-authors of $a_0$ at event $E_n$. Also, for $i \in [L] := \{1, \ldots, L\}$, set
\[
    m_{n, i} = \sum_{\ell = 1}^n \bone_{\C_{\ell}}(i),
\]
which is the number of papers $a_{0}$ has co-authored with $a_{i}$ by event $E_n$. Clearly, $0 \le m_{n, i} \le n$. We set $m_{0, i} = 0$ for all $i \in [L]$. Given $\C_{1}, \ldots, \C_n$, at event $E_{n + 1}$ a new paper is written in the following way:
\begin{enumerate}
    \item[(a)] For any $i \in [L]$,
\begin{equation}\label{eq:model-new-paper}
    \bone_{\C_{n + 1}}(i) \mid \C_1, \ldots, \C_n \sim \bern(F_{n + 1}(m_{n, i})),
\end{equation}
where $F_{n + 1}: \{0\} \cup \NN \rightarrow [0, 1]$. Since $m_{n, i} \le n$, we take $F_{n + 1}$ to be supported on $\{0\} \cup [n]$. 
    \item[(b)] Given $\C_1, \ldots, \C_n$, co-authors are included in $\C_{n + 1}$ independently of each other, i.e. the indicators $\{\bone_{\C_{n}}(i)\}_{i \in [L]}$ are conditionally independent given $\C_1, \ldots, \C_n$.
\end{enumerate}
 Thus, $\#\C_{n + 1} \mid \C_1, \ldots, \C_n$ is 
a sum of $L$ independent Bernoulli variables with parameters $F_{n + 1}(m_{n, i}), i \in [L]$. A cartoon of event $E_{n + 1}$ is shown in Figure~\ref{fig:cartoon}.
 
\begin{figure}[!t]
    \centering
    \includegraphics[scale = 0.8]{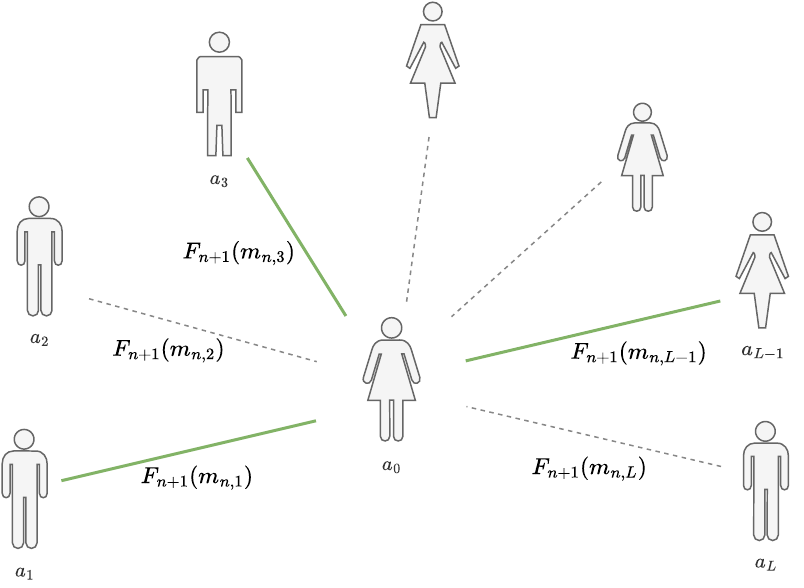}
    \caption{A cartoon of the $(n + 1)$-the paper writing event for author $a_0$: Solid green lines indicate co-authorship; thus $\C_{n + 1} = \{a_1, a_3, a_{L - 1}\}$.}
    \label{fig:cartoon}
\end{figure}

Note that in the above model, there are two (infinite dimensional) unknown parameters which have to be estimated: the intensity functional $\lambda(t)$ of the paper writing process $N$, and the event-specific co-authorship probabilities $\{F_n(k)\}_{n \ge 1, k \ge 0}$. One may of course consider parametric sub-models. For instance, for short intervals of time, one may use a homogeneous Poisson process, thereby reducing the intensity functional to a scalar parameter $\lambda > 0$. More importantly, one parametric model for $F_n(k)$ that we will be of particular interest to us is the following:
\begin{equation}\label{eq:param-submodel}
    F_n(k) = \begin{cases}
        a_n k + b_n & \text{if } 0 \le k \le n - 1, \\
        0 & \text{otherwise},
    \end{cases}
\end{equation}
where $a_n$ and $b_n$ are unknown parameters. Since $F_n$ is supported on $\{0\} \cup [n - 1]$, we must have that $b_n = F_n(0) \in [0, 1]$ and for $n \ge 2$, $-\frac{b_n}{n - 1} \le a_n \le \frac{1 - b_n}{n  - 1}$.

Our main objects of interest will be the observables $X_k[s, t] =$ the number of $k$-author papers written by $a_0$ during $[s, t]$, $k \ge 1$. Notice that
\[
    X_{k}[s, t] = \sum_{n \ge 1} \I(\#\C_{n} = k - 1)\bone_{[s, t]}(E_n)
\]
and
\[
    \sum_{k \ge 1} X_k[s, t] = N[s, t] \sim \pois\bigg(\int_s^t \lambda(u) \, du\bigg).
\]

\subsection{Estimation framework}\label{sec:est}
We now discuss non-parametric estimation of the parameters $\lambda(t)$ and
$(F_n(k))_{n \ge 1, k \ge 0}$. We also consider estimation of $F_n(k)$ under the parametric sub-model described in \eqref{eq:param-submodel}.
\subsubsection{Estimation of the intensity functional $\lambda(t)$.} 
Since we usually have knowledge of approximately when an article is written (e.g., when a paper is first submitted to a journal, or uploaded to a repository such as the arXiv), it is reasonable to assume that we know the timings of the paper writing events $E_n$. Thus we may employ standard techniques for estimating intensity functionals of Poisson processes to estimate $\lambda(t)$. In this article, we will use the kernel smoothing approach of \cite{diggle1985kernel}. To elaborate, let us denote $\int_0^t \lambda(u) \, du$ by $\Lambda(t)$. For sufficiently small $h>0$, we can approximately write
\[
    \lambda(t) \approx \frac{\Lambda(t + h)-\Lambda(t - h)}{2h} .
\]
As $N[t - h, t + h] \sim \pois(\Lambda(t + h)-\Lambda(t - h))$, a natural estimator of $\lambda(t)$ is  
\[
    \hat{\lambda}(t) = \frac{1}{2h}N[t - h, t + h] = \frac{1}{2h}\sum_{n \ge 1} \bone_{[t - h, t + h]}(E_n) = \frac{1}{h} \sum_{n \ge 1} K_{1}\bigg(\frac{E_{n} - t}{h}\bigg),
\]
where $K_{1}(u) = \frac{1}{2} \bone_{[-1, 1]}(u)$. A natural generalisation of this estimator is therefore
\[
    \hat{\lambda}(t)=\frac{1}{h}\sum_{n \ge 1}K\bigg(\frac{E_{n} - t}{h}\bigg),
\]
where $K : \RR \rightarrow \RR_+$ is an appropriate \emph{kernel} (i.e. a non-negative integrable function satisfying $\int_{\RR} K(u) \,du = 1$).

\subsubsection{Estimation of \texorpdfstring{$F_n(k)$}{}}
We estimate $F_n(k)$ by maximising a conditional likelihood. For $n \ge 1$, the conditional likelihood of the parameters $\{F_n(k)\}_{0 \le k \le n - 1}$ given $\C_1, \ldots, \C_{n - 1}$ is
\begin{align*}
    \mathcal{L}_{\cond}\big(\{F_n(k)\}_{0 \le k \le n - 1}\big) &= \prod_{i \in [L]} [F_n(m_{n - 1,i})]^{\bone_{\C_{n}}(i)}[1 - F_n(m_{n - 1, i})]^{1 - \bone_{\C_{n}}(i)} \\ 
                                                      &= \prod_{k = 0}^{n - 1} \prod_{\substack{i \in [L] \\  m_{n - 1, i} = k}}
                                           [F_n(k)]^{\bone_{\C_{n}}(i)} [1 - F_n(k)]^{1 - \bone_{\C_{n}}(i)}.
\end{align*}
We maximise $\mathcal{L}_{\cond}$ with respect to $F_n(k)$ by setting
\[
    0 = \frac{\partial \log \mathcal{L}_{\cond}}{\partial F_{n}(k)} = \sum_{i \in [L]} \I(m_{n - 1, i} = k) \bigg(\frac{\bone_{\C_n}(i)}{F_n(k)} - \frac{1 - \bone_{\C_n}(i)}{1 - F_n(k)}\bigg).
\]
Solving for $F_n(k)$, we get
\begin{equation}\label{eq:npest}
    \hat{F}_n(k) = \frac{\sum_{i \in [L]} \I(m_{n - 1, i} = k) \bone_{\C_{n}}(i)}{\sum_{i \in [L]} \I(m_{n - 1, i} = k)},
\end{equation}
provided that $\sum_{i \in [L]} \I(m_{n - 1, i} = k) \ne 0$. That the latter condition holds almost surely as $L$ goes to infinity is shown in Theorem~\ref{thm:ests-well-defined}. For definiteness, when this condition does not hold, we set $\hat{F}_n(k) = 0$.

\subsubsection{Parametric estimation of \texorpdfstring{$F_n(k)$}{}} We now consider the parametric setup of \eqref{eq:param-submodel}. Then, for $n \ge 1$, the conditional likelihood of the parameters $a_n, b_n$ given $\C_1, \ldots, \C_{n - 1}$ is
\begin{align*}
    \mathcal{L}_{\cond}(a_n, b_n) &= \prod_{i \in [L]} (a_n m_{n - 1, i} + b_n)^{\bone_{\C_{n}}(i)} (1 - a_n m_{n - 1, i} - b_n)^{1 - \bone_{\C_{n}}(i)}.
\end{align*}
We estimate $a_n, b_n$ by maximising $\mathcal{L}_{\cond}$ as before, by setting
\begin{align*}
    0 = \frac{\partial \log \mathcal{L}_{\cond}}{\partial a_n} &= \sum_{i \in [L]}\bigg[\frac{\bone_{\C_n}(i) m_{n - 1, i}}{a_n m_{n - 1, i} + b_n} - \frac{(1 - \bone_{\C_n}(i))m_{n - 1, i}}{1 - a_n m_{n - 1, i} - b_n}\bigg], \\
    0 = \frac{\partial \log \mathcal{L}_{\cond}}{\partial b_n} &= \sum_{i \in [L]}\bigg[\frac{\bone_{\C_n}(i)}{a_n m_{n - 1, i} + b_n} - \frac{1 - \bone_{\C_n}(i)}{1 - a_n m_{n - 1, i} - b_n}\bigg].
\end{align*}
Solving for $a_n$ and $b_n$ we obtain
\begin{align}\label{eq:pesta}
    \hat{a}_n &= \frac{\frac{1}{L}\sum_{i \in [L]}\bone_{C_n}(i) m_{n - 1,i} - \frac{1}{L}\sum_{i \in [L]}\bone_{C_n}(i) \frac{1}{L}\sum_{i \in [L]} m_{n - 1, i}}{\frac{1}{L}\sum_{i \in [L]}m^2_{n - 1, i} - \big(\frac{1}{L}\sum_{i \in [L]}m_{n - 1, i}\big)^2}, \\ \label{eq:pestb}
    \hat{b}_n &= \frac{1}{L}\sum_{i \in [L]}\bone_{C_n}(i) - \frac{\hat{a}_n}{L}\sum_{i \in [L]} m_{n - 1, i},
\end{align}
provided that $\frac{1}{L}\sum_{i \in [L]}m^2_{n - 1, i} \ne \big(\frac{1}{L}\sum_{i \in [L]}m_{n - 1, i}\big)^2$. That the latter condition holds almost surely as $L$ goes to infinity is shown in Theorem~\ref{thm:ests-well-defined}. For definiteness, when this condition does not hold, we set $\hat{a}_n = 0$. Incidentally, the estimators turn out to be the same as the least squares estimators of $a_n, b_n$ in the linear model $\bone_{\C_n}(i) = a_n m_{n - 1, i} + b_n$.

\section{Indices for measuring collaboration}\label{sec:indices}
In the scientometrics literature, several indices of collaboration have been proposed. Suppose that $f_{k}$ denotes frequency of $k$-authors papers in a body of literature consisting of $J = \sum_{k \ge 1} f_k$ papers in total. For example, the \emph{Collaborative Index} (CI) defined in \cite{lawani1980quality} is given by
\[
    \CI = \frac{\sum_{k \ge 1} k f_k}{J}.
\]
The \emph{Degree of Collaboration (DC)} \citep{subramanyam1983collaborative} is defined as
\[
    \DC = 1 - \frac{f_1}{J}.
\]
Note that $\DC$ is the fraction of multi-author papers.

Another measure of interest is the so-called \emph{Collaborative Coefficient (CC)} \citep{ajiferuke1988collaborative}:
\[
    \CC = 1 - \sum_{k \ge 1} \frac{1}{k} \frac{f_k}{J}.
\]
To understand this index, imagine that each paper comes with a credit of \$1, which is shared equally by all its authors. Thus for a $k$-author paper, each author receives a credit of \$$\frac{1}{k}$ each. Thus $1 - \CC$ is the average credit received for writing a paper. The more the average credit, the less the extent of collaboration. Note that $0 \le \CC < 1$. $\CC = 0$ if and only if only we have only single-author papers. 

Note that all of the indices discussed so far are static in nature. We now describe dynamic mean-field versions of these indices. To that end, note that if we impose a mean-field assumption on the underlying collaboration network, then it is enough to consider local versions of these indices, i.e. compute them locally on the collaboration network of a single author $a_0$. Thus the collaborative index during $[s, t]$ becomes the average number of co-authors of $a_0$ during $[s, t]$:
\begin{equation}
    \label{eq:def-CI}
    I_{\CI}[s, t] := \sum_{k \ge 1} (k - 1) \frac{X_k[s, t]}{N[s, t]}.
\end{equation}
Note that here we slightly deviate from the definition of $\CI$ by subtracting $1$ from the analogous expression.

Similarly, the dynamic mean-field version of the degree of collaboration may be defined as follows:
\begin{equation}
    \label{eq:def-DC}
    I_{\DC}[s, t] = 1 - \frac{X_1[s, t]}{N[s, t]}.
\end{equation}
In words, $I_{\DC}[s, t]$ is the the fraction of multi-author papers out of all papers written by $a_0$ during time-interval $[s, t]$.

Finally, the dynamic mean-field version of the collaborative coefficient is given by
\begin{equation}
    \label{eq:def-CC}
    I_\CC[s, t] = 1 - \sum_{k \ge 1} \frac{1}{k} \frac{X_k[s, t]}{N[s, t]}.
\end{equation}

In fact, we can define a generalization of these indices by noting that they all can be expressed as 
$\sum_{k \ge 1} \varphi(k) \frac{X_k[s, t]}{N[s, t]}$, for a suitable non-decreasing function $\varphi : \NN \rightarrow \RR$, satisfying $\varphi(1) = 0$. In particular,  $\varphi(k) = (k - 1)$ for $I_\CI$, $\varphi(k) = \I(k \ge 2)$ for $I_\DC$, and $\varphi(k) = 1 - \frac{1}{k}$ for $I_\CC$.
Therefore, given any non-decreasing function $\varphi : \NN \rightarrow \RR$, we define the \emph{generalised index of collaboration} as 
\begin{equation}\label{eq:def-comb ci}
    I_\varphi[s, t] := \sum_{k \ge 1} \varphi(k) \frac{X_k[s, t]}{N[s, t]}. 
\end{equation}
In Section~\ref{sec:theory-indices}, we derive formulae for the expectation and temporal rate of change of $I_\varphi$. In Section~\ref{sec:simulations}, we empirically study the impact of intensity functional and the co-authorship probability parameters $\{F_n(k)\}_{n \ge 1, k \ge 0}$ on the indices $I_{\CI}$, $I_{\DC}$ and $I_{\CC}$.

\section{Theory}\label{sec:theory}
In this section, we describe our main theoretical findings. The proofs are deferred to the appendix.
\subsection{Instantaneous behaviour of \texorpdfstring{$X_k, k \ge 1$}{}}
For small $h$, we can obtain explicit first-order approximations for the mean and variance of $X_k[t, t + h]$, $k \ge 1$, and also for the covariance and correlation of $X_k[t, t + h]$ and $X_{k'}[t, t+ h]$, $k \ne k'$. These depend crucially on the joint distribution of $\#\C_{N[0, t] + 1}$ and $\#\C_{N[0, t] + 2}$, and the marginal distribution of $\#\C_{N[0, t] + 1}$. To that end, define for any $k, k' \ge 1$, 
\begin{align}
    G_t(k, k') &= \P(\#\C_{N[0, t] + 1} = k, \#\C_{N[0, t] + 2} = k'), \text{ and } \\
    H_t(k) &= \P(\#\C_{N[0, t] + 1} = k).
\end{align}
\begin{theorem}\label{thm:xk-moments}
We have, for any $k \ge 0$,
\begin{equation}\label{eq:xk-mean-var}
    \lim_{h \downarrow 0} \frac{\E(X_{k + 1}[t, t + h])}{h} = \lim_{h \downarrow 0} \frac{\var(X_{k + 1}[t, t + h])}{h} = \lambda(t) H_t(k). 
\end{equation}
On the other hand, for any $k \ne k' \ge 0$,
\begin{equation}\label{eq:xk-cov}
    \lim_{h \downarrow 0} \frac{\cov(X_{k + 1}[t, t + h], X_{k' + 1}[t, t + h])}{h^2} = (\lambda(t))^2 \bigg(\frac{G_t(k, k') + G_t(k', k)}{2} - H_t(k) H_t(k')\bigg).
\end{equation}
As a consequence, for any $k \ne k' \ge 0$, 
\begin{equation}\label{eq:xk-corr}
    \lim_{h \downarrow 0} \frac{\corr(X_{k + 1}[t, t + h], X_{k' + 1}[t, t + h])}{h} = \lambda(t)\frac{\frac{1}{2}(G_t(k, k') + G_t(k', k)) - H_t(k) H_t(k')}{\sqrt{H_t(k) H_t(k')}}.
\end{equation}
\end{theorem}

\subsection{Indices of collaboration}\label{sec:theory-indices}
We first give an expression for the expectation of the generalised index of collaboration $I_\varphi$.
\begin{lemma}\label{lem:ind-expression}
We have
\[
    \E I_\varphi[s, t] = \sum_{n \ge 1} \E \bigg[ \frac{\bone_{[s, t]}(E_n)}{N[s, t]} \bigg] \E \varphi(\# \C_n + 1).
\]
\end{lemma}
In Lemma~\ref{lem:ind-expression}, the effect of the Poisson process is decoupled from the co-author inclusion mechanism. The terms corresponding to the Poisson process can be computed explicitly.
\begin{lemma}\label{lem:UV}
We have
\[
    \E \bigg[\frac{\bone_{[s, t]}(E_n)}{N[s, t]}\bigg] = \E \bigg[\frac{\I(U \le n - 1, U + V \ge n)}{V}\bigg] = \sum_{k = 0}^{n - 1} \P(U = k) \E \bigg[ \frac{\I(V \ge n - k)}{V} \bigg],
\]
where $U = N[0, s]$ and $V = N[s, t]$.
\end{lemma}

On the other hand, obtaining explicit expressions for $\E \varphi(\# \C_n + 1)$ is difficult. However, in some special cases we can have explicit formulae. For instance, in case of $I_{\CI}$, $\varphi(k) = k - 1$. If we further assume $F_n$ to be linear as in \eqref{eq:param-submodel}, then we have the following recursion for $\E \varphi(\#\C_n + 1) = \E \#\C_n$.
\begin{lemma}\label{lem:linear_F}
Consider the parametric sub-model \eqref{eq:param-submodel}. Then
\[
    \E \#\C_{n} = a_n \sum_{\ell = 1}^{n - 1} \E \#\C_{\ell} + L b_n.
\]
\end{lemma}
We can explicitly solve the recurrence developed in Lemma~\ref{lem:linear_F}.
\begin{lemma}\label{lem:linear_F_soln}
Consider the parametric sub-model \eqref{eq:param-submodel}. Then $\E\#\C_1 = Lb_1, \E\#\C_2 = L(b_2 + a_2 b_1)$, and for $n \ge 3$,
\[
    \E\#\C_n = L \bigg(b_n + a_n b_{n - 1} + \sum_{j = 1}^{n - 2} b_j \prod_{\ell = j + 1}^{n - 1} (1 + a_\ell)\bigg).
\]
\end{lemma}

\begin{example}
    Take $F_n(0) = b_n = q(1 - \frac{1}{\log(n + 2)})$, and $a_n = \frac{p}{n}$, where $p, q \in [0, 1]$. In Figure~\ref{fig:rec-sol}-(a), we plot $\E\#\C_n$ as a function of $n$, for $p = 0.4$, $q = 0.05$. In Figure~\ref{fig:rec-sol}-(b), we plot $F_n(k)$ as a function of $n$ for $k = 0, 1, 2, 3$. This plot is qualitatively quite similar to what we observe in our real-data analyses in Section~\ref{sec:real-data}.
\end{example}
\begin{figure}[!t]
     \centering
     \begin{tabular}{cc}
     \includegraphics[width = 0.45\textwidth]{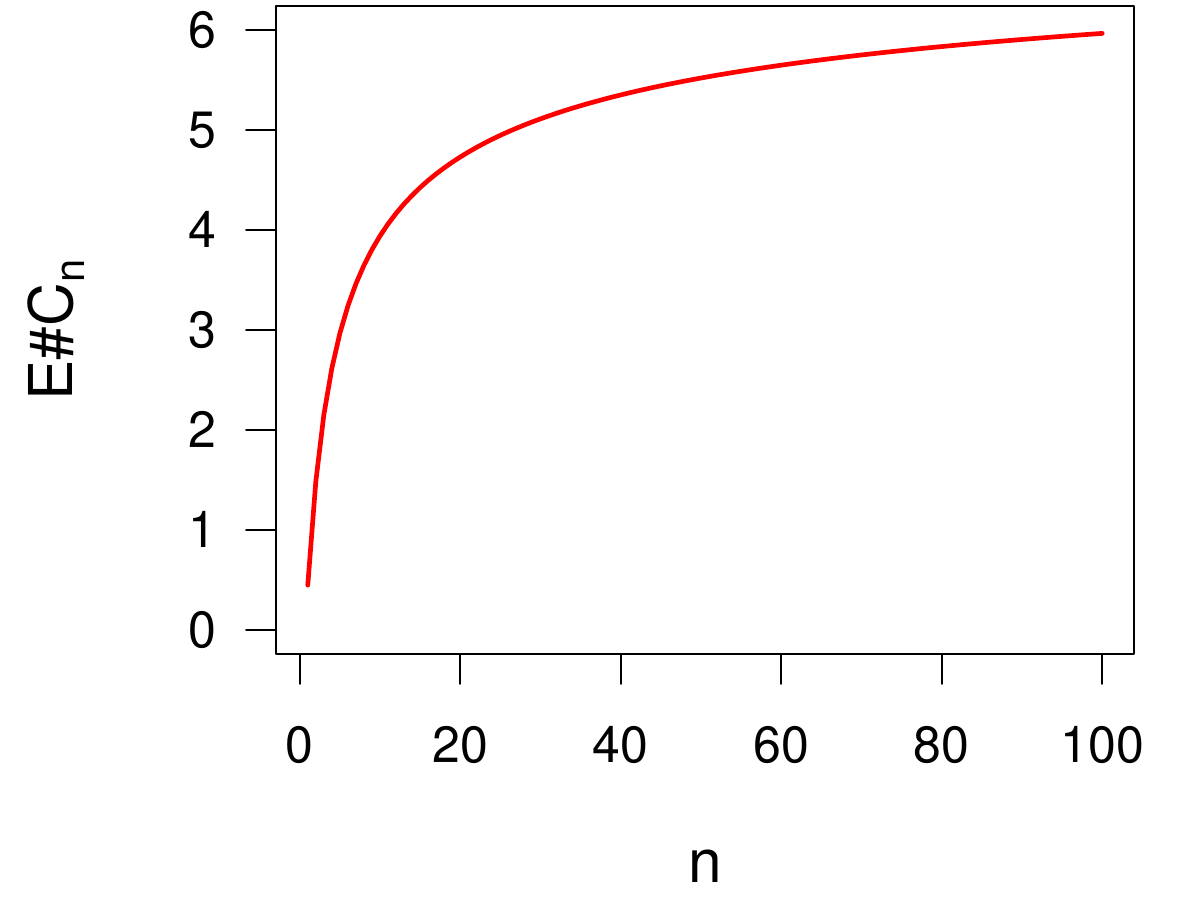} & \includegraphics[width =  0.45\textwidth]{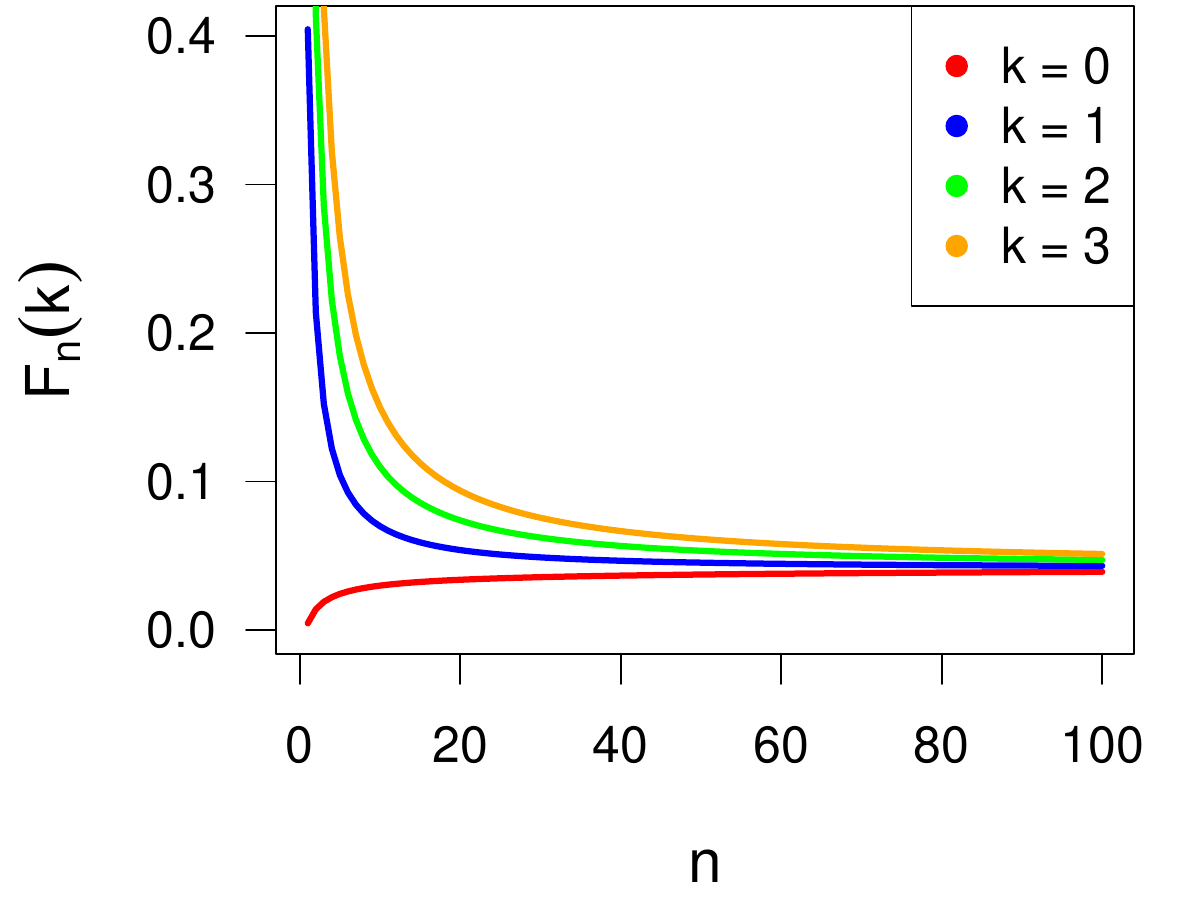} \\
     (a) & (b)
     \end{tabular}
     \caption{(a) $\E\#\C_n$ versus $n$, (b) $F_n(k) = a_n k + b_n$ versus $n$ for $k = 0, 1, 2, 3$, for $a_n = \frac{0.4}{n}$, $b_n = 0.05 (1  - \frac{1}{\log(n + 2)})$, with $L = 100$.}
      \label{fig:rec-sol}
\end{figure}

For small $h$, we can obtain an expression for $\E I_\varphi[t, t + h]$ in terms of $\varphi$ and $\#\C_{N[0, t] + 1}$, up to the first order, in the spirit of Theorem~\ref{thm:xk-moments}.
\begin{theorem}\label{thm:ind_rate}
We have
\[
    \lim_{h \downarrow 0} \frac{\E I_\varphi[t, t+h]}{h} = \lambda(t) \E \varphi(\#\C_{N[0, t] + 1} + 1).
\]
\end{theorem}

\subsection{An uninteresting special case}
A main feature of the model \eqref{eq:model-new-paper} is that the co-author sets $\C_n$ are dependent. We now look at a situation where they become independent. A fortiori, the correlation between $X_k[t, t + h]$ and $X_{k'}[t, t + h]$ becomes second order in $h$, and $\E I_{\varphi}[s, t]$ becomes a constant.
\begin{corollary} \label{cor:const_F}
If for all $n \ge 1$, $F_n(k)$ does not depend on $k$, i.e. $F_n(k) = p_n$ for some $p_n \in [0, 1]$, for all $0 \le k \le n - 1$, then the random sets $\C_n$ are independent. Further, if $F_n(k)$ depends neither on $k$ nor on $n$, i.e. if for all $n \ge 1$, $F_{n}(k) = p$ for some $p \in [0, 1]$, for all $0 \le k \le n - 1$, then
\begin{itemize}
    \item [(i)] $\lim_{h\downarrow 0}\frac{\corr(X_{k + 1}[t, t + h], X_{k' + 1}[t, t + h])}{h} = 0$ for any $k \ne k' \ge 0$.
    \item [(ii)] $\E I_{\varphi}[s, t] = \E\varphi(B + 1)$, where $B \sim \mathrm{Binomial}(L, p)$.
\end{itemize}
\end{corollary}

\subsection{Asymptotic properties of the estimators}\label{sec:theory-estimators}

We first show that the estimators given by \eqref{eq:npest}, \eqref{eq:pesta} and \eqref{eq:pestb} are all almost-surely well-defined as $L \to \infty$.
\begin{theorem}\label{thm:ests-well-defined}
We have the following:
\begin{itemize}
    \item[(a)] Fix $n \ge 1$ and $k \in \{0\} \cup [n]$. Then 
    \[
        \P\bigg(\liminf_{L \to \infty}\sum_{i \in [L]}\I\{m_{n, i} = k\} > 0 \bigg) = 1.
    \]
    \item[(b)] Fix $n \ge 1$. Then
    \[
        \P\bigg(\liminf_{L \to \infty} \bigg[\frac{1}{L}\sum_{i \in [L]} m^2_{n, i} - \bigg(\frac{1}{L}\sum_{i \in [L]}m_{n, i}\bigg)^2\bigg] > 0\bigg) = 1.
    \]
\end{itemize}
\end{theorem}

Note that for fixed $n$, $\{m_{n,i}\}_{i \in [L]}$ are i.i.d. Let $p_{n,k} := \P(m_{n,1} = k)$ for $k \in \{0\} \cup [n]$.
\begin{theorem}[Consistency and asymptotic normality of $\hat{F}_n(k)$]\label{thm:consistency-nonparam}
    Fix $n \ge 1$ and $k \in \{0\} \cup [n - 1]$. Then, as $L \to \infty$, 
\begin{align}\label{eq:asconv}
    \hat{F}_n(k) \stackrel{p}{\to} F_{n}(k)
\end{align}
and
\begin{align}\label{eq:wconv}
    \sqrt{L} \big(\hat{F}_n(k) - F_n(k)\big) \stackrel{d}{\to} N(0, \sigma^2),
\end{align}
where
\begin{align}\label{eq:av}
    \sigma^2 = \frac{F_n(k)(1 - F_n(k))}{p_{n - 1, k}}.
\end{align}
\end{theorem}

\begin{corollary}[Asymptotic confidence interval for ${F}_n(k)$]\label{cor:confint}
Let
\begin{align}\label{eq:ev}
    \hat{\sigma}^2 := \frac{\hat{F}_n(k)(1 - \hat{F}_n(k))}{\frac{1}{L}\sum_{i \in [L]}\I(m_{n - 1, i} = k)}.
\end{align}
Then an $(1 - q)\%$ asymptotic confidence interval for ${F}_n(k)$ is given by
\[
   \bigg(\hat{F}_n(k) - \frac{\hat{\sigma}}{\sqrt{L}} z_{q/2}, \hat{F}_n(k) + \frac{\hat{\sigma}}{\sqrt{L}} z_{q/2}\bigg),
\]
where $z_{q/2}$ is the $(1 - q/2)$-th quantile of the standard normal distribution.
\end{corollary}

\begin{theorem}[Consistency and asymptotic normality of $\hat{a}_n$ and $\hat{b}_n$]
\label{thm:consistency-param}
Fix $n \ge 1$. Then, as $L \to \infty$, 
\[
    \hat{a}_n \stackrel{p}{\to} a_n, \quad \hat{b}_n \stackrel{p}{\to} b_n
\]
and
\begin{equation}\label{eq:anpest}
    \sqrt{L} \big(\hat{a}_n - {a}_n\big) \stackrel{d}{\to} N(0, \sigma_{a}^2), \quad \sqrt{L} \big(\hat{b}_n - {b}_n\big) \stackrel{d}{\to} N(0, \sigma_{b}^2),
\end{equation}
for some $\sigma^2_{a}, \sigma^2_{b} > 0$.
\end{theorem} 
\begin{remark}
    Explicit expressions for $\sigma_a^2$ and $\sigma_b^2$ are derived in the proof of Theorem~\ref{thm:consistency-param}. Unfortunately, these expressions are not very pleasant to look at!
\end{remark}

\begin{corollary}[Asymptotic confidence intervals for $\hat{a}_n$ and $\hat{b}_n$]\label{cor:confint-param}
There exists consistent estimators $\hat{\sigma}^2_a$ and $\hat{\sigma}^2_b$ of $\sigma^2_a$ and $\sigma^2_b$, respectively. An $(1 - q)\%$ asymptotic confidence intervals for $a_n$ is given by
\[
   \bigg(\hat{a}_n - \frac{\hat{\sigma}_a}{\sqrt{L}} z_{q/2}, \hat{a}_n + \frac{\hat{\sigma}_a}{\sqrt{L}} z_{q/2}\bigg),
\]
and an $(1 - q)\%$ asymptotic confidence intervals for $b_n$ is given by
\[
   \bigg(\hat{b}_n - \frac{\hat{\sigma}_b}{\sqrt{L}} z_{q/2}, \hat{b}_n + \frac{\hat{\sigma}_b}{\sqrt{L}} z_{q/2}\bigg),
\]
where, $z_{q/2}$ is the $(1 - q/2)$-th quantile of the standard normal distribution.
\end{corollary}

\section{Simulation studies}\label{sec:simulations}
In this section, we study the behaviour of the indices of collaboration for various choices of the intensity functional $\lambda(t)$ and co-authorship probability parameters $(F_n(k))_{n \ge 1, k \ge 0}$ through a series of simulation experiments.

We plot the dynamics of the three collaboration coefficients ($I_\CI$, $I_\DC$, and $I_\CC$) in several different settings.
In each setting, we use $L = 100$ authors and report the values of the indices per year (i.e. $t - s = 1$ year), averaged over $10$ Monte Carlo runs. These experiments are all run on a laptop with 8GB of RAM and an Intel Core i3-6006U CPU. 

\subsection{Effect of \texorpdfstring{$F_n(k)$}{}.}\label{sec:effect-F}
In the first three settings (shown in Figures~\ref{fig:simulation_1}-\ref{fig:simulation_3}), we take $\lambda(t)$ to be constant, so that the effect of the collaboration probabilities $F_n(k), n \ge 1, k \ge 0$ can be understood cleanly.

In the simplest case, shown in Figure~\ref{fig:simulation_1}, we take $F_n(k), 0 \le k \le n - 1, n \ge 1,$ to be constant free of both $n$ and $k$. Note that all the three collaboration coefficients are essentially constant. This is what was shown in Corollary~\ref{cor:const_F}.

In the setting of Figure~\ref{fig:simulation_2}, we set $F_n(k)$ to be a function of $k$ only. In this case, the model does take the effect of past collaboration into account and this highly influences the dynamics of the various indices as seen from Figure~\ref{fig:simulation_2}.

Finally, in the setting of Figure~\ref{fig:simulation_3}, we set $F_n(\cdot)$ to be a (non-decreasing) function of $n$ only. Although the model does not take past history of collaboration into account in this case, the probabilities of collaboration (i.e. $F_n(k)$ for $k \ge 1$) do increase over the course of writing papers. This explains the non-decreasing nature of the curves shown in Figure~\ref{fig:simulation_3}.

\subsection{Effect of \texorpdfstring{$\lambda(t)$}{}}\label{sec:effect-lambda}
In the settings of Figures~\ref{fig:suppl_simulation_1}-\ref{fig:suppl_simulation_6}, we use non-constant $\lambda(t)$ and the same $F_n(k)$'s as in Section~\ref{sec:effect-F}. We use a piecewise constant $\lambda(t)$ in Figures~\ref{fig:suppl_simulation_1}-\ref{fig:suppl_simulation_3} and a piecewise linear one in Figures~\ref{fig:suppl_simulation_4}-\ref{fig:suppl_simulation_6}.

We find from these experiments that the co-authorship probability functions $F_n(k)$ have a much greater impact on the indices of collaboration than $\lambda(t)$. Also, the simulations in Figures~\ref{fig:suppl_simulation_1} and \ref{fig:suppl_simulation_2} confirm part (ii) of Corollary~\ref{cor:const_F}.

\begin{figure}[!htbp]
     \centering
     \begin{tabular}{ccc}
     \includegraphics[width = 0.3\textwidth]{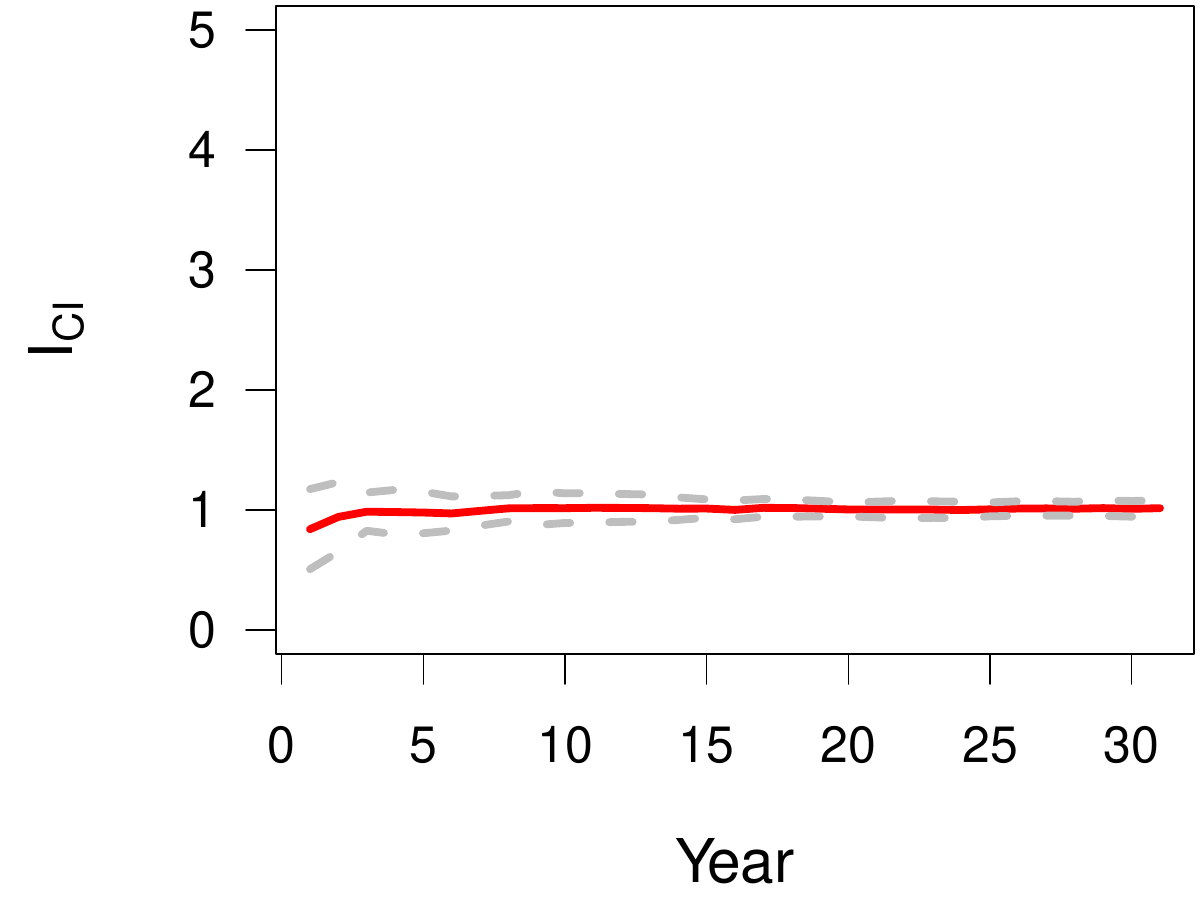} & \includegraphics[width = 0.3\textwidth]{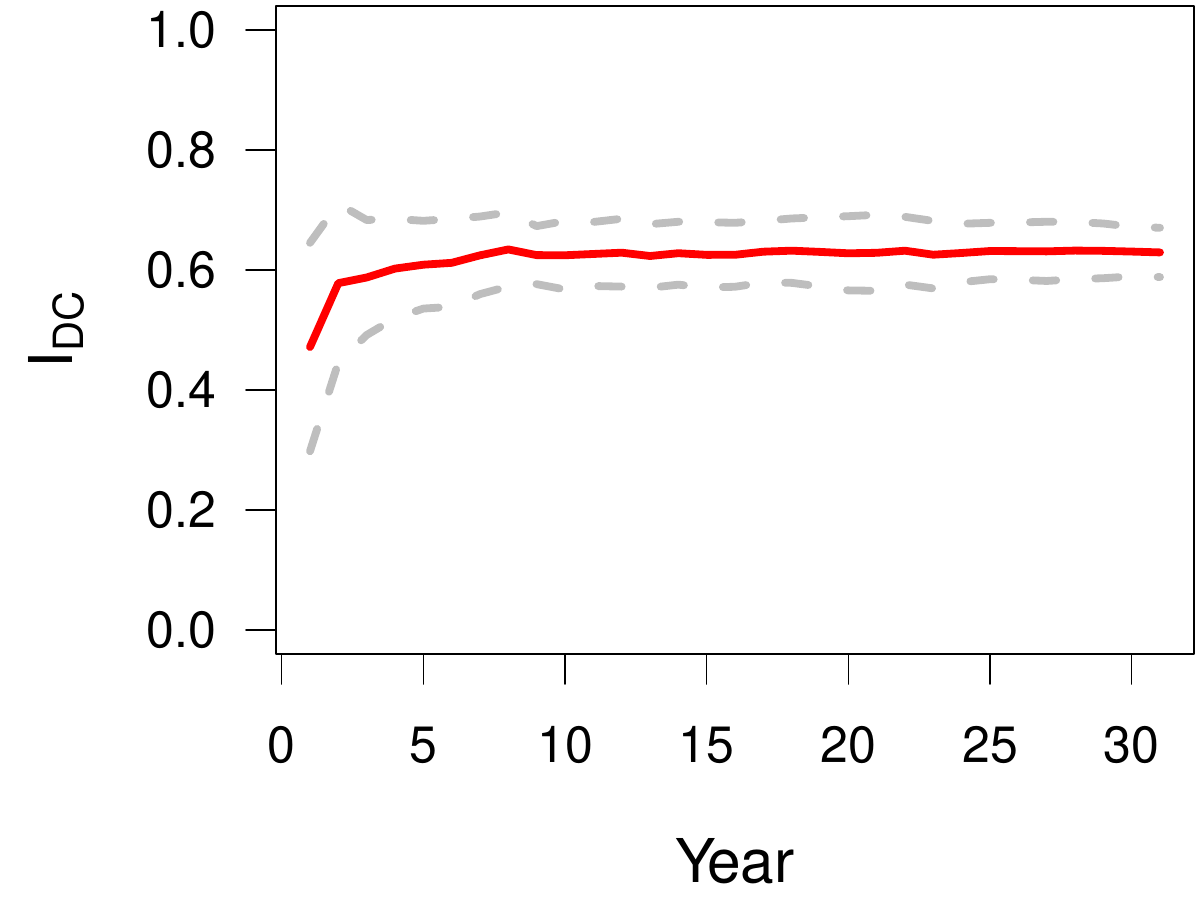} & \includegraphics[width = 0.3\textwidth]{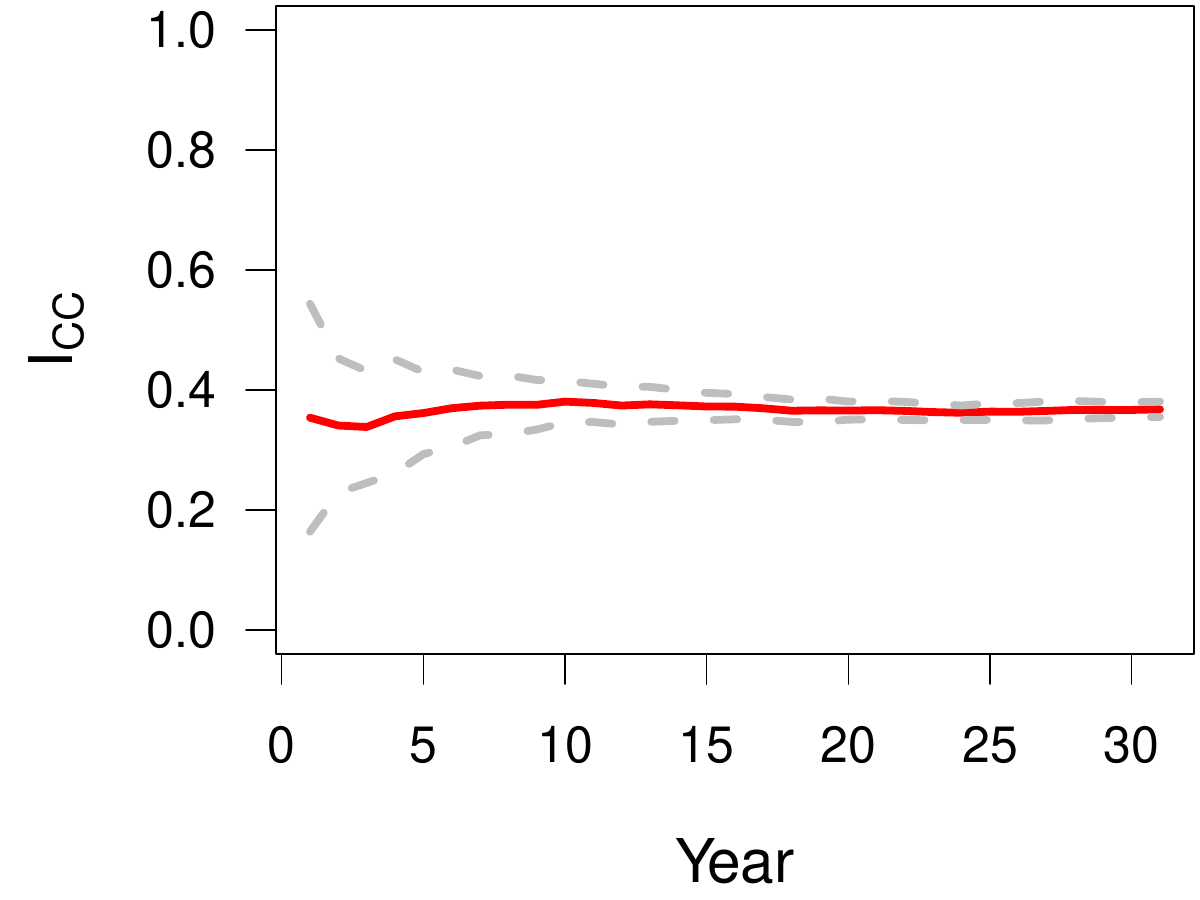} 
     \end{tabular}
     \caption{Different indices of collaboration for $\lambda(t) = 0.5$ per month, and $F_n(k) = 0.01$ for all $n \ge 1, k \ge 0$. In each subplot, the red line depicts yearly values of the corresponding index averaged over $10$ Monte Carlo runs. The dashed gray lines show $\pm 1$ standard errors.}
     \label{fig:simulation_1}
\end{figure}

\begin{figure}[!htbp]
     \centering
     \begin{tabular}{ccc}
     \includegraphics[width = 0.3\textwidth]{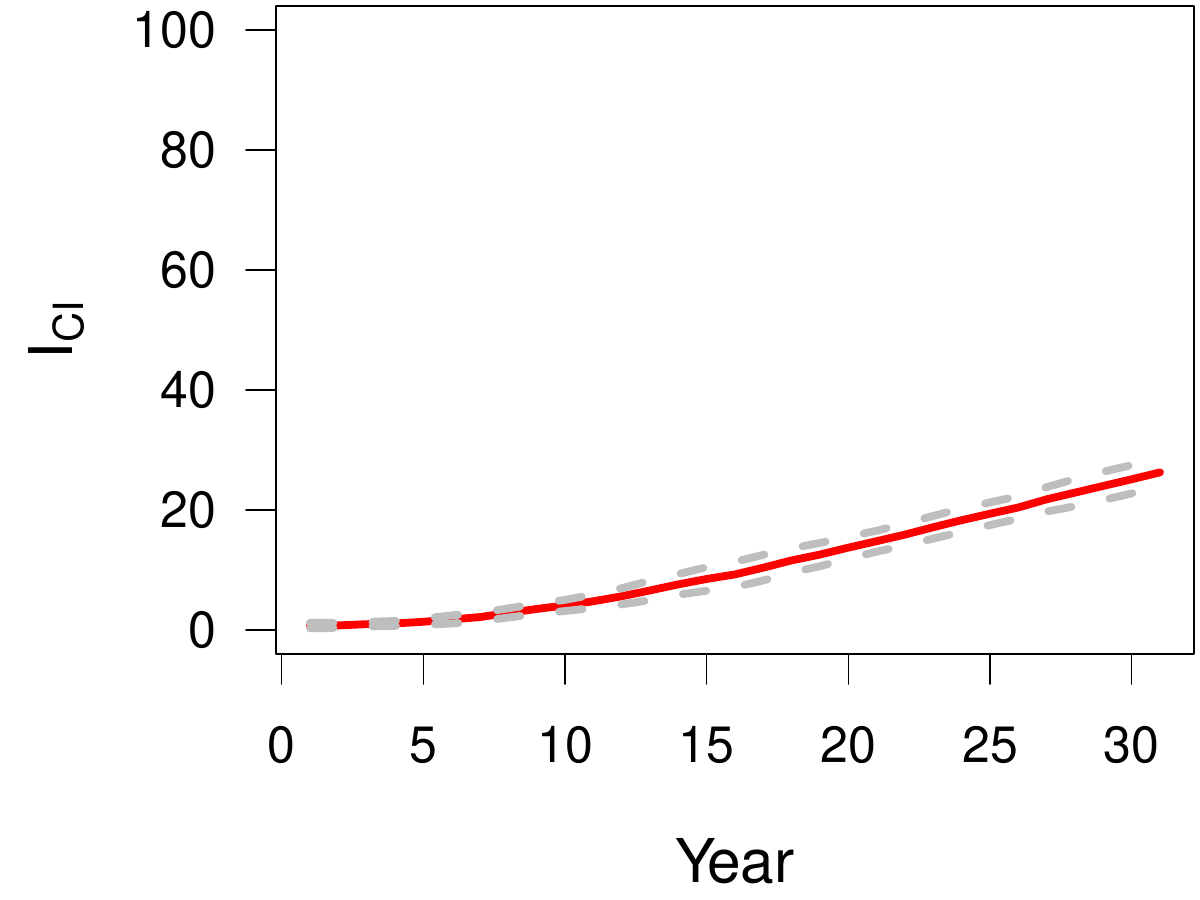} & \includegraphics[width = 0.3\textwidth]{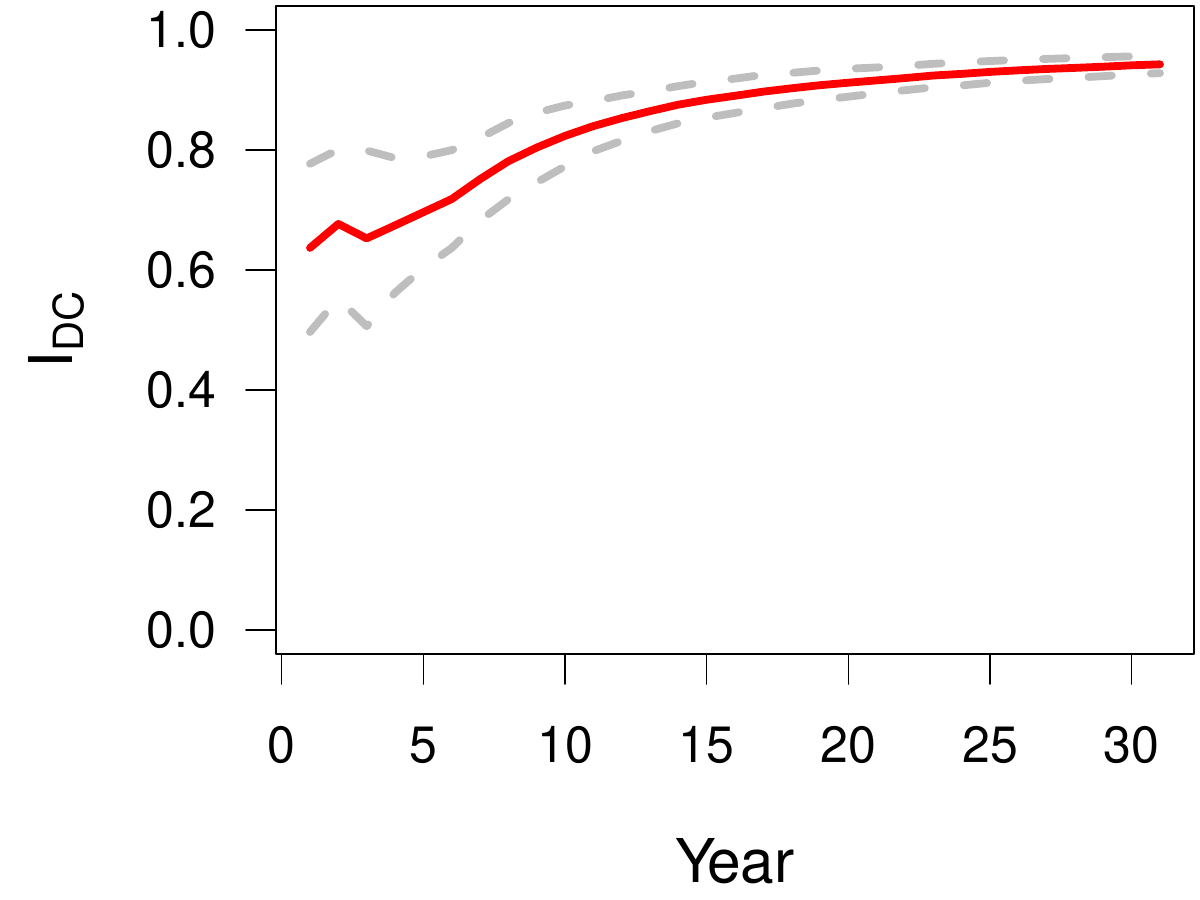} & \includegraphics[width = 0.3\textwidth]{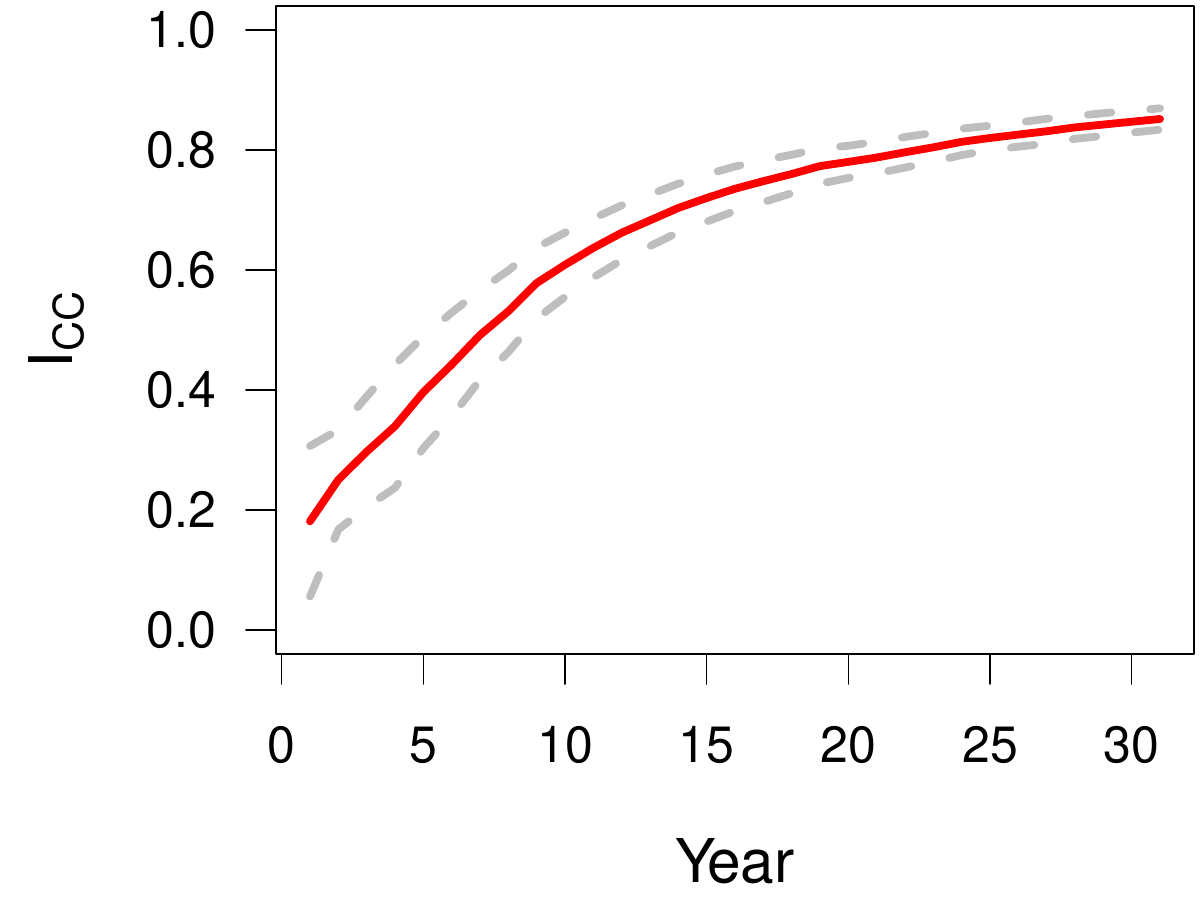}
     \end{tabular}
     \caption{Different indices of collaboration for $\lambda(t) = 0.5$ per month, and $F_n(k) = (0.05k + 0.005) \wedge 1$ for all $n \ge 1$.}
     \label{fig:simulation_2}
\end{figure}

\begin{figure}[!htbp]
     \centering
     \begin{tabular}{ccc}
         \includegraphics[width = 0.3\textwidth]{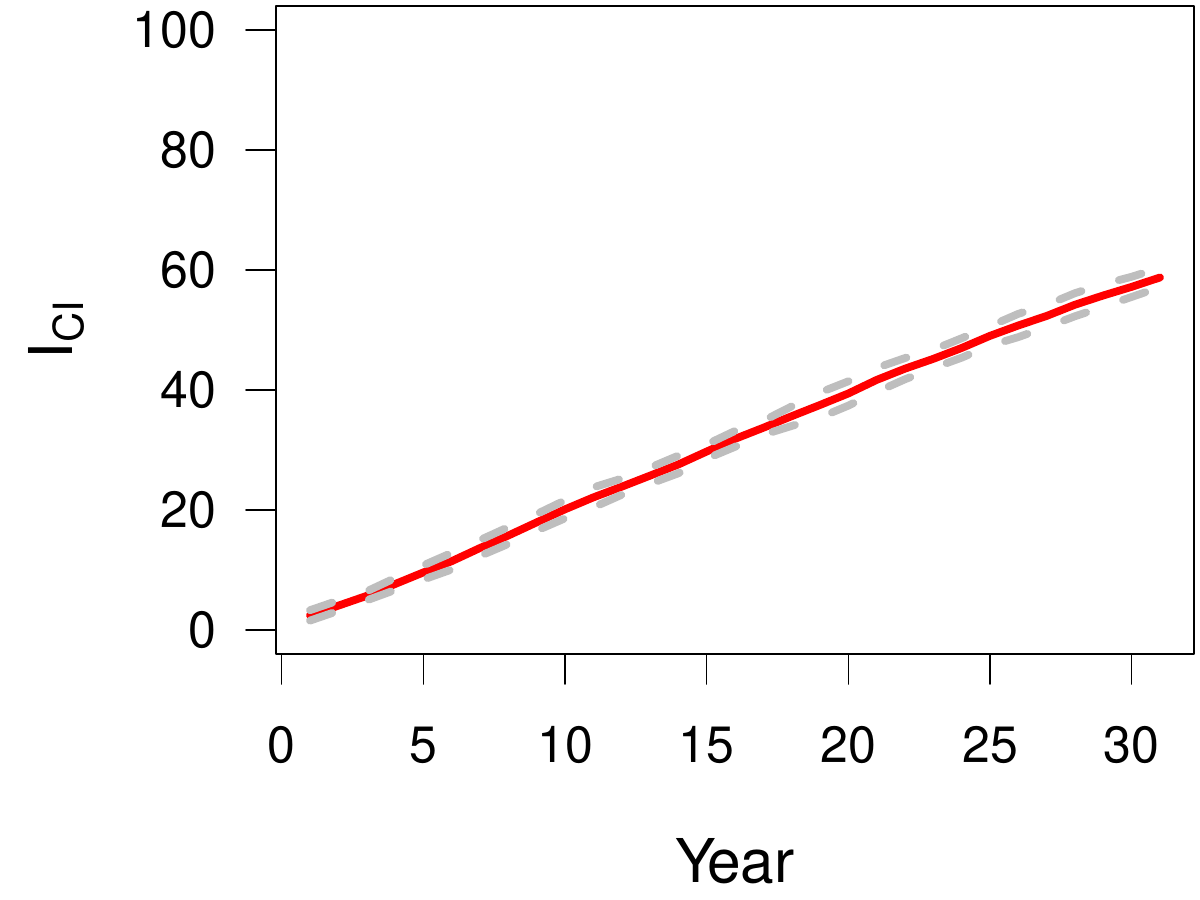} & \includegraphics[width = 0.3\textwidth]{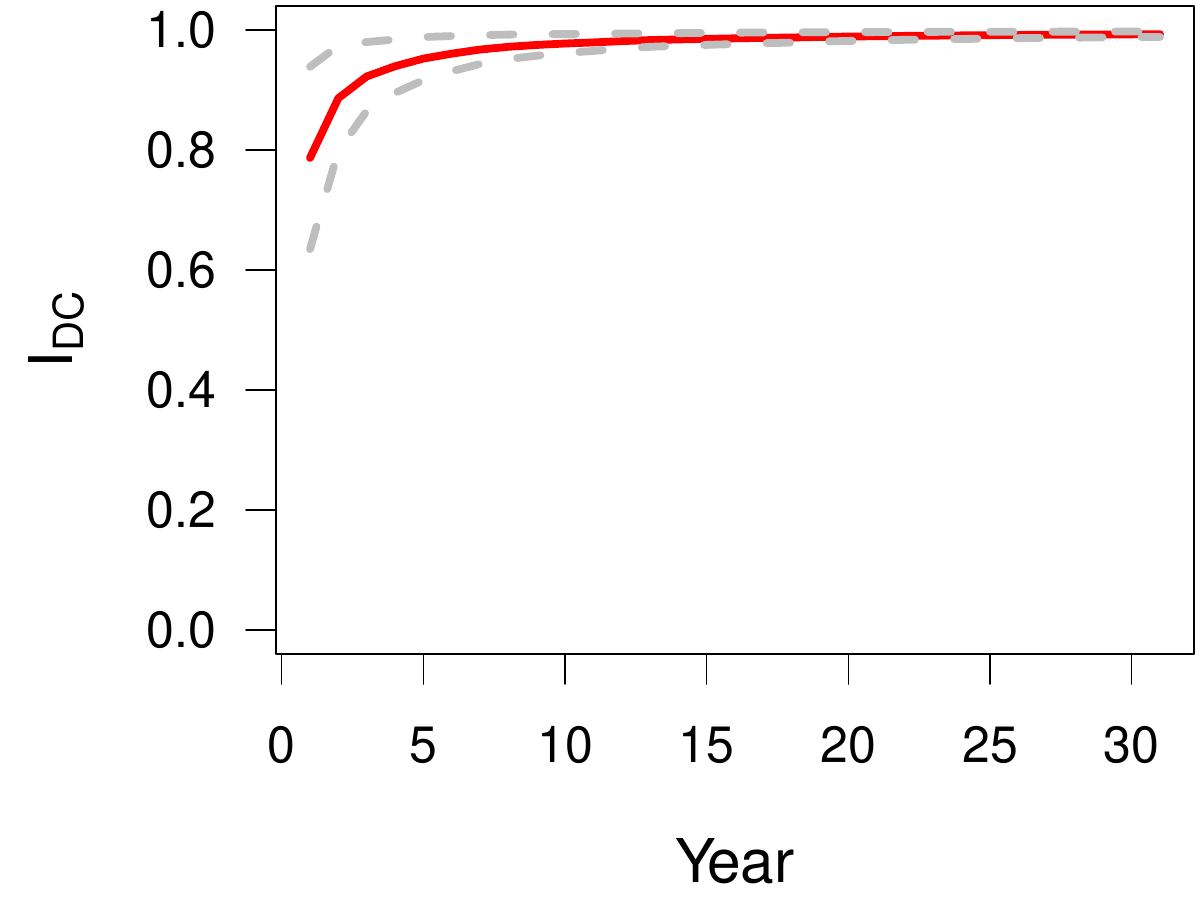} & \includegraphics[width = 0.3\textwidth]{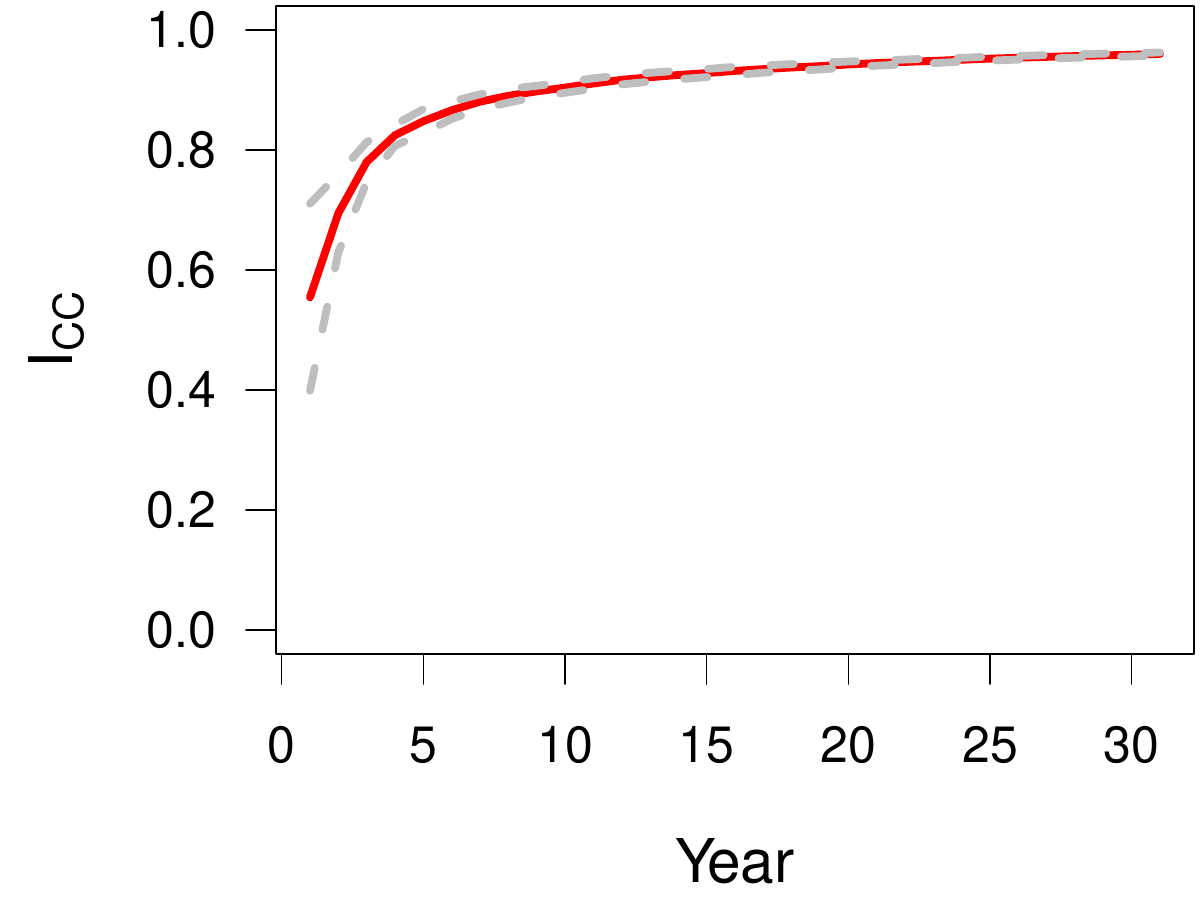}
     \end{tabular}
     \caption{Different indices of collaboration for $\lambda(t) = 0.5$ per month, and $F_n(k) = \frac{n}{180} \wedge 1$ for all $k \ge 0$.}
     \label{fig:simulation_3}
\end{figure}


\begin{figure}[!htbp]
     \centering
     \begin{tabular}{ccc}
     \includegraphics[width = 0.3\textwidth]{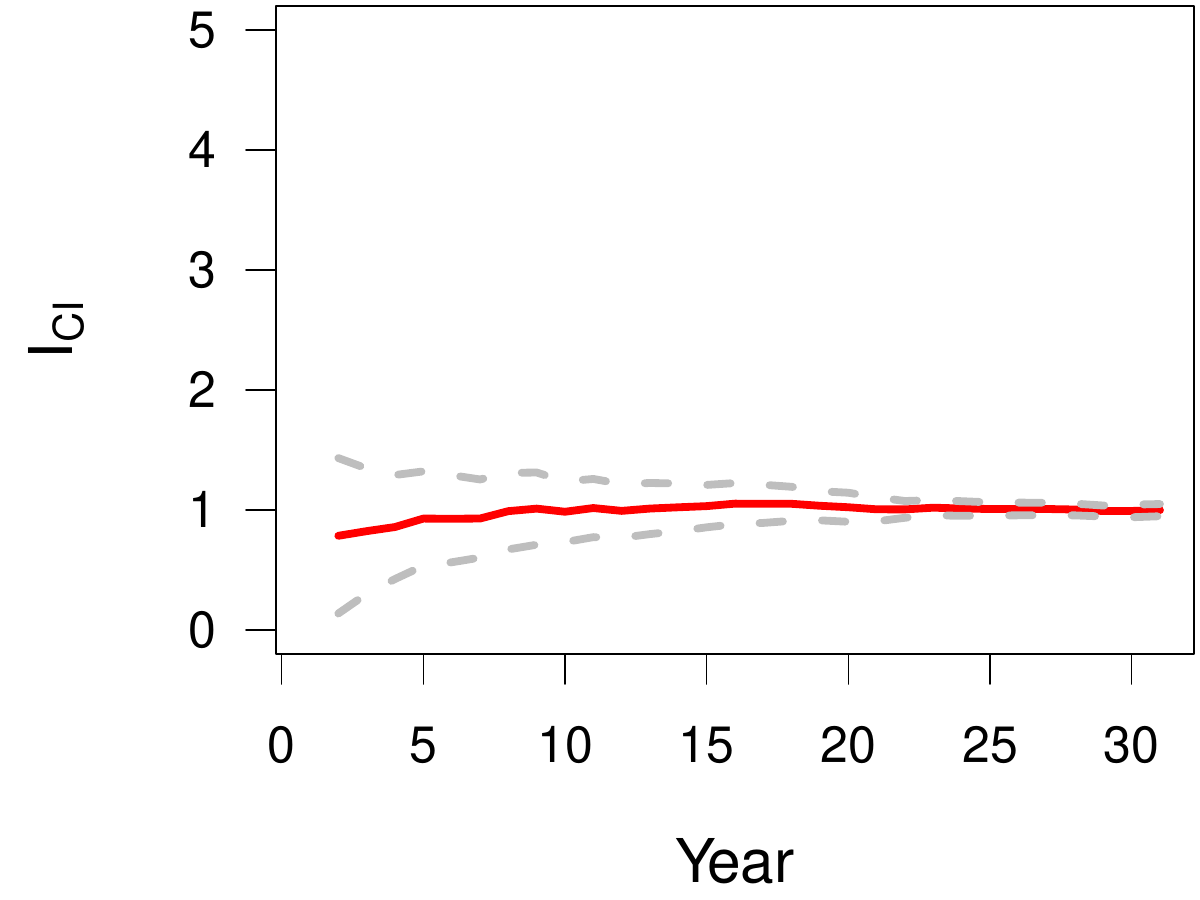} & \includegraphics[width = 0.3\textwidth]{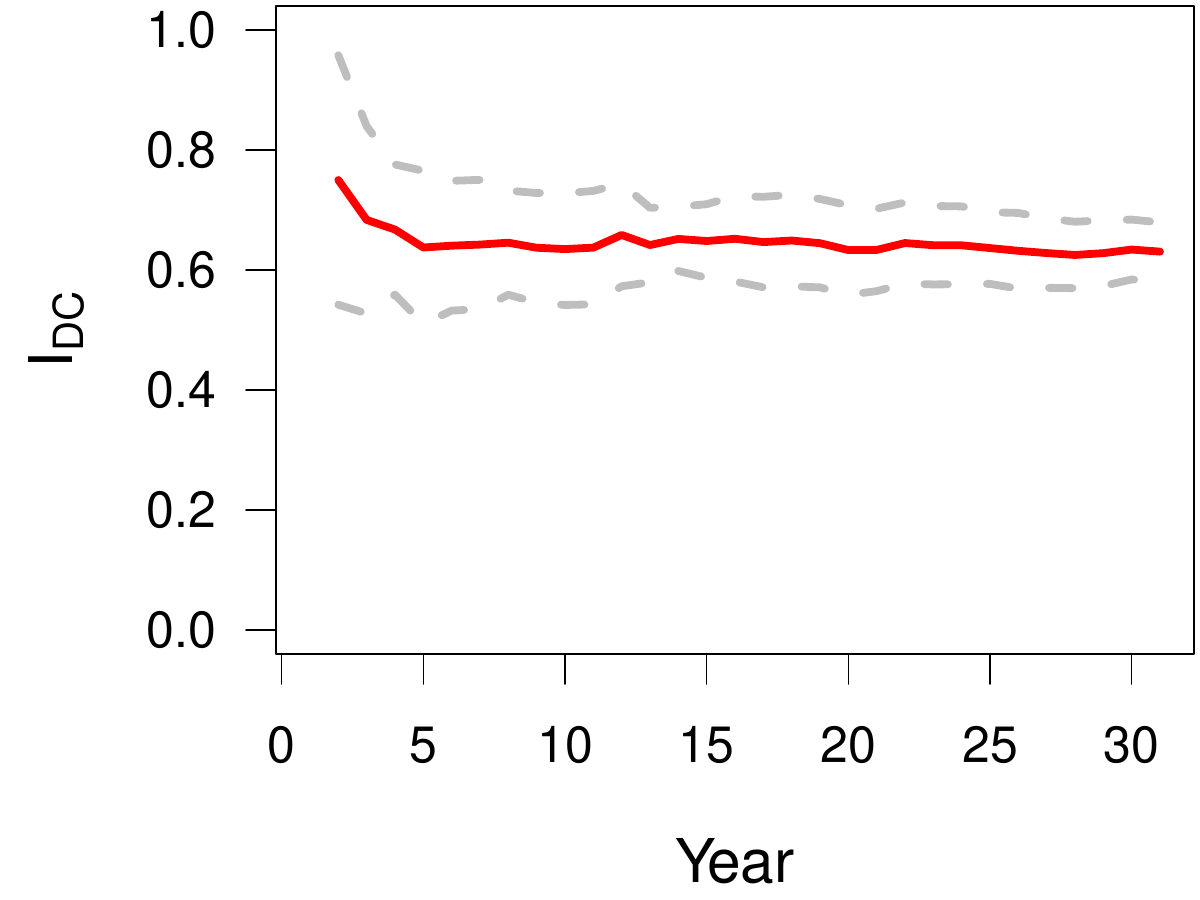} & \includegraphics[width = 0.3\textwidth]{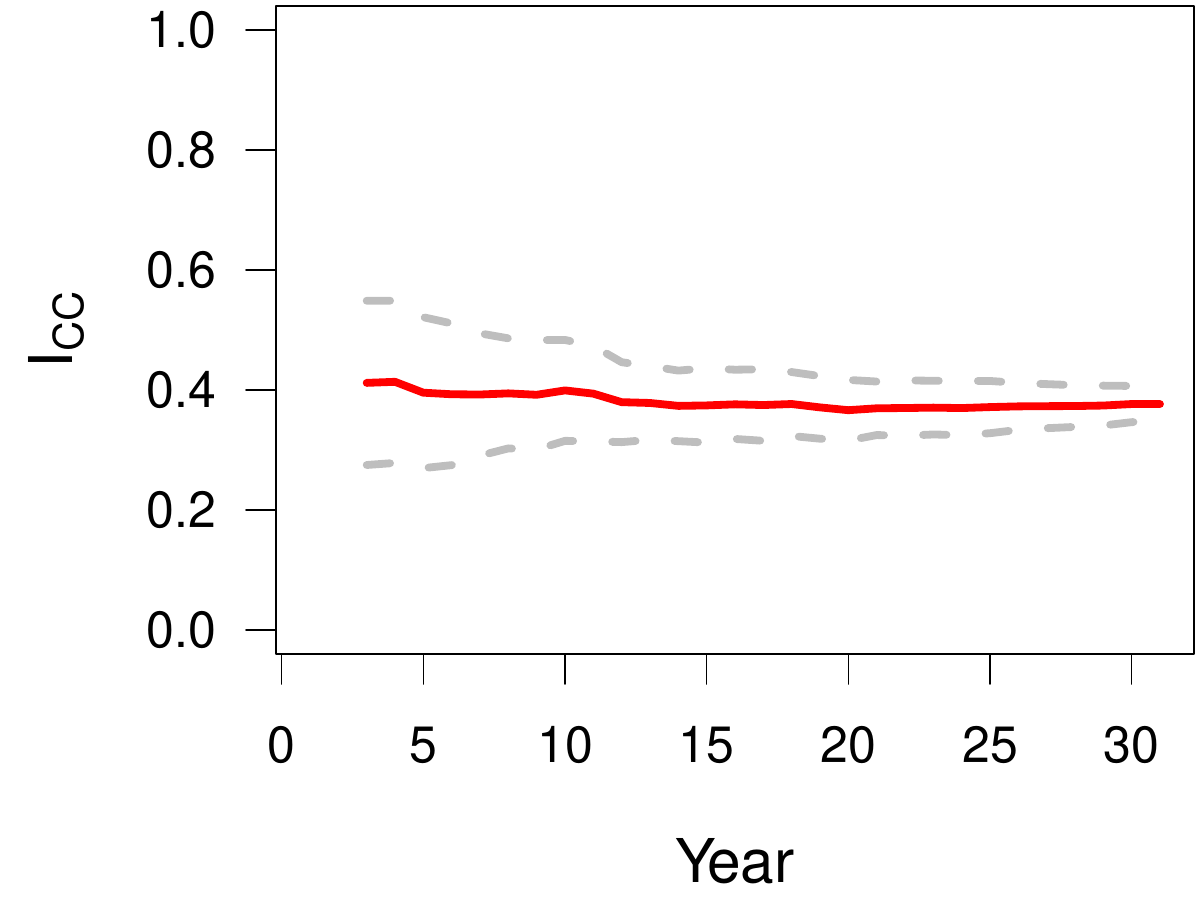} 
     \end{tabular}
     \caption{Different indices of collaboration for $\lambda(t) = \frac{1}{6}\bone_{[0, 100)}(t) + \frac{1}{3}\bone_{[100, 200)}(t) + \frac{1}{2}\bone_{[200,\infty)}(t)$ per month, and $F_n(k) = 0.01$ for all $n \ge 1, k \ge 0$.}
     \label{fig:suppl_simulation_1}
\end{figure}

\begin{figure}[!htbp]
     \centering
     \begin{tabular}{ccc}
     \includegraphics[width = 0.3\textwidth]{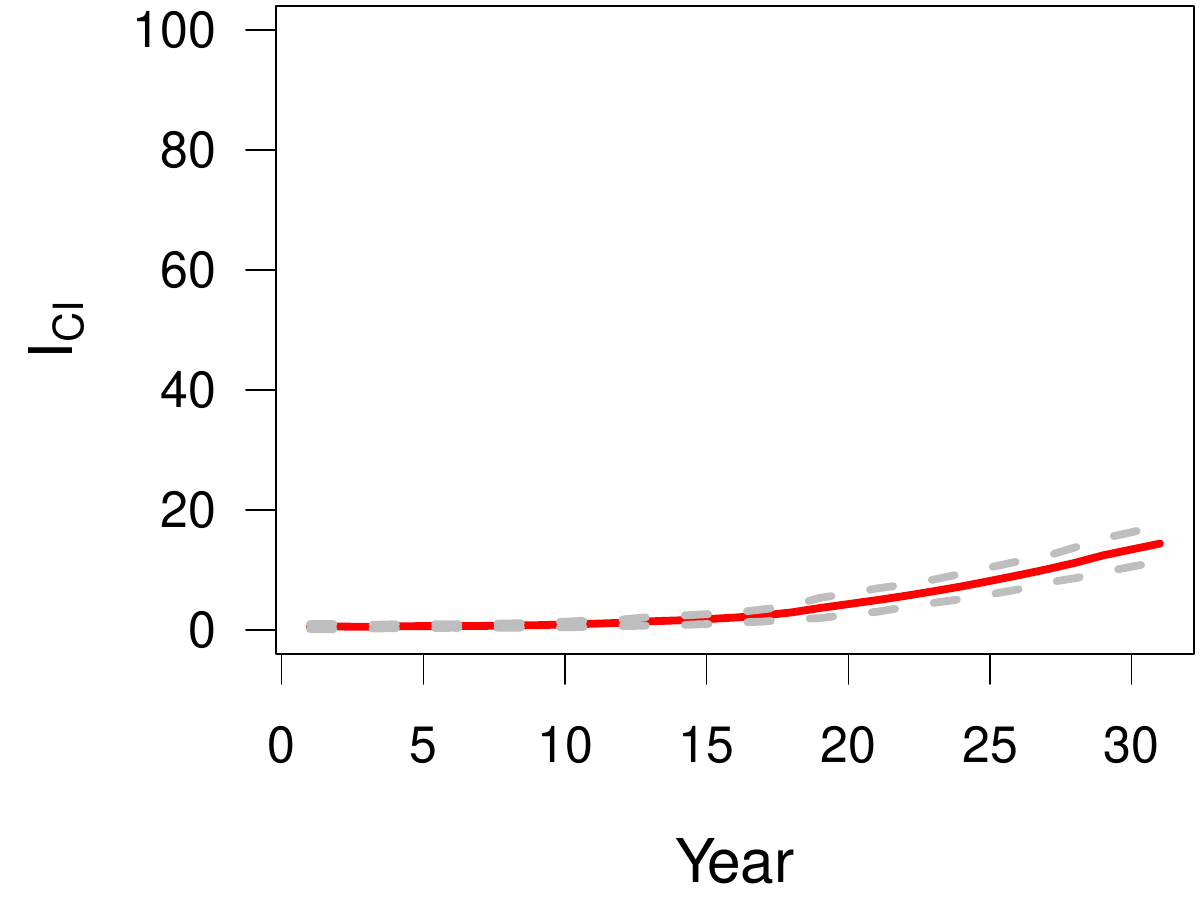} & \includegraphics[width = 0.3\textwidth]{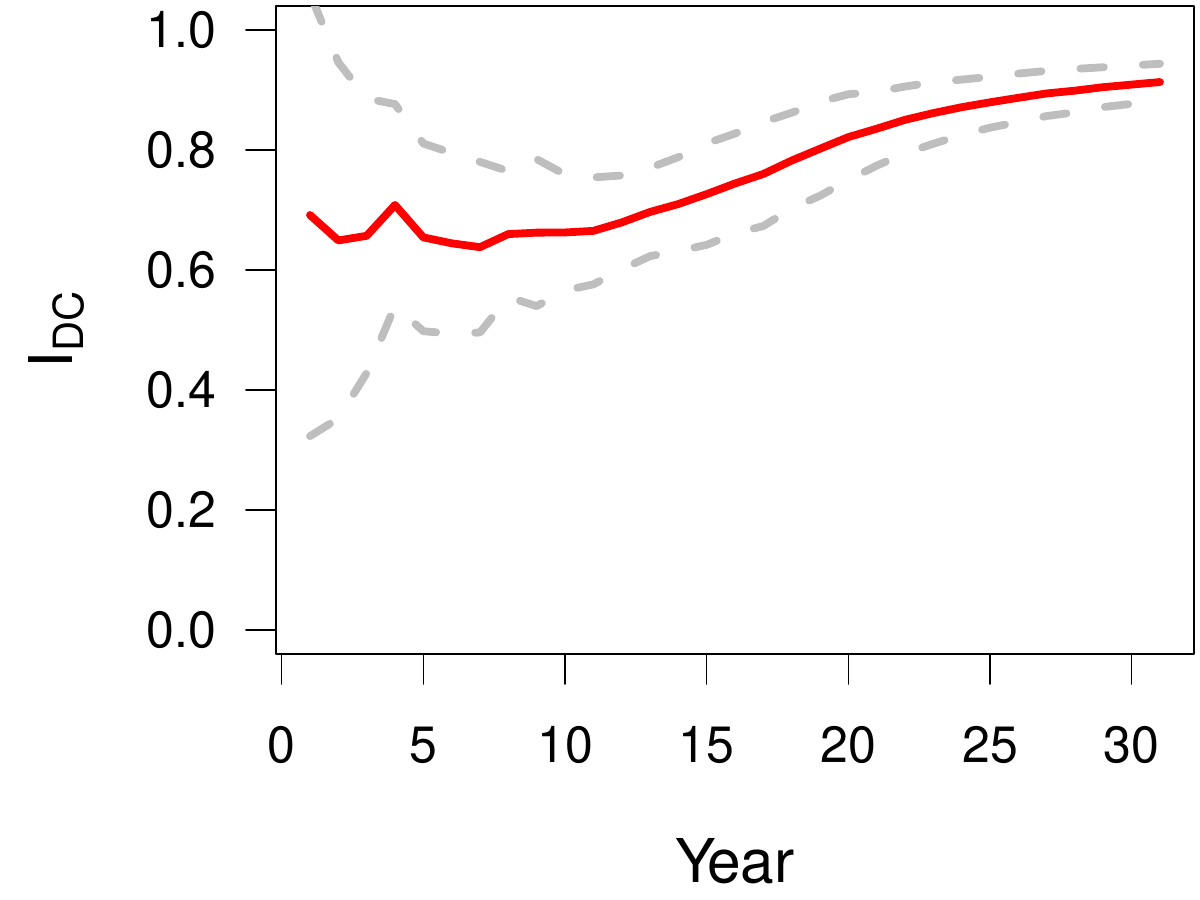} & \includegraphics[width = 0.3\textwidth]{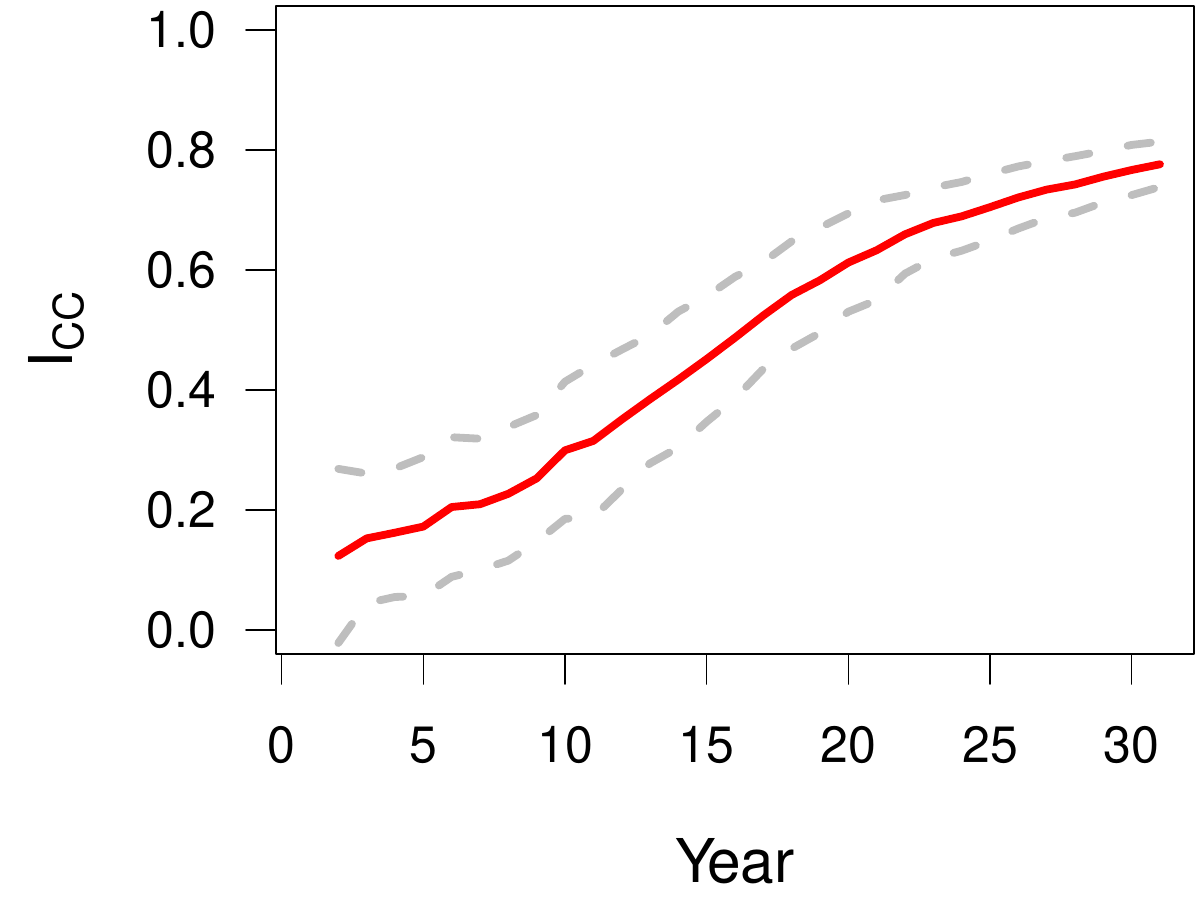}
     \end{tabular}
     \caption{Different indices of collaboration for $\lambda(t) = \frac{1}{6}\bone_{[0, 100)}(t) + \frac{1}{3}\bone_{[100, 200)}(t) + \frac{1}{2}\bone_{[200, \infty)}(t)$ per month, and $F_n(k) = (0.05k + 0.005) \wedge 1$ for all $n \ge 1$.}
     \label{fig:suppl_simulation_2}
\end{figure}

\begin{figure}[!htbp]
     \centering
     \begin{tabular}{ccc}
     \includegraphics[width = 0.3\textwidth]{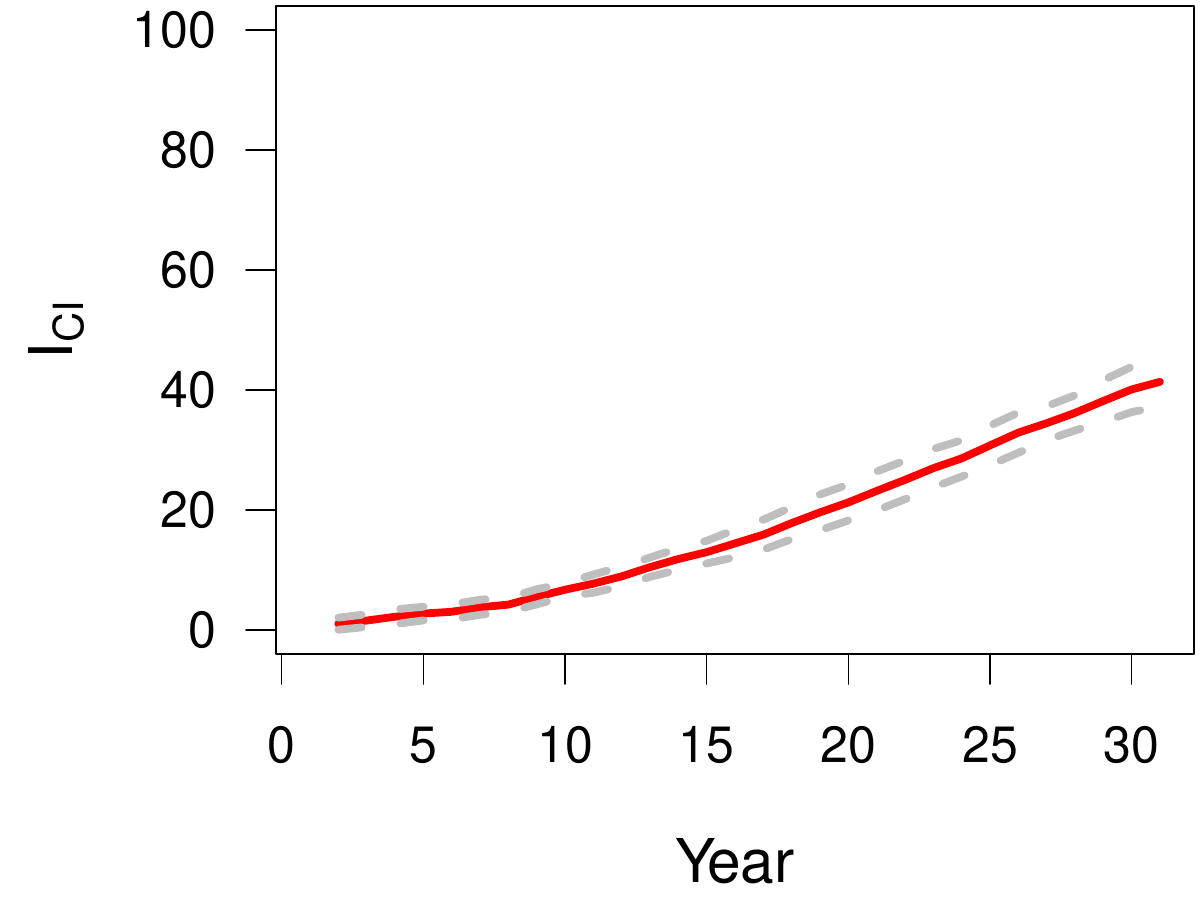} & \includegraphics[width = 0.3\textwidth]{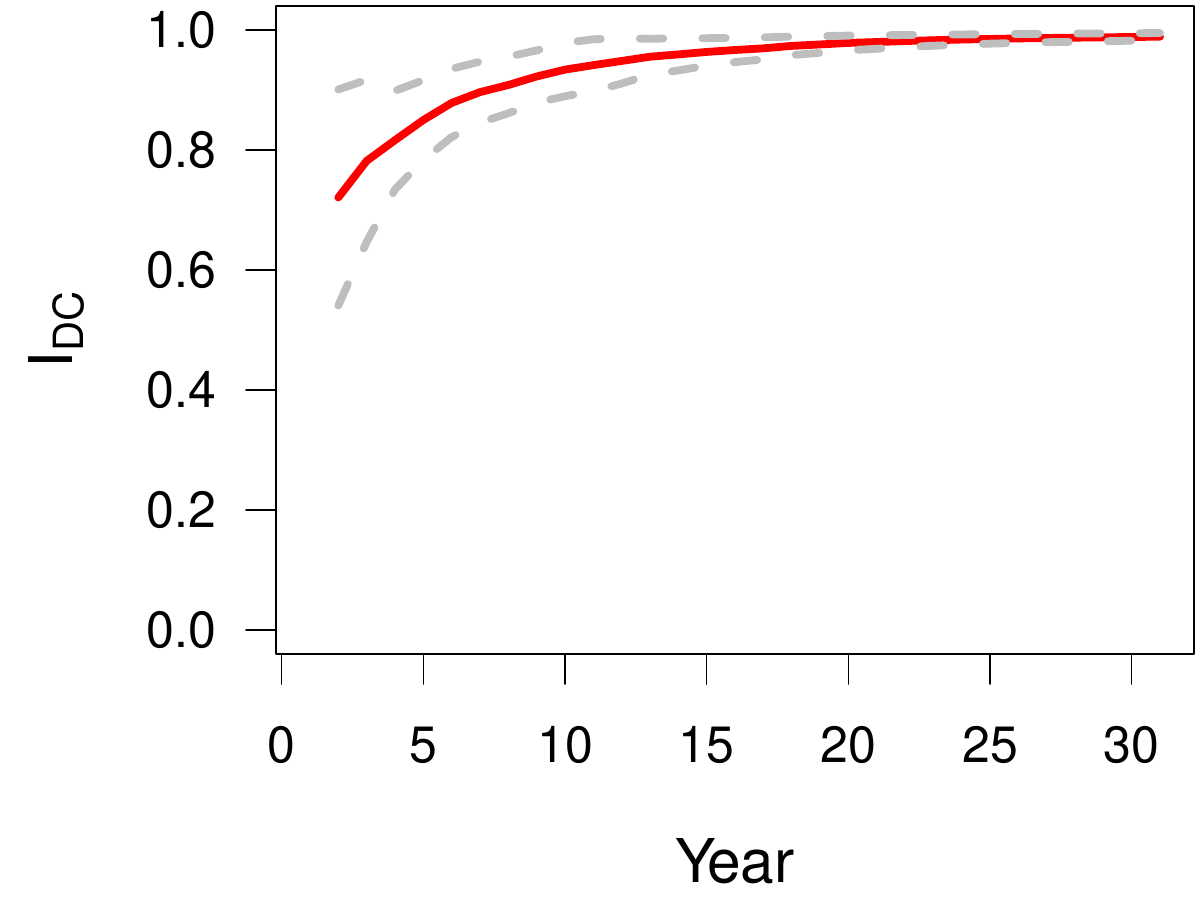} & \includegraphics[width = 0.3\textwidth]{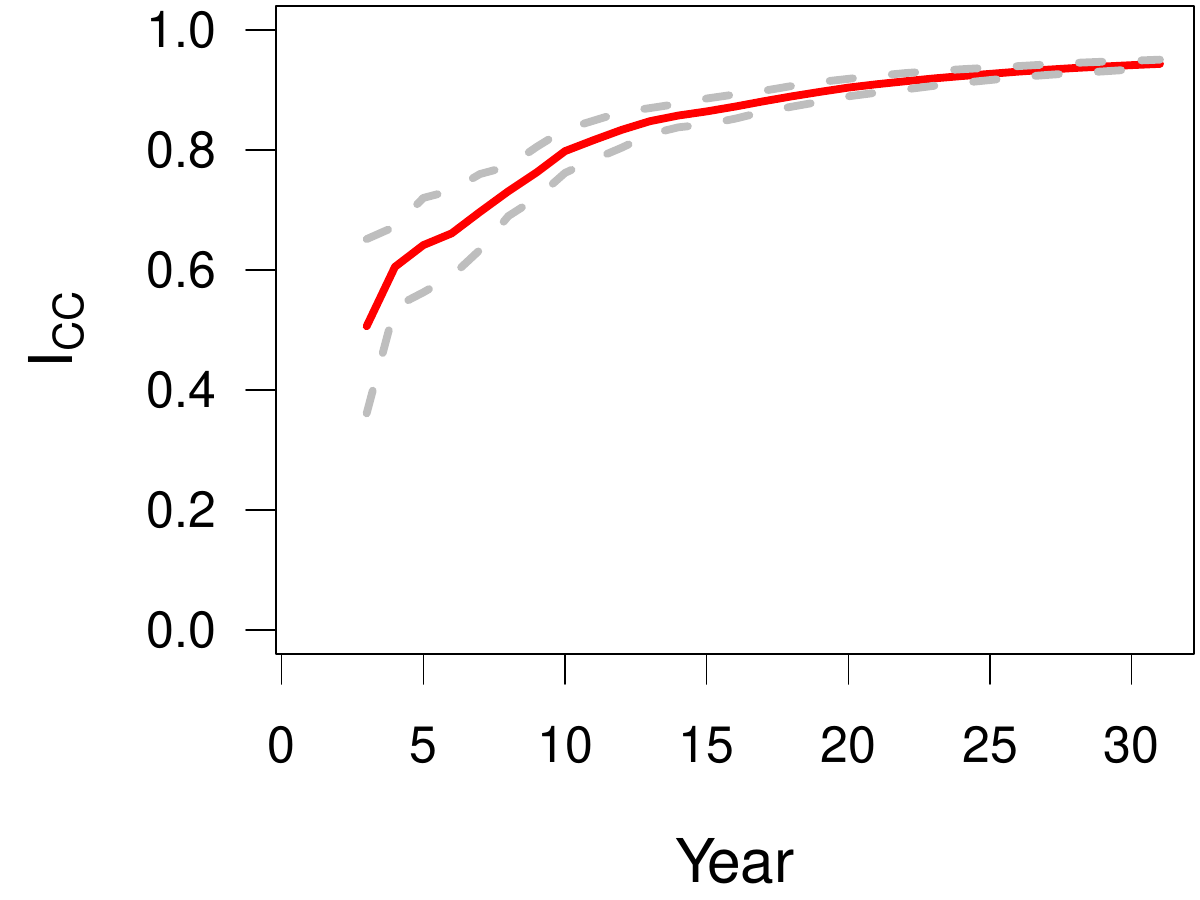}
     \end{tabular}
     \caption{Different indices of collaboration for $\lambda(t) = \frac{1}{6}\bone_{[0, 100)}(t) + \frac{1}{3}\bone_{[100, 200)}(t) + \frac{1}{2}\bone_{[200, \infty)}(t)$ per month, and $F_n(k) = \frac{n}{180} \wedge 1$ for all $k \ge 0$.}
     \label{fig:suppl_simulation_3}
\end{figure}

\begin{figure}[!htbp]
     \centering
     \begin{tabular}{ccc}
     \includegraphics[width = 0.3\textwidth]{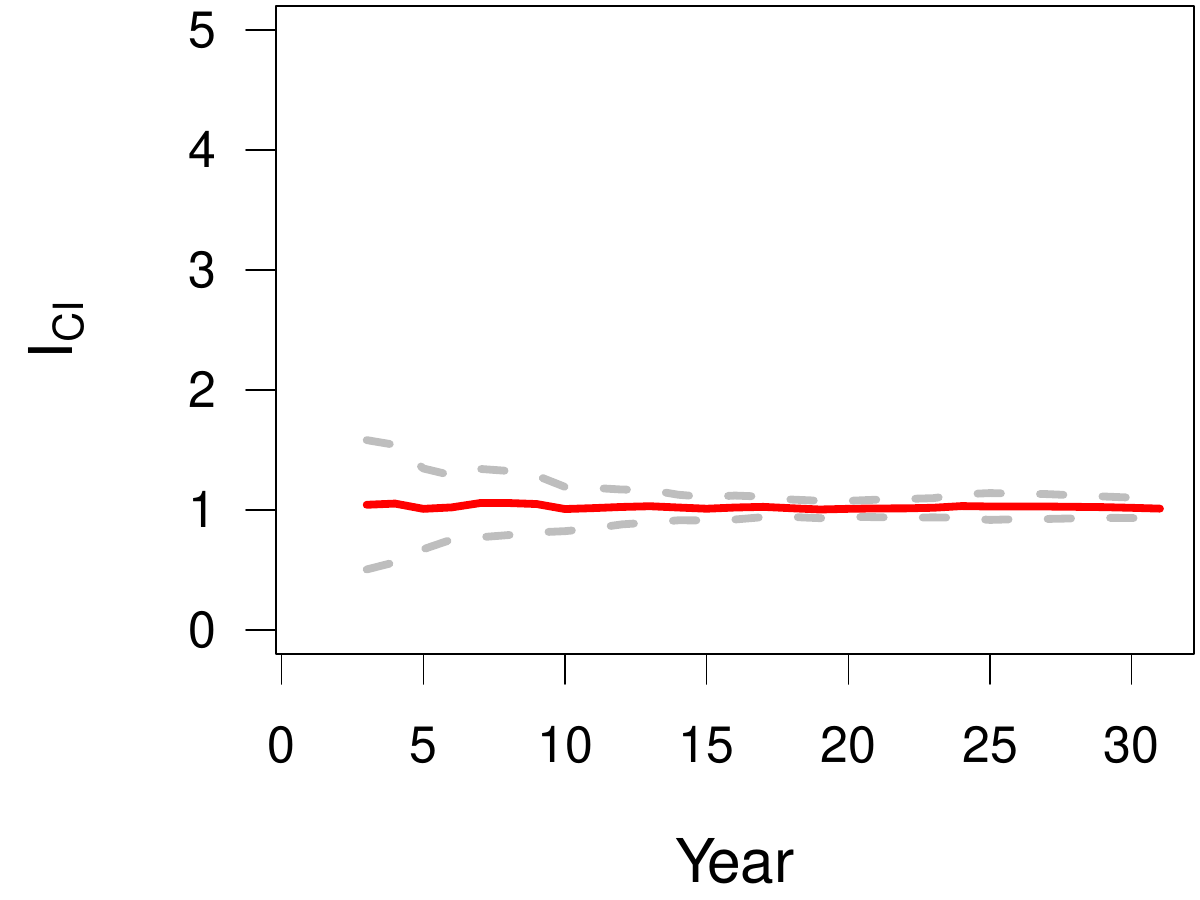} & \includegraphics[width = 0.3\textwidth]{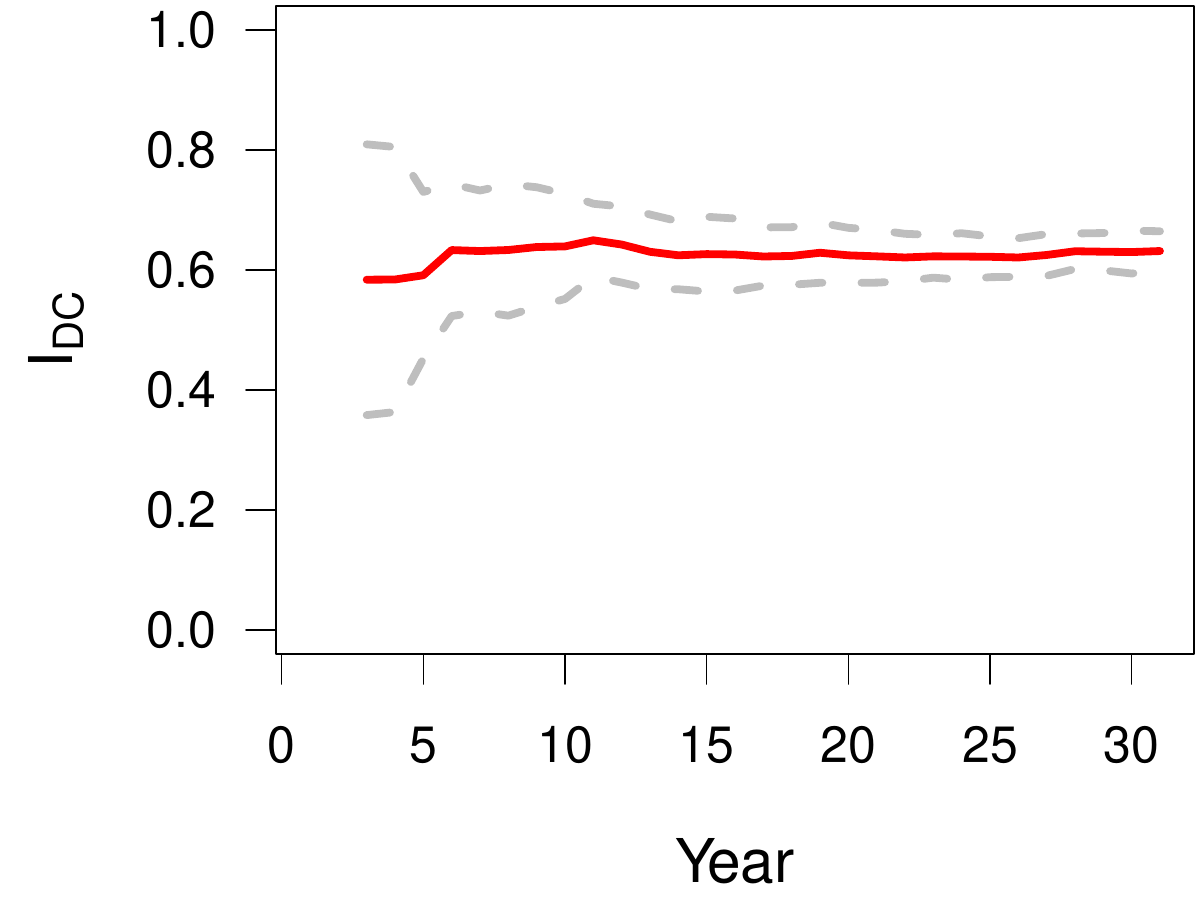} & \includegraphics[width = 0.3\textwidth]{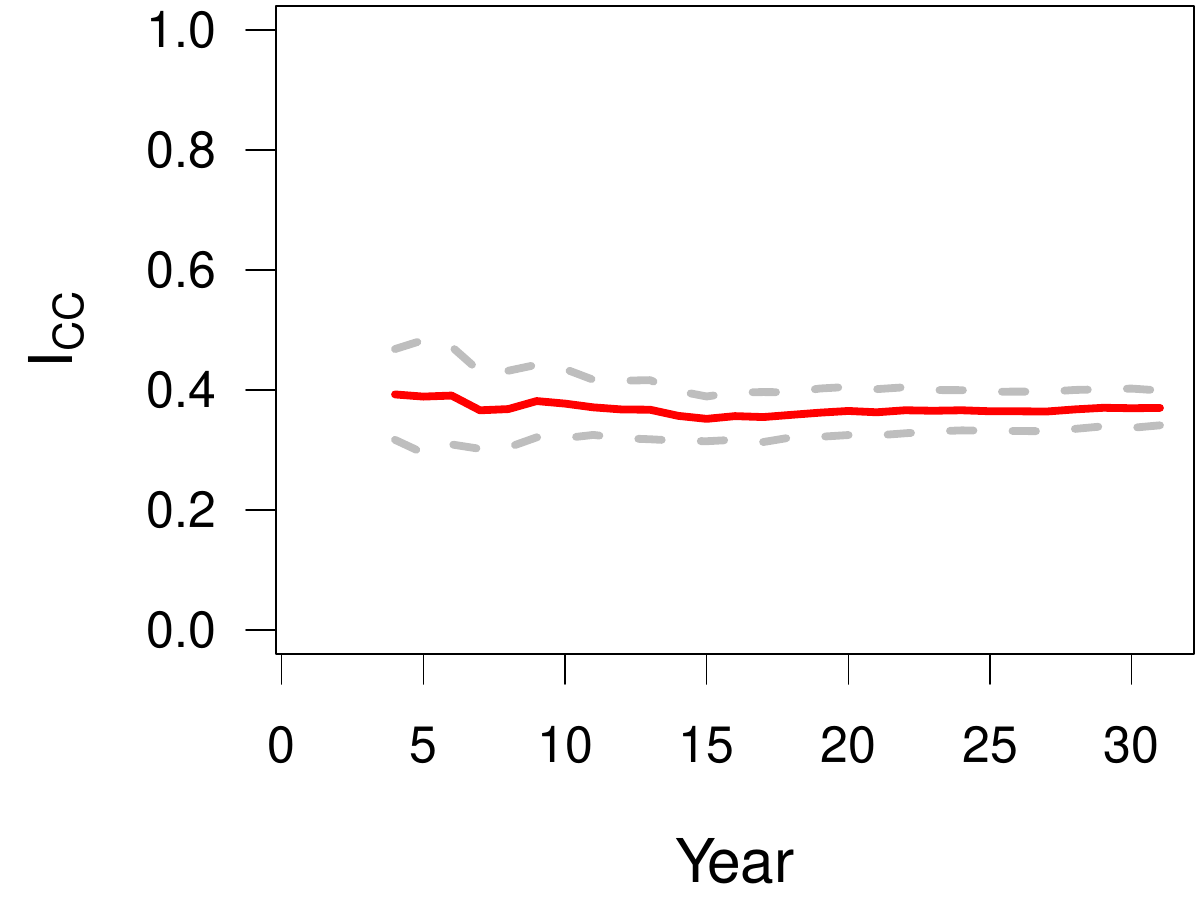} 
     \end{tabular}
     \caption{Different indices of collaboration for $\lambda(t) = \frac{t}{200}\bone_{[0, 100)}(t) + \frac{t}{400}\bone_{[100, 200)}(t) + \big(\frac{t}{720} \wedge 1\big)\bone_{[200, \infty)}(t)$ per month, and $F_n(k) = 0.01$ for all $n \ge 1, k \ge 0$.}
     \label{fig:suppl_simulation_4}
\end{figure}

\begin{figure}[!htbp]
     \centering
     \begin{tabular}{ccc}
     \includegraphics[width = 0.3\textwidth]{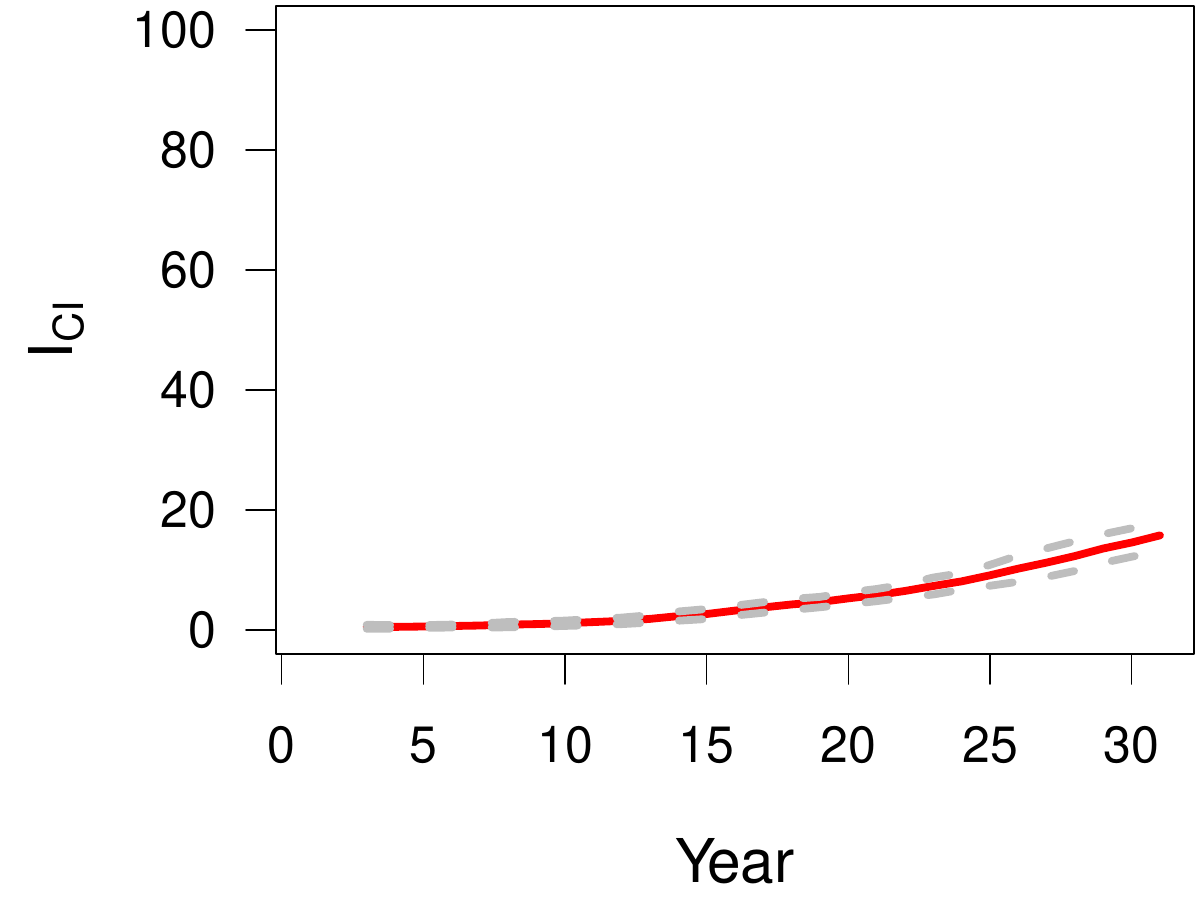} & \includegraphics[width = 0.3\textwidth]{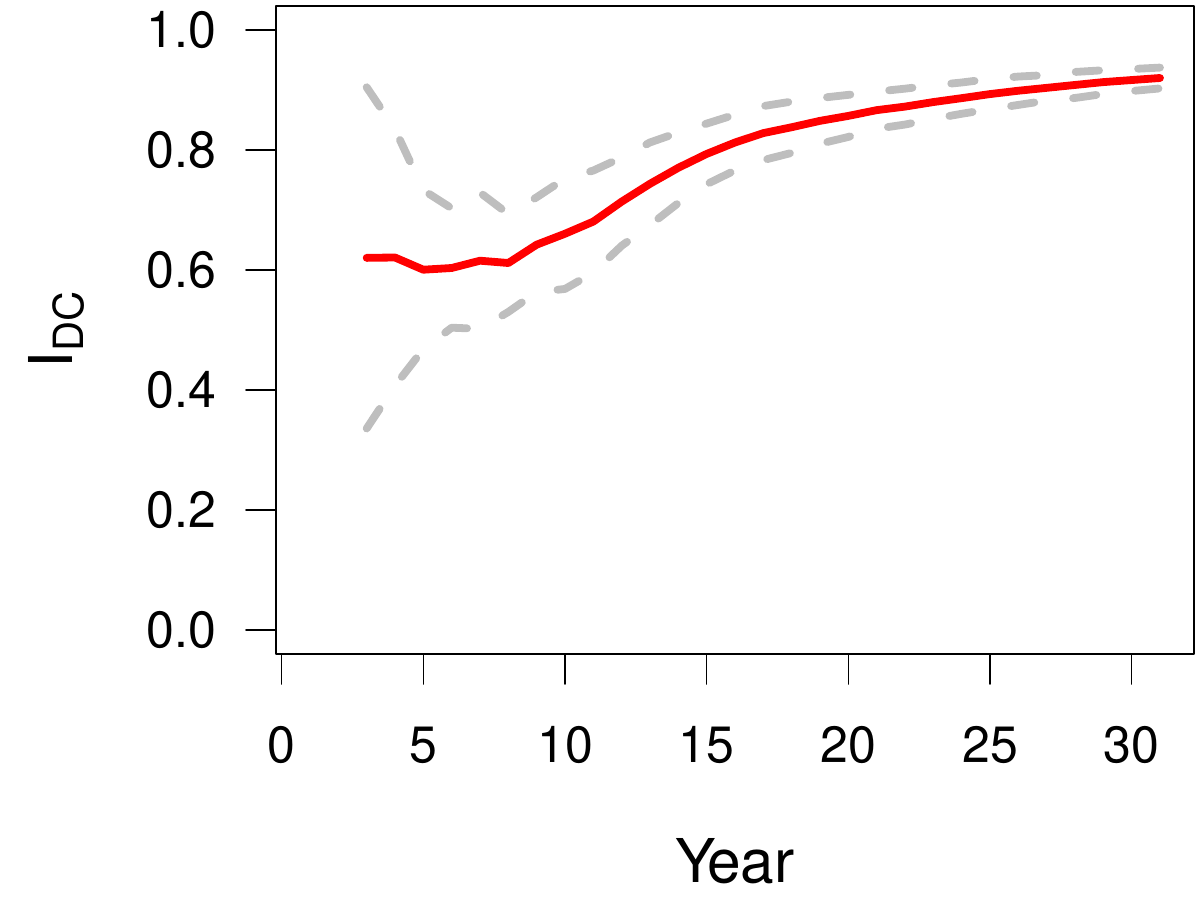} & \includegraphics[width = 0.3\textwidth]{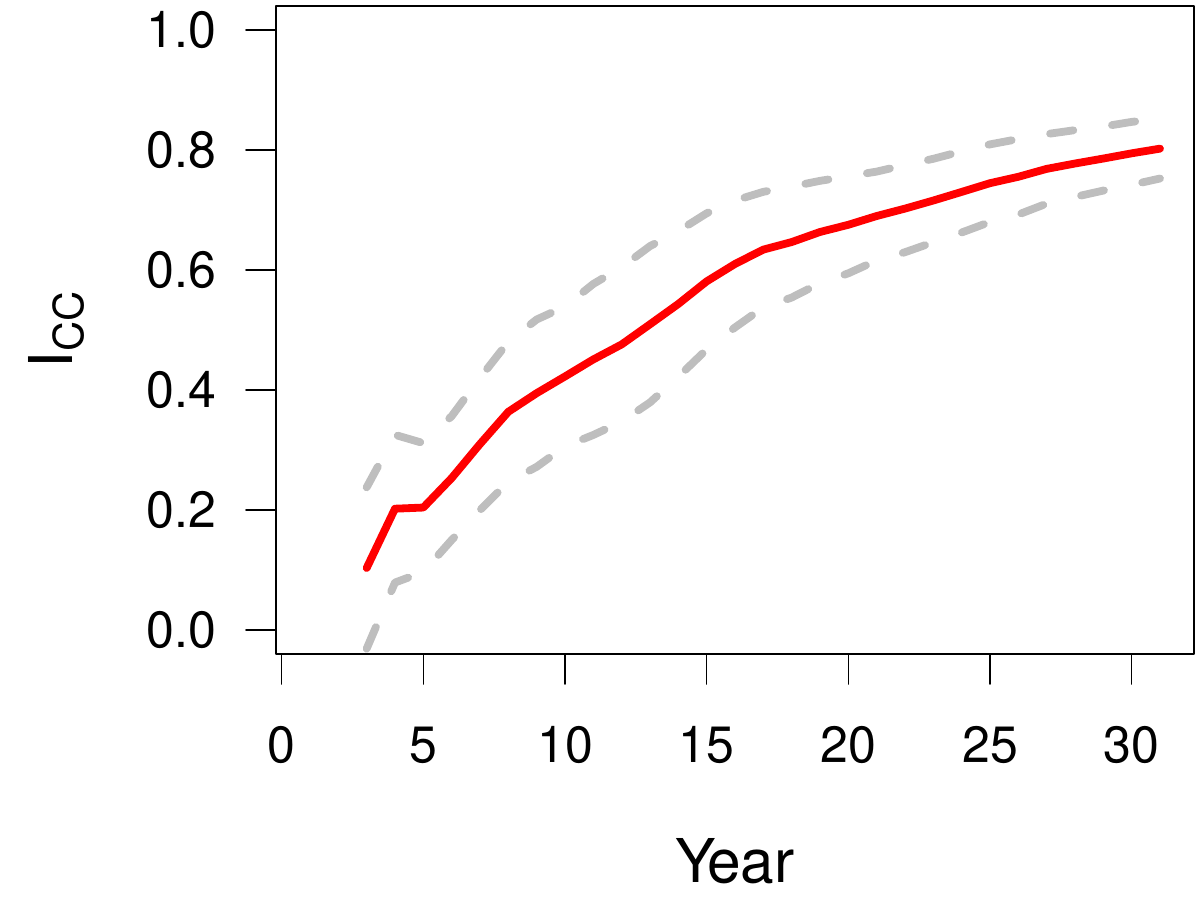}
     \end{tabular}
     \caption{Different indices of collaboration for $\lambda(t) = \frac{t}{200}\bone_{[0, 100)}(t) + \frac{t}{400}\bone_{[100, 200)}(t) + \big(\frac{t}{720} \wedge 1\big)\bone_{[200, \infty)}(t)$ per month, and $F_n(k) = (0.05k + 0.005) \wedge 1$ for all $n \ge 1$.}
     \label{fig:suppl_simulation_5}
\end{figure}

\begin{figure}[!htbp]
     \centering
     \begin{tabular}{ccc}
     \includegraphics[width = 0.3\textwidth]{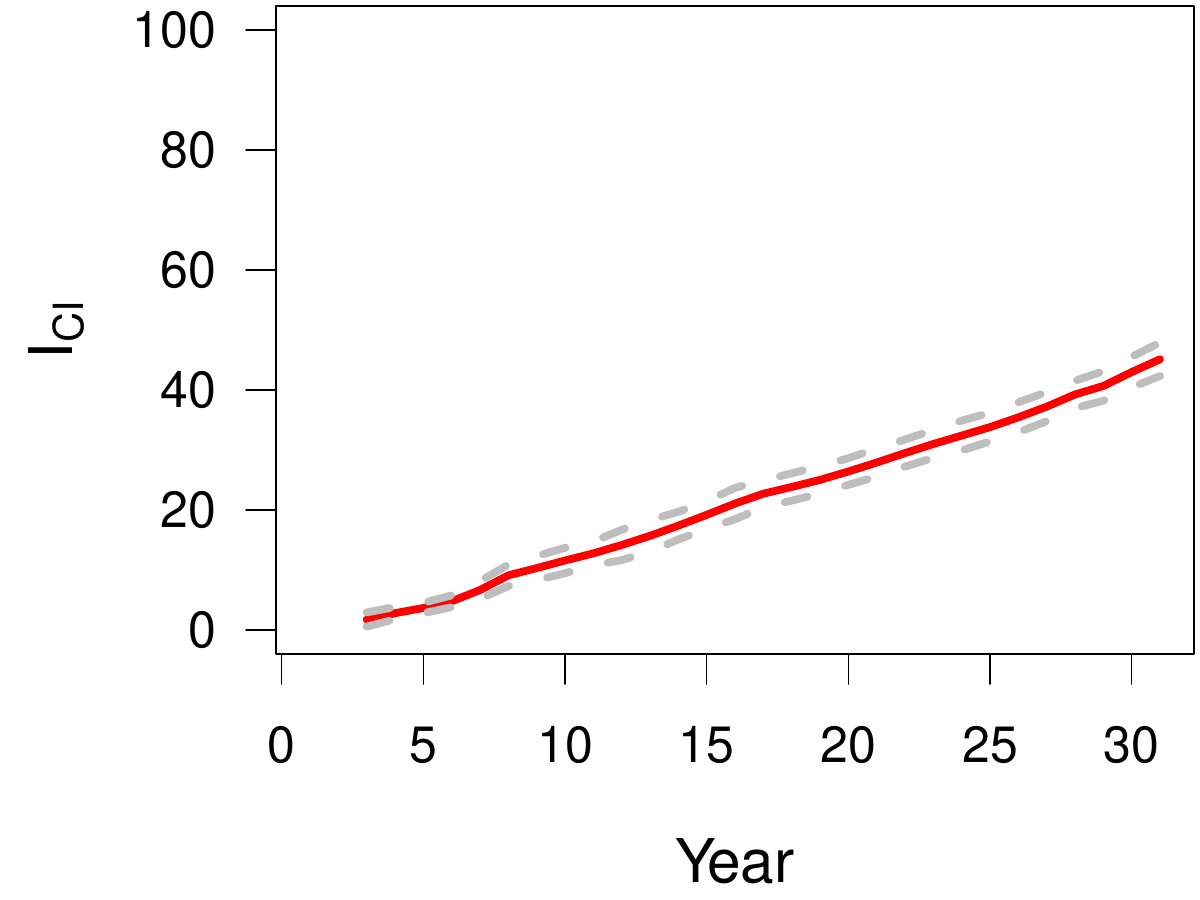} & \includegraphics[width = 0.3\textwidth]{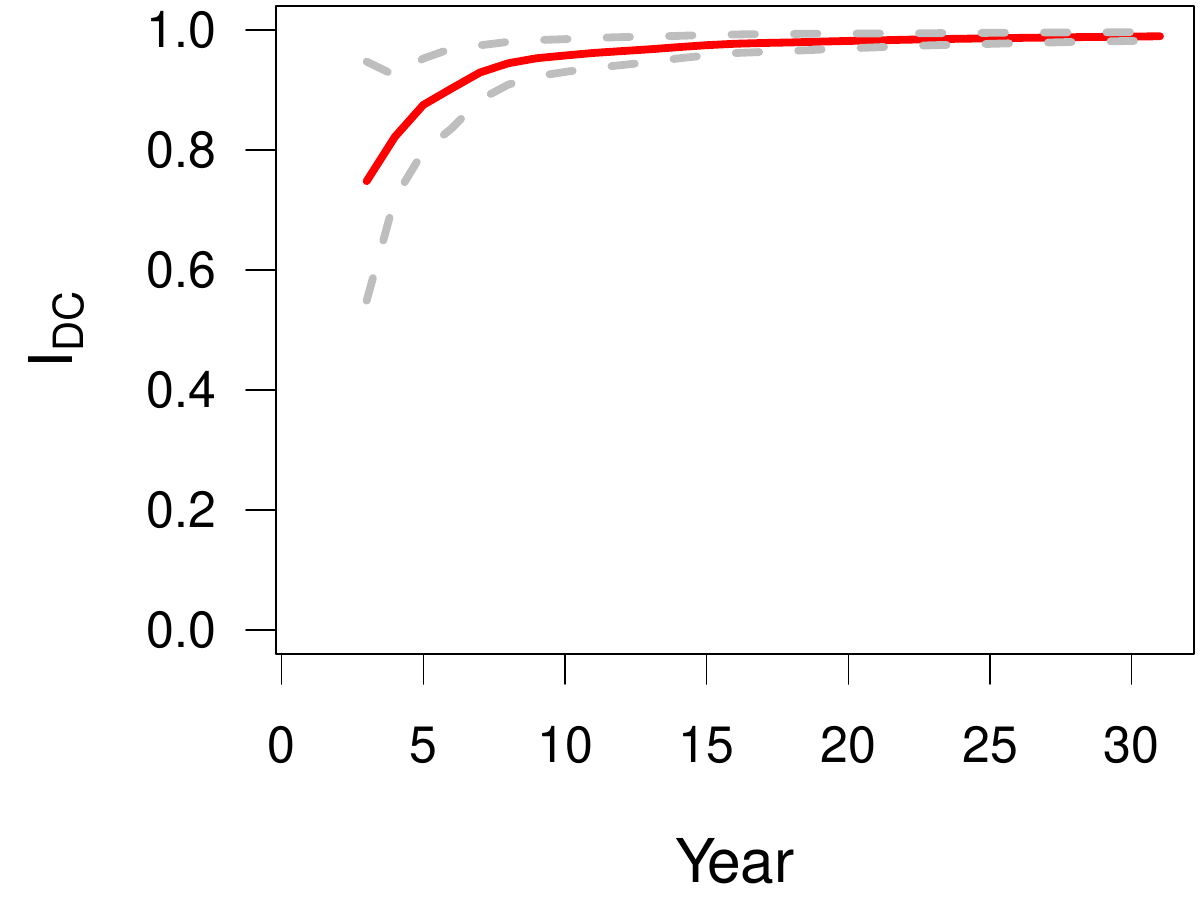} & \includegraphics[width = 0.3\textwidth]{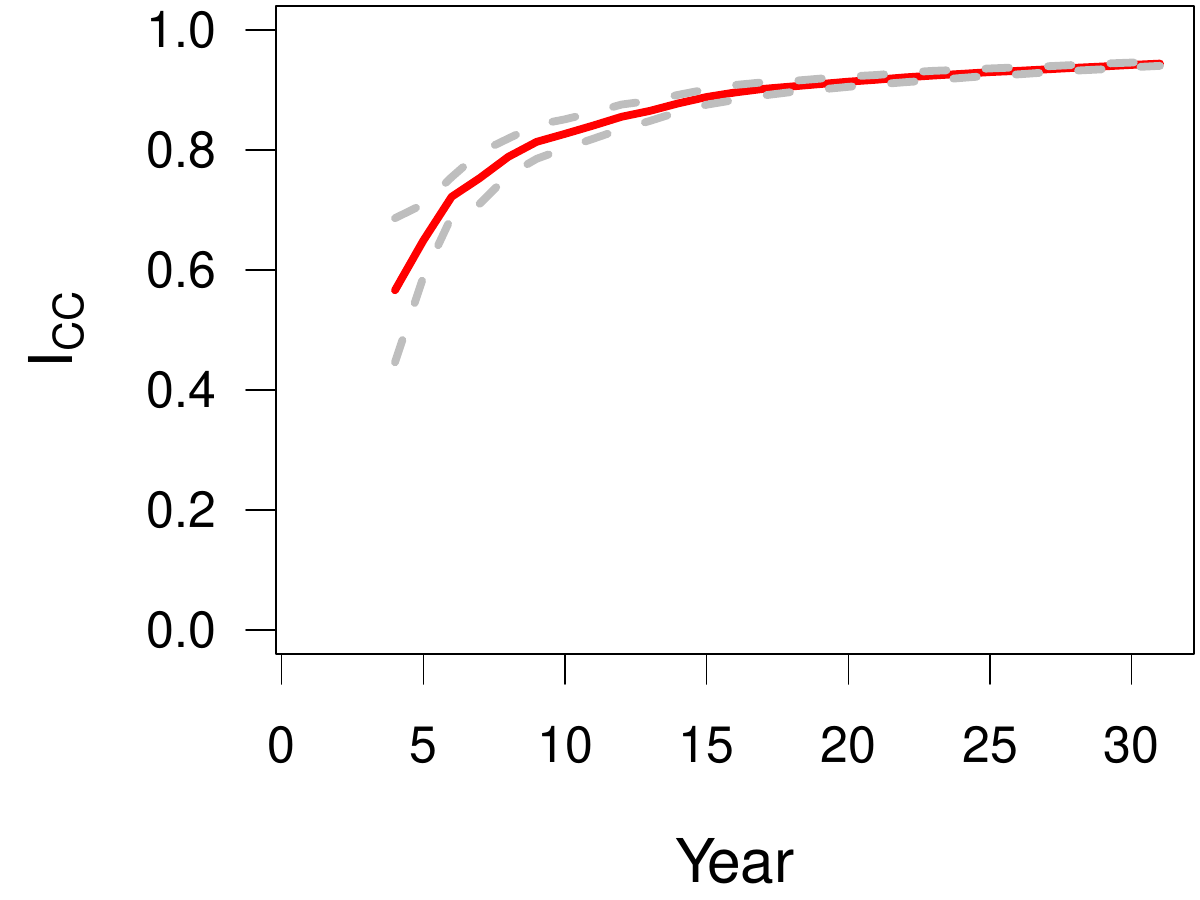}
     \end{tabular}
     \caption{Different indices of collaboration for $\lambda(t) = \frac{t}{200}\bone_{[0, 100)}(t) + \frac{t}{400}\bone_{[100, 200)}(t) + \big(\frac{t}{720} \wedge \frac{1}{2}\big)\bone_{[200, \infty)}(t)$ per month, and $F_n(k) = \frac{n}{180} \wedge 1$ for all $k \ge 0$.}
     \label{fig:suppl_simulation_6}
\end{figure}


\section{Collaboration dynamics in Computer Science, Mathematics and Physics}\label{sec:real-data}
We analyse the collaboration dynamics in several disciplines using data from the arXiv repository \cite{clement2019arxiv}. For each paper submitted to arXiv during the period 1991-2022, the dataset has three variables, namely, (i) \textit{id}: unique identifier of the paper which also contains the time (month and year) of paper submission; (ii) \textit{categories}: the discipline and/or the sub-discipline of the paper; (iii) \textit{authors}: name of the authors of the paper. Therefore, the dataset has all the necessary information for computing the various indices of collaboration as a function of time.

\textbf{Indices of collaboration.} Since we are assuming a mean-field model, we consider all authors to be equivalent to any one individual author. The various indices of collaboration are computed by considering all papers written during a time-interval into account (c.f. the discussion in Section~\ref{sec:indices}). 

In Figure~\ref{fig:CS_indices_all_authors}, we plot the three indices $I_\CI$, $I_\DC$ and $I_\CC$ per year. Further, we also consider the mean-field collaboration dynamics of the top $100$ most productive Computer Science (CS) authors (in terms of number of papers written) in order to understand the collaboration behaviour of extremely productive authors. In Figure~\ref{fig:CS_indices_top_100_authors}, yearly indices computed based on the top 100 most productive authors are plotted. Notice that for a typical top author in CS, the indices indicate higher degree of collaboration than a typical author of the entire discipline. In fact, the $I_{\CC}$ index for them has reached almost 1 by 2020, i.e. they have mostly abandoned writing single-author papers!

We also construct similar plots for Mathematics and Physics --- the various indices of collaboration are computed by taking either all papers written during a time-interval into account (Figures~\ref{fig:math_indices_all_authors} and \ref{fig:phy_indices_all_authors}), or only the papers written by the top 100 most productive authors (Figures~\ref{fig:math_indices_top_100_authors} and \ref{fig:phy_indices_top_100_authors}). As in the case of CS, we observe that for a typical top author, the indices indicate higher degree of collaboration than a typical author in both Mathematics and Physics. Also, the overall tendency of collaboration in Mathematics is much lower than that of CS and Physics.

\textbf{No of co-authors per paper.} In Figures~\ref{fig:csWRITE}-(a), \ref{fig:mathWRITE}-(a) and \ref{fig:phyWRITE}-(a), we plot the average number of co-authors in the $k$-th paper versus $k$, computed based on the top $100$ most productive authors in CS, Mathematics and Physics, respectively. The average number of co-authors in the $k$-th paper seems to increase as $k$ increases across all three disciplines. However, the behaviour of the rate of increment is markedly different --- it gradually decreases in the case of CS, remains mostly a constant in the case of Mathematics and exhibits a rather interesting piecewise constant pattern in the case of Physics. The behaviour seen in the case of CS can qualitatively be explained via Lemma~\ref{lem:linear_F} (see, e.g., Figure~\ref{fig:rec-sol}-(a)). We also notice that for Physics, as $k$ increases, the rate of increase (as a function of $k$) in the number of co-authors in the $k$-th paper is much higher than the same for CS or Mathematics (likely due to the effect of very large collaborative projects in Physics).

\textbf{Temporal behaviour of correlations.} We choose $1000$ authors randomly from the set of all CS authors and plot in Figure~\ref{fig:csWRITE}-(b) the empirical correlation between $X_1[t, t+\delta]$ and  $X_k[t, t+\delta]$ based on the observed values of $(X_1[t, t+\delta], X_k[t, t+\delta])$ for these authors for $k = 2, 3, 4, 5$ with $\delta = 1$ year. Similar plots for Mathematics and Physics appear in Figures~\ref{fig:mathWRITE}-(b) and \ref{fig:phyWRITE}-(b), respectively. These plots suggest the presence of non-trivial correlations.

\textbf{Estimation of co-authorship probability parameters.} Using the methodology described in Section~\ref{sec:est}, we estimate the co-authorship probability parameters $(F_n(k))_{n\ge1}$ for $k = 0, 1, 2, 3$. Without loss of generality, let $[L_1]$ be the (random) set of authors who write a paper with $a_0$ \emph{at some point of time} and let $[L_2] := [L]\setminus[L_1]$. Then we may write
\begin{align*}
    \hat{F}_n(k) = \frac{\sum_{i \in [L_1]} \I(m_{n-1,i} = k) \bone_{\C_{n}}(i)}{\sum_{i \in [L_1]} \I(m_{n-1,i} = k) + \sum_{i \in [L_2]} \I(m_{n-1,i} = k)}.
\end{align*}
Note that
\begin{align*}
    \sum_{i \in [L_2]} \I(m_{n-1,i} = k) = \begin{cases}
        L_2 & \text{ if } k = 0, \\
        0 & \text{ if } k \ge 1.
    \end{cases}
\end{align*}
While the terms $\sum_{i \in [L_1]} \I(m_{n-1,i} = k) \bone_{\C_{n}}(i)$ and $\sum_{i \in [L_1]} \I(m_{n-1,i} = k)$ are easy to estimate, the value of $L_2$ is computationally expensive to estimate for a large time-series of collaboration data. For simplicity, we estimate $L_2$ by $M - \sum_{i \in [L_1]} \I(m_{n-1,i} = k)$, where $M$ is the maximum (over all months) of the monthly counts of productive authors (i.e. authors who have written a paper in that month).

We plot the estimated co-authorship probability parameters $(F_n(k))_{n\ge1}$ for $k = 0, 1, 2, 3$  in Figures~\ref{fig:csWRITE}-(c), \ref{fig:mathWRITE}-(c) and \ref{fig:phyWRITE}-(c). Notice the qualitative similarity of these plots with Figure~\ref{fig:rec-sol}-(b), where we plot these parameters for an instance of the simple parametric sub-model \eqref{eq:param-submodel}.


\begin{figure}[!htbp]
     \centering
     \begin{tabular}{ccc}
     \includegraphics[width = 0.3\textwidth]{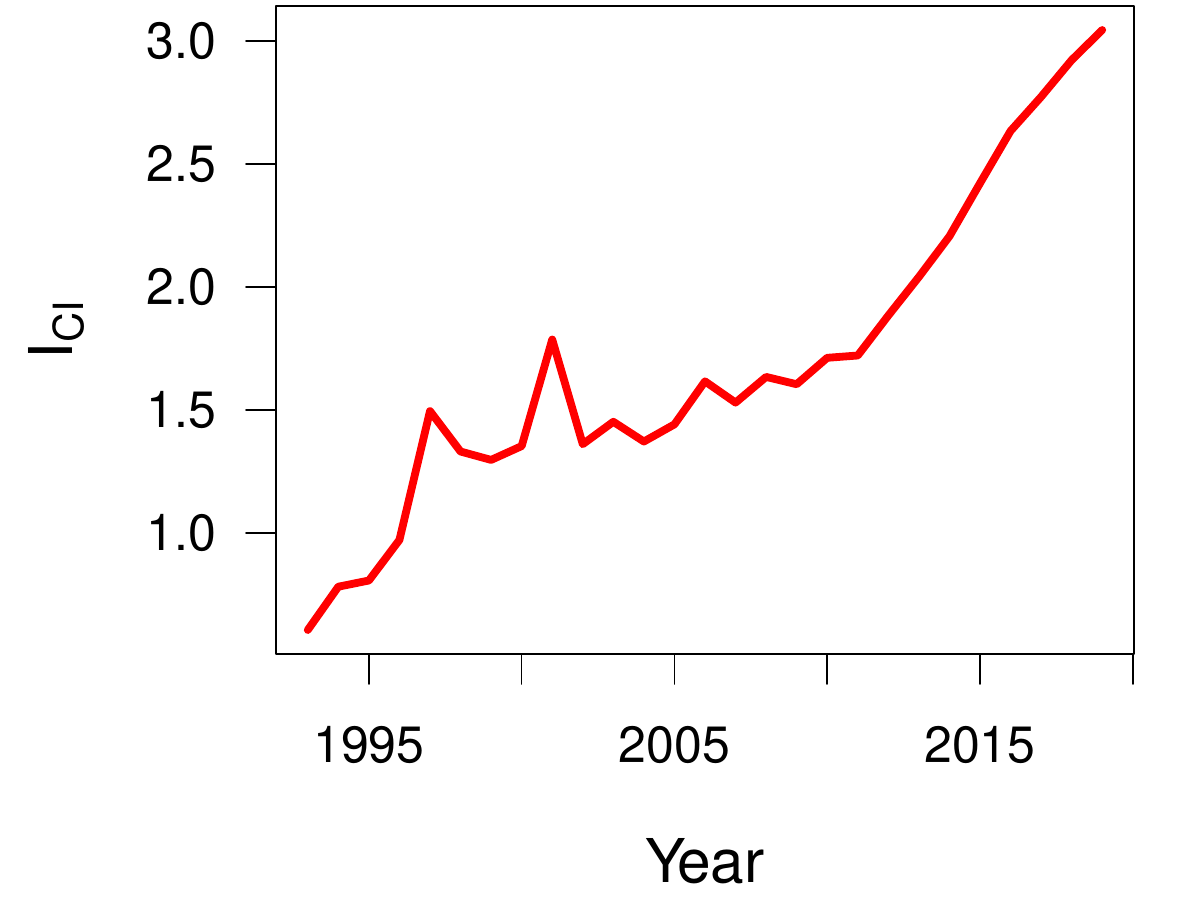} & \includegraphics[width = 0.3\textwidth]{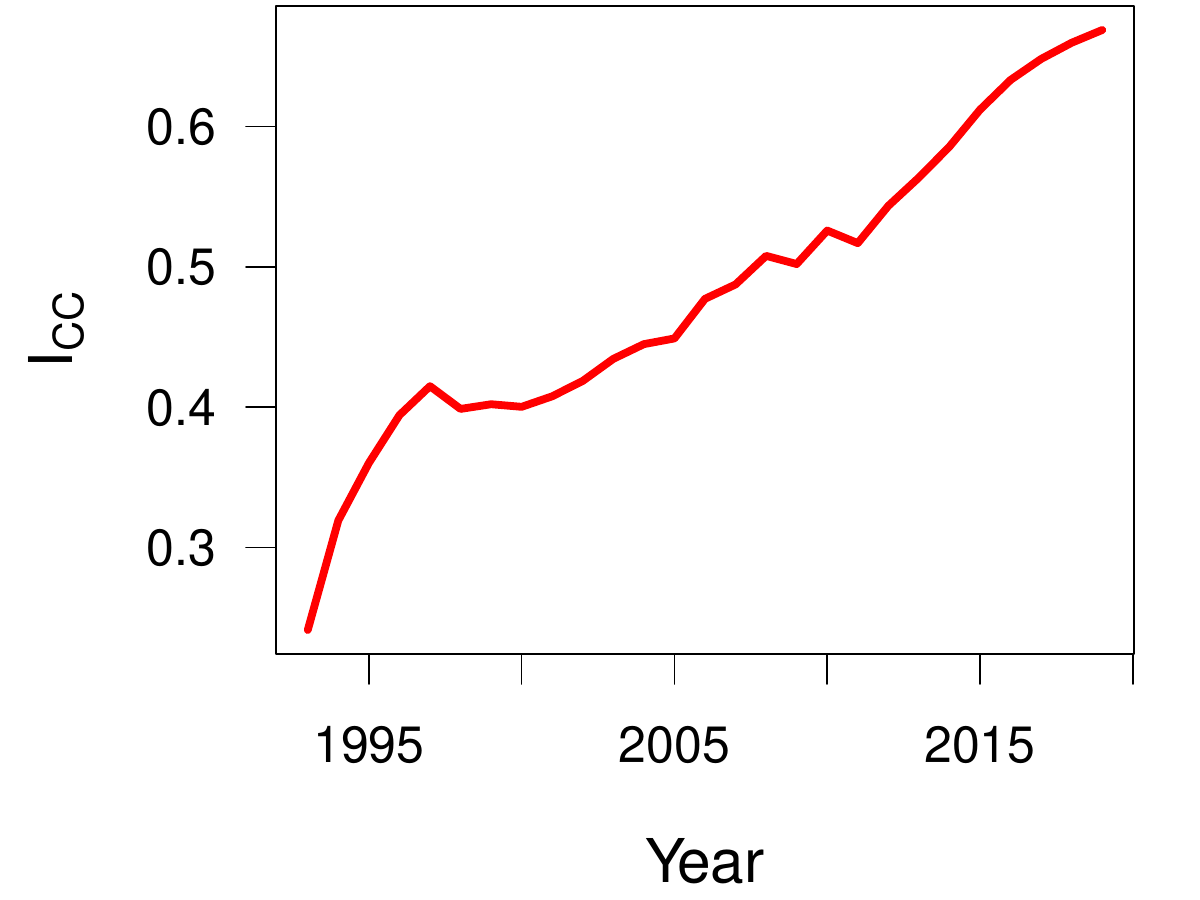} & \includegraphics[width = 0.3\textwidth]{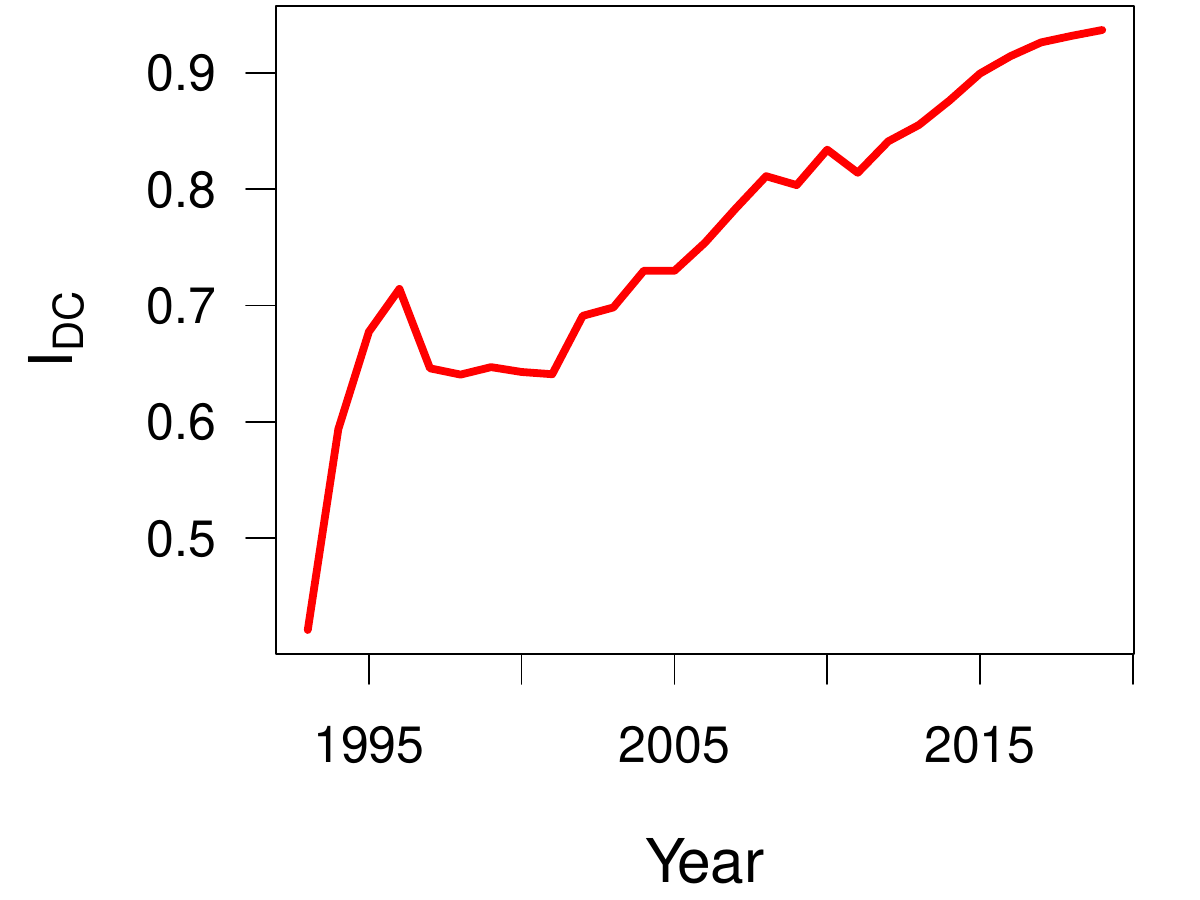}
     \end{tabular}
     \caption{ArXiv data for CS: Yearly plot of $I_\CI$, $I_\DC$, and $I_\CC$ based on all authors.}
     \label{fig:CS_indices_all_authors}
\end{figure}

\begin{figure}[!htbp]
     \centering
     \begin{tabular}{ccc}
     \includegraphics[width = 0.3\textwidth]{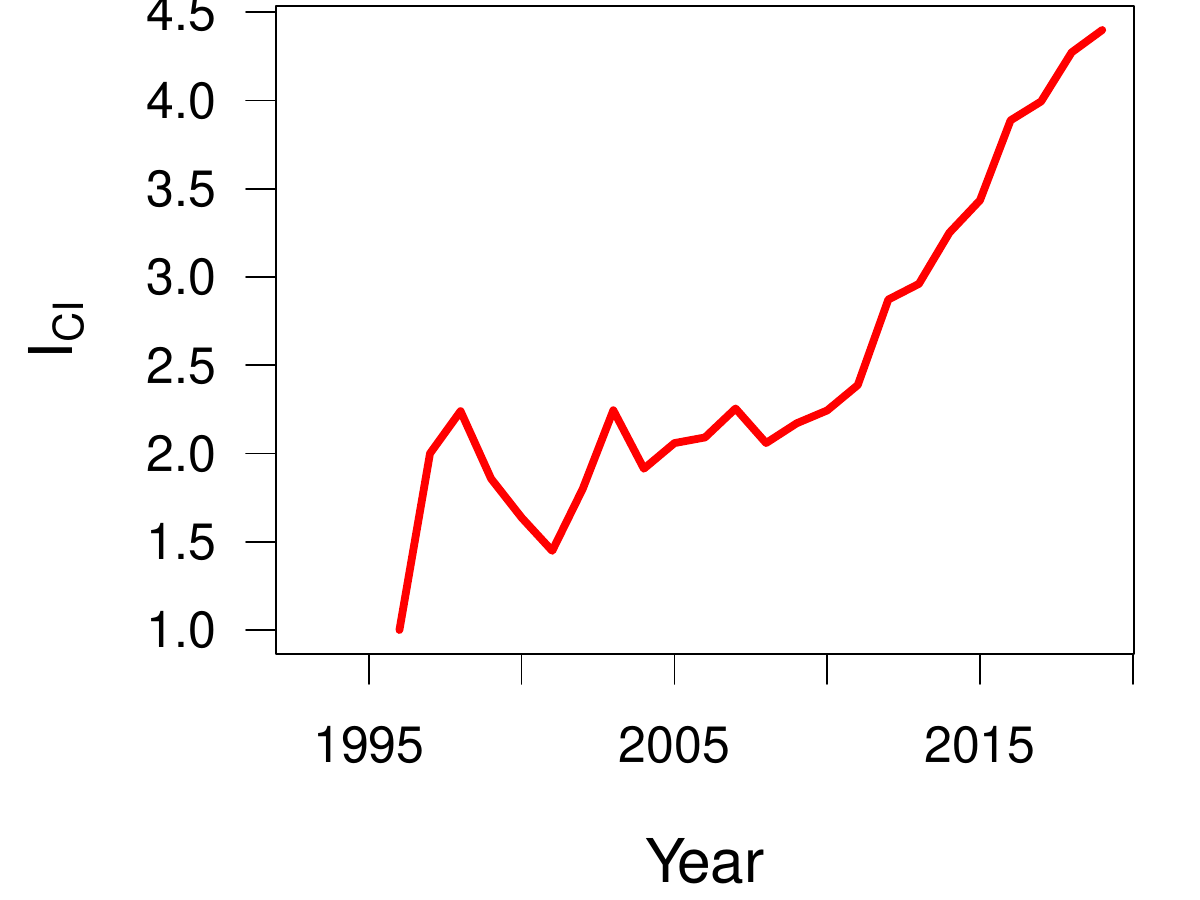} & \includegraphics[width = 0.3\textwidth]{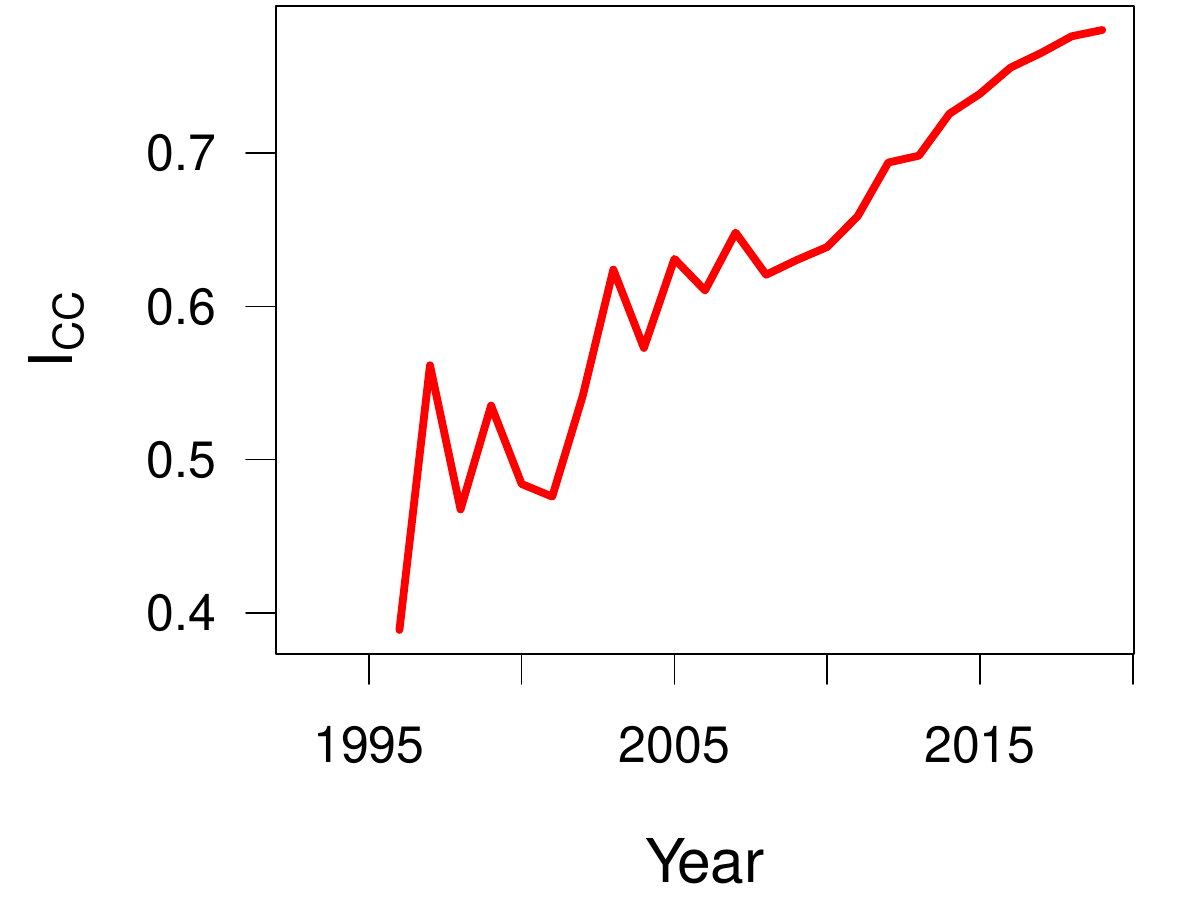} & \includegraphics[width = 0.3\textwidth]{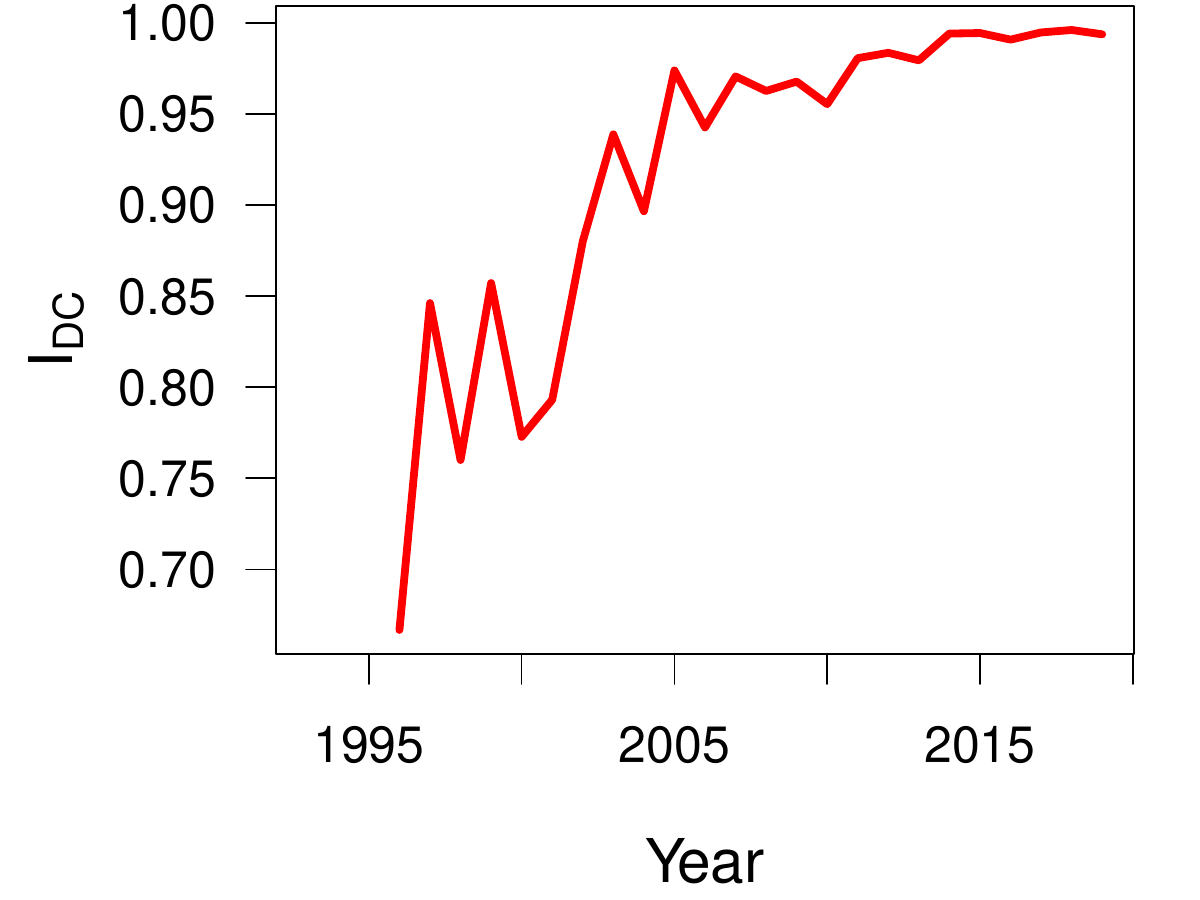}
     \end{tabular}
     \caption{ArXiv data for CS: Yearly plot of $I_\CI$, $I_\DC$, and $I_\CC$ based on the top $100$ authors.}
     \label{fig:CS_indices_top_100_authors}
\end{figure}

\begin{figure}[!htbp]
     \centering
     \begin{tabular}{ccc}
     \includegraphics[width = 0.3\textwidth]{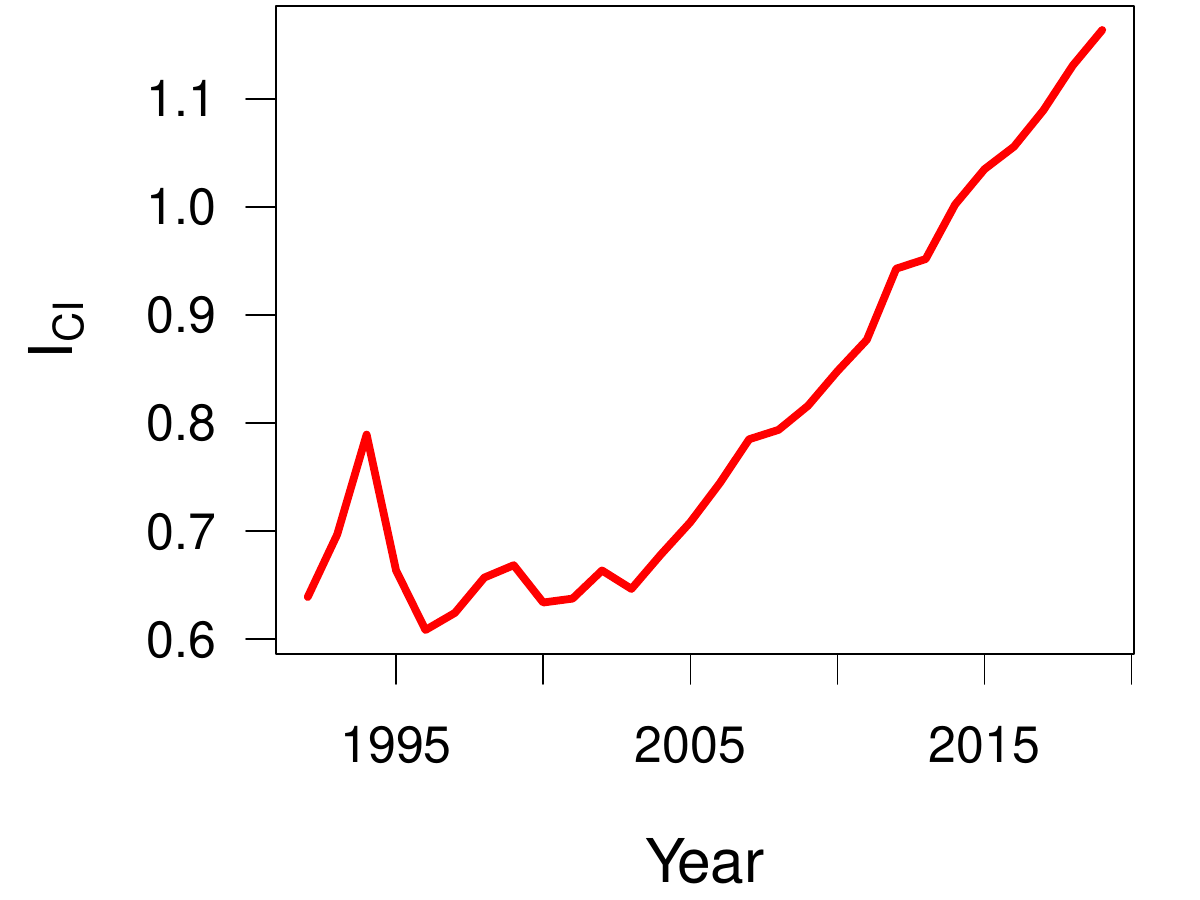} & \includegraphics[width = 0.3\textwidth]{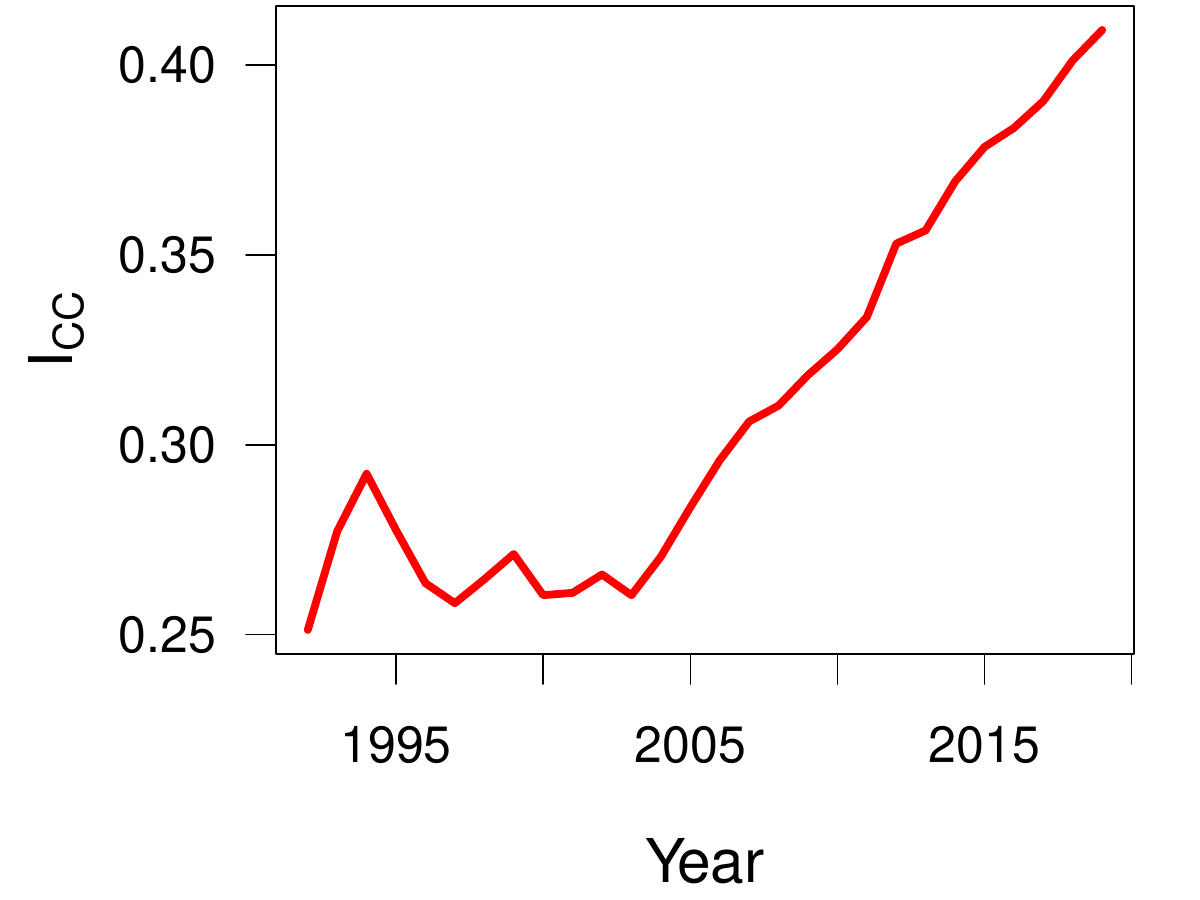} & \includegraphics[width = 0.3\textwidth]{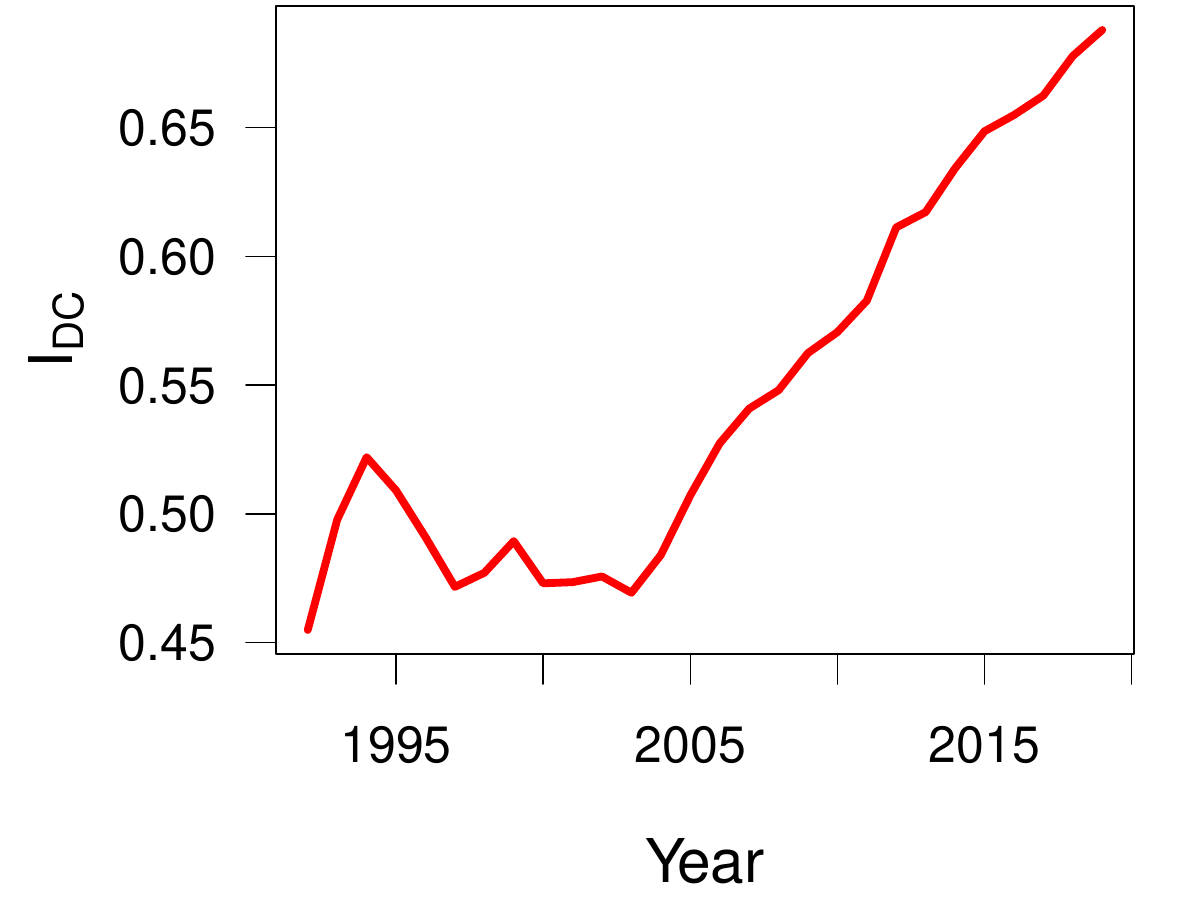}
     \end{tabular}
     \caption{ArXiv data for Mathematics: Yearly plot of $I_\CI$, $I_\DC$, and $I_\CC$ based on all authors.}
     \label{fig:math_indices_all_authors}
\end{figure}

\begin{figure}[!htbp]
     \centering
     \begin{tabular}{ccc}
     \includegraphics[width = 0.3\textwidth]{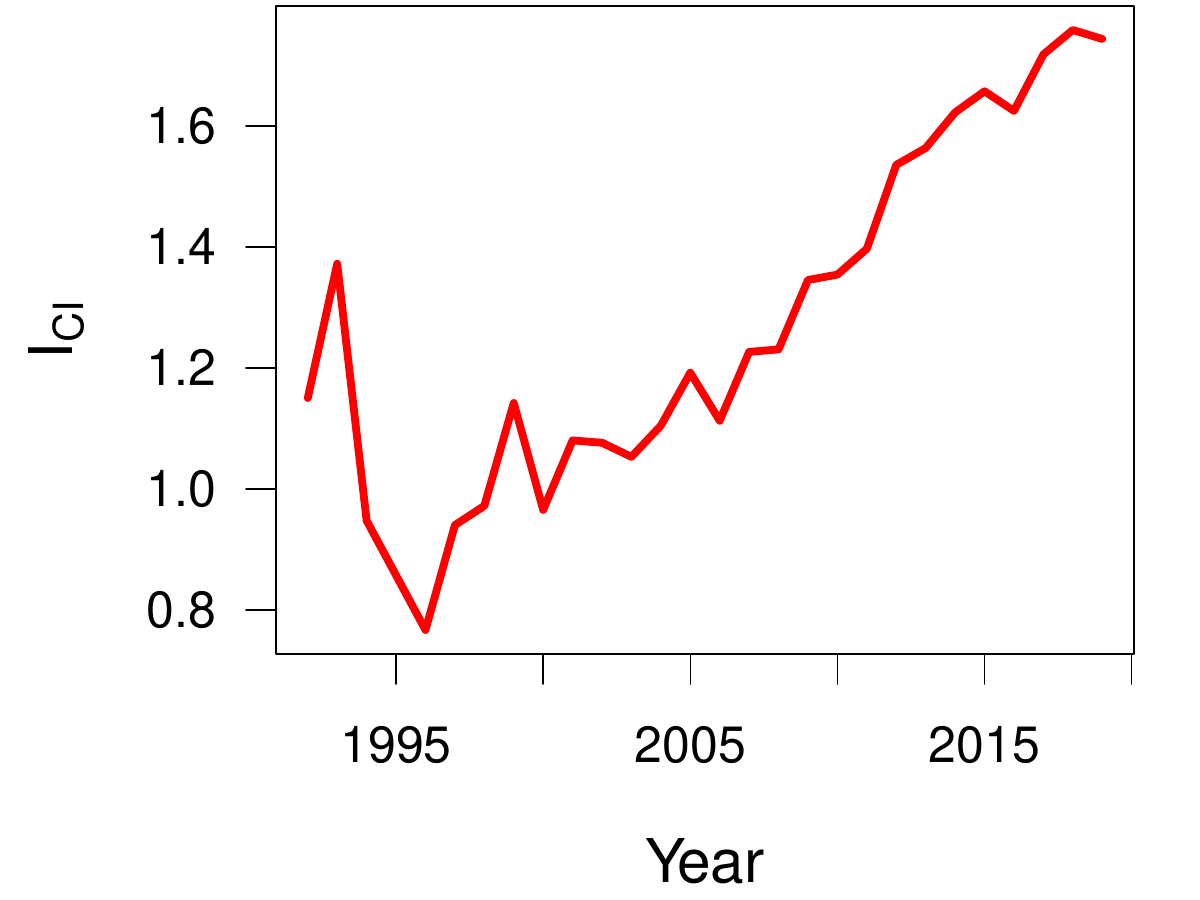} & \includegraphics[width = 0.3\textwidth]{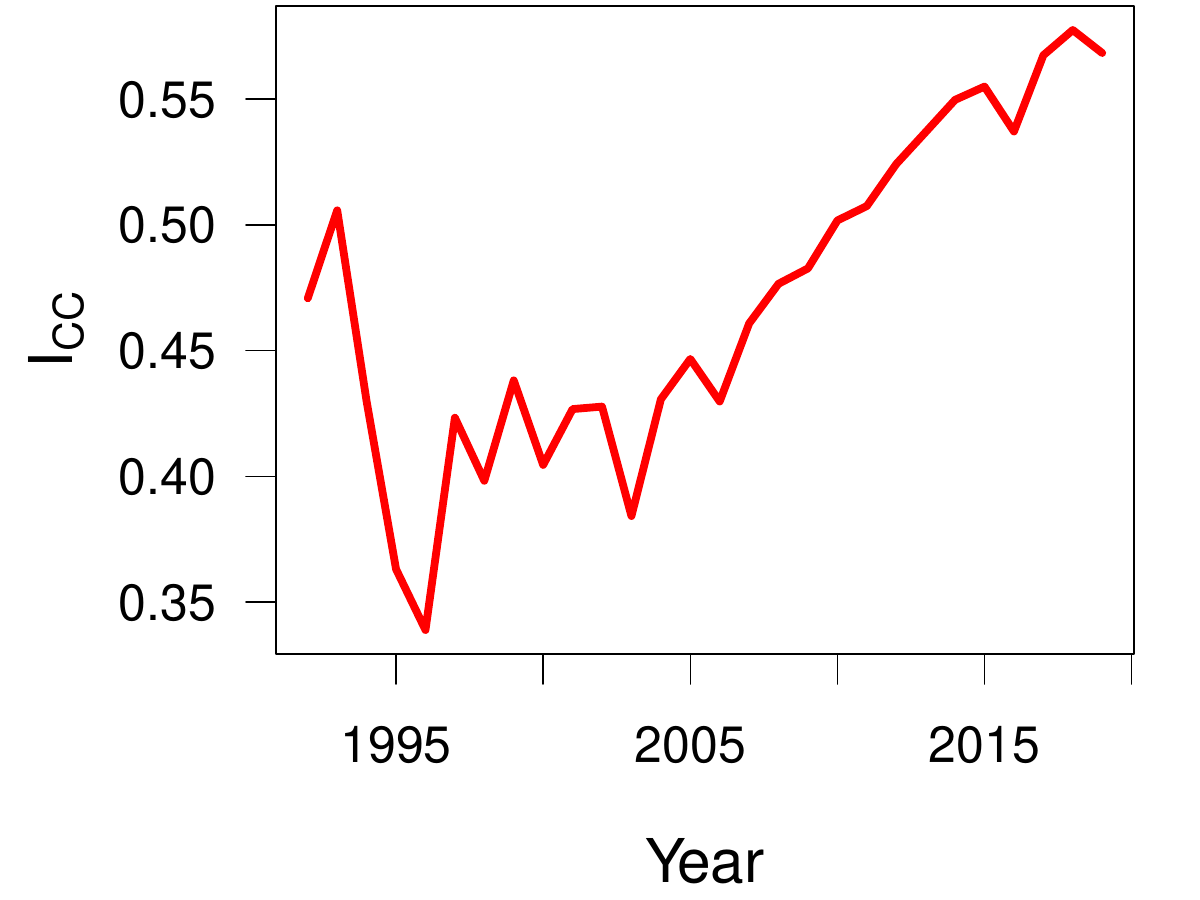} & \includegraphics[width = 0.3\textwidth]{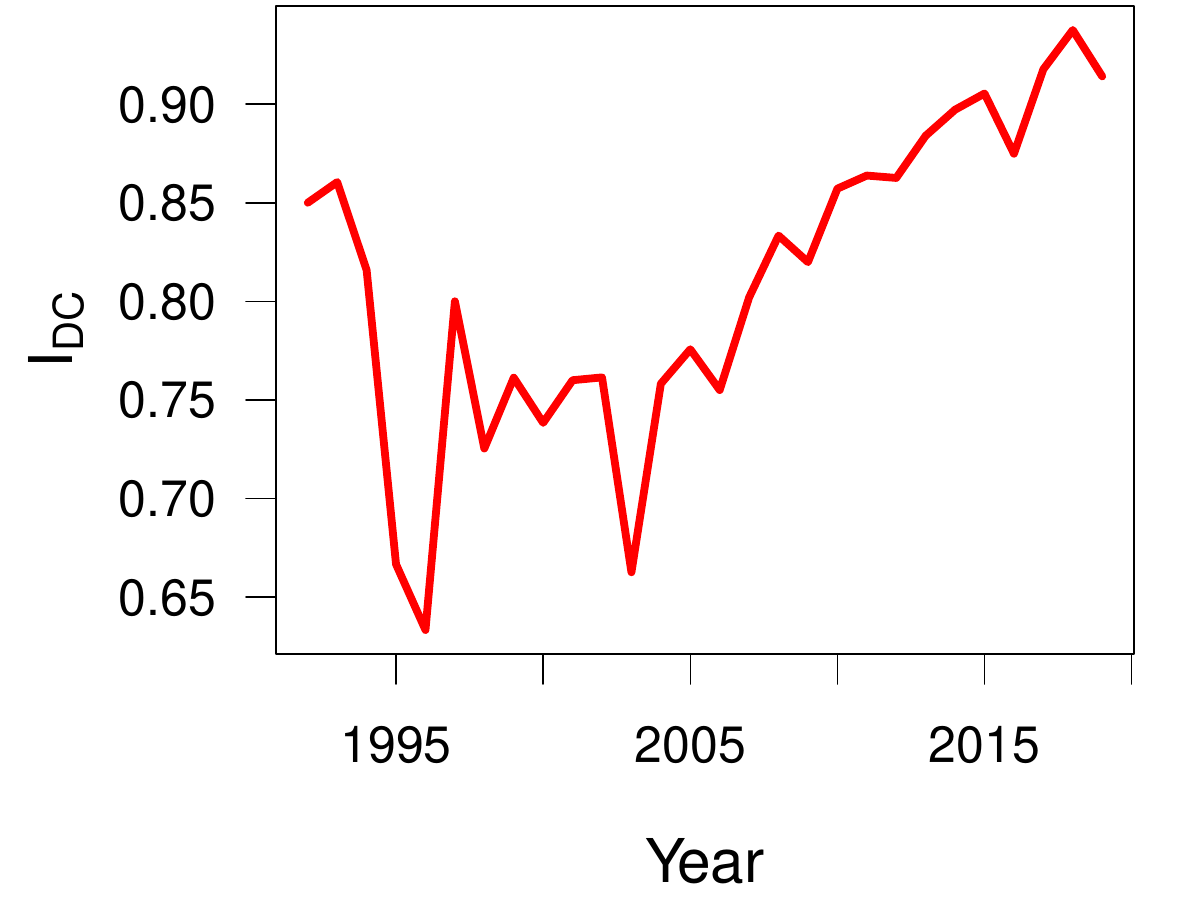}
     \end{tabular}
     \caption{ArXiv data for Mathematics: Yearly plot of $I_\CI$, $I_\DC$, and $I_\CC$ based on the top $100$ authors.}
     \label{fig:math_indices_top_100_authors}
\end{figure}

\begin{figure}[!htbp]
     \centering
     \begin{tabular}{ccc}
     \includegraphics[width = 0.3\textwidth]{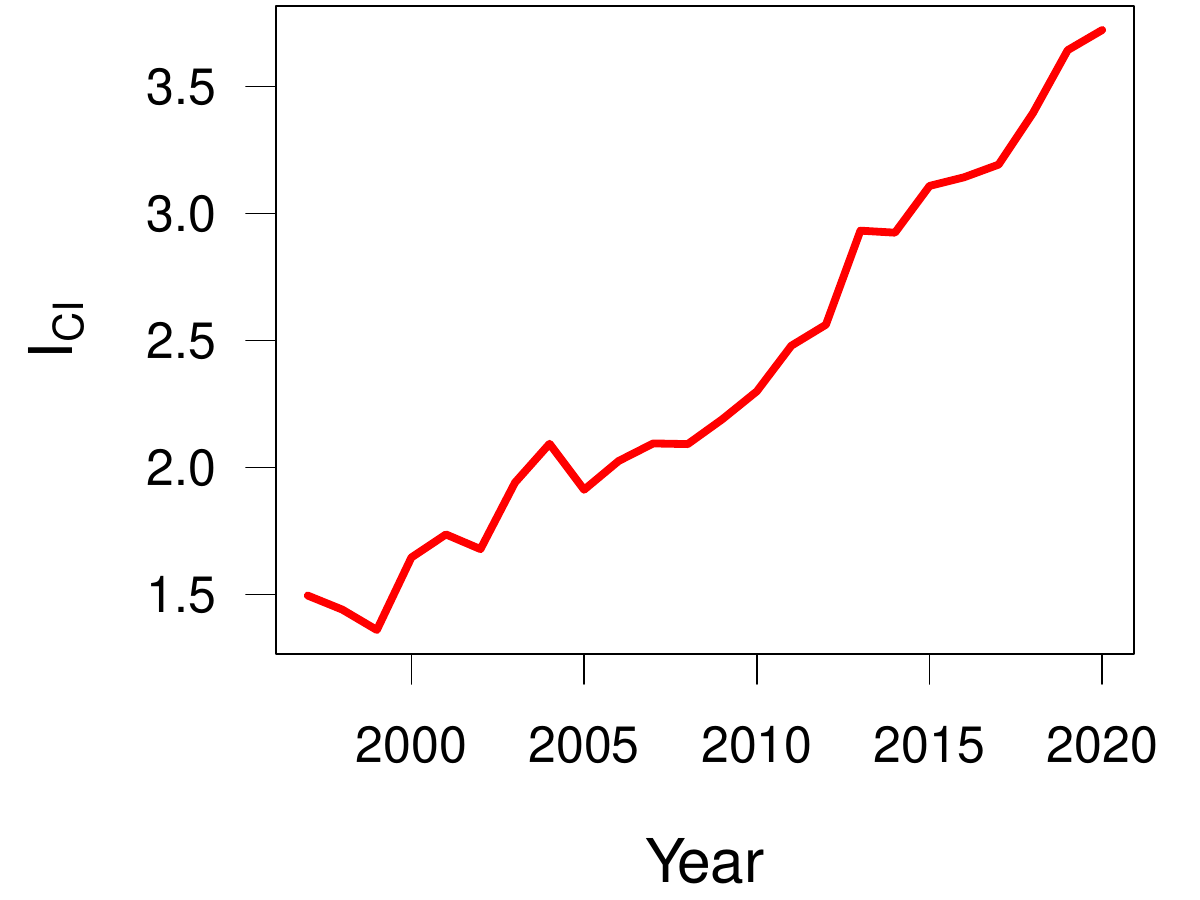} & \includegraphics[width = 0.3\textwidth]{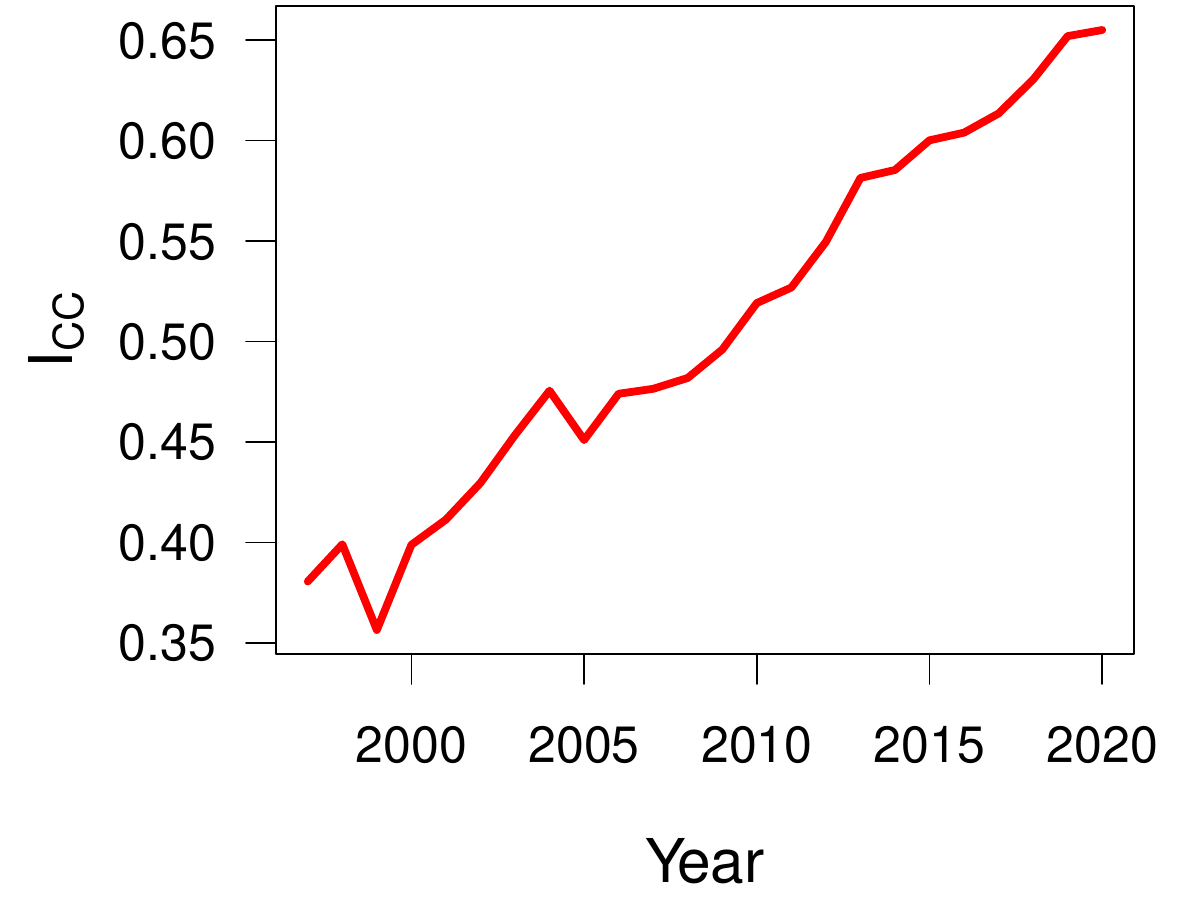} & \includegraphics[width = 0.3\textwidth]{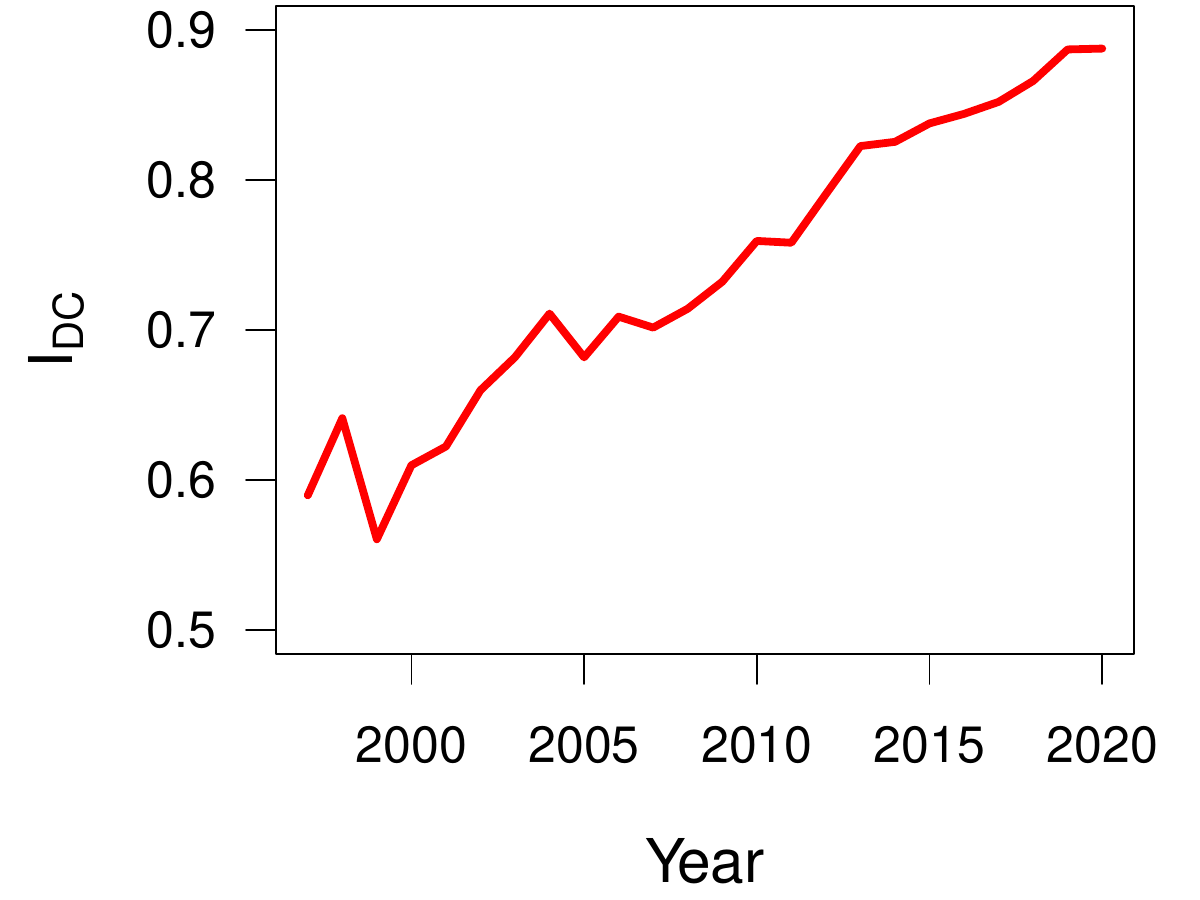}
     \end{tabular}
     \caption{ArXiv data for Physics: Yearly plot of $I_\CI$, $I_\DC$, and $I_\CC$ based on all authors.}
     \label{fig:phy_indices_all_authors}
\end{figure}

\begin{figure}[!htbp]
     \centering
     \begin{tabular}{ccc}
     \includegraphics[width = 0.3\textwidth]{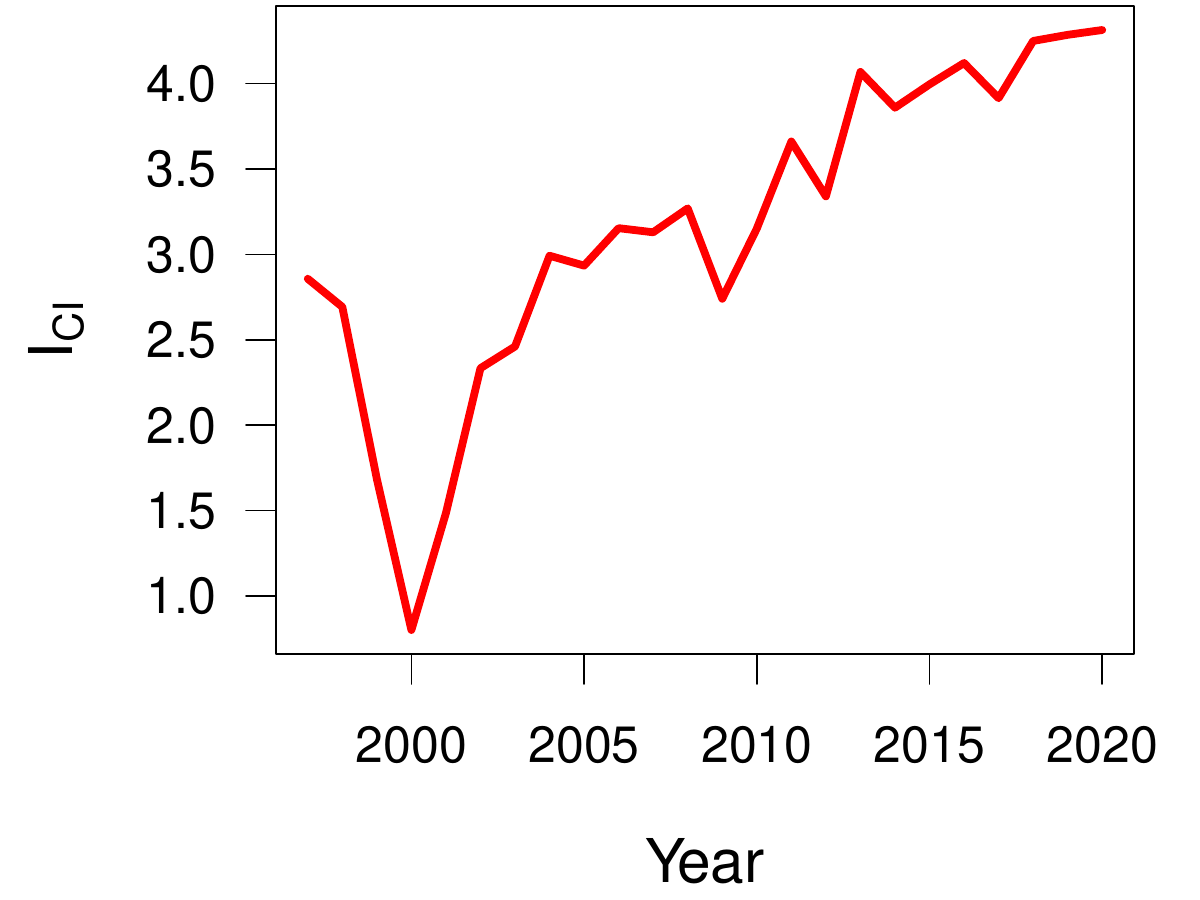} & \includegraphics[width = 0.3\textwidth]{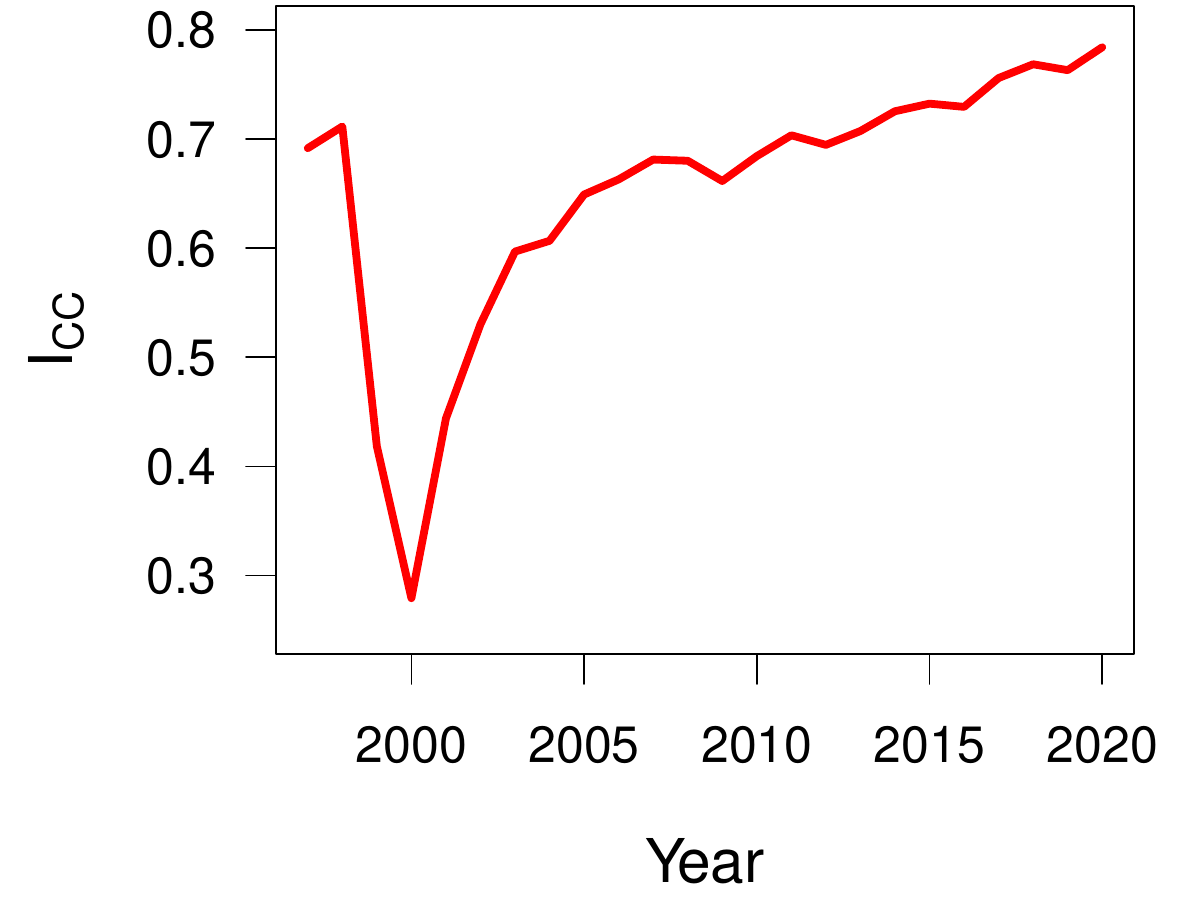} & \includegraphics[width = 0.3\textwidth]{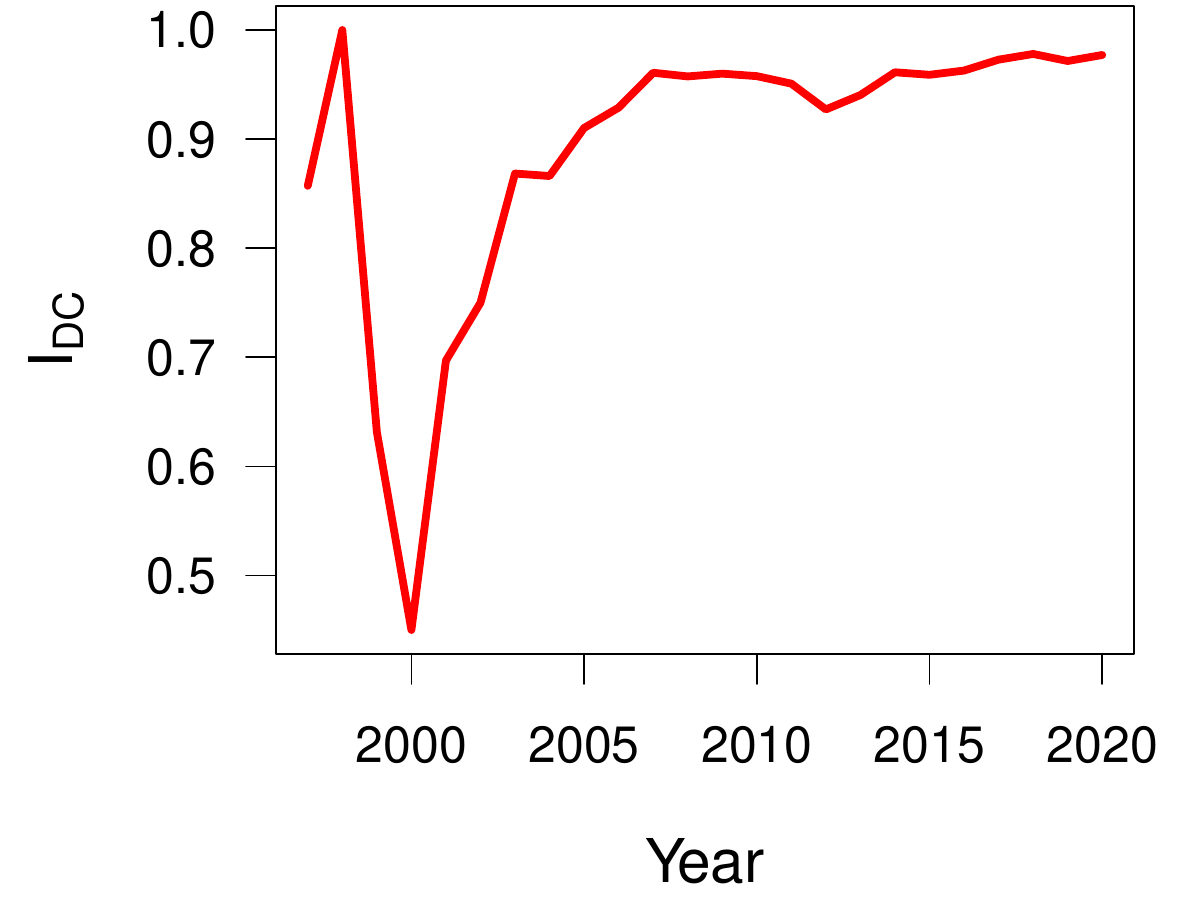}
     \end{tabular}
     \caption{ArXiv data for Physics: Yearly plot of $I_\CI$, $I_\DC$, and $I_\CC$ based on the top $100$ authors.}
     \label{fig:phy_indices_top_100_authors}
\end{figure}


\begin{figure}[!htbp]
     \centering
     \begin{tabular}{ccc}
     \includegraphics[width = 0.3\textwidth]{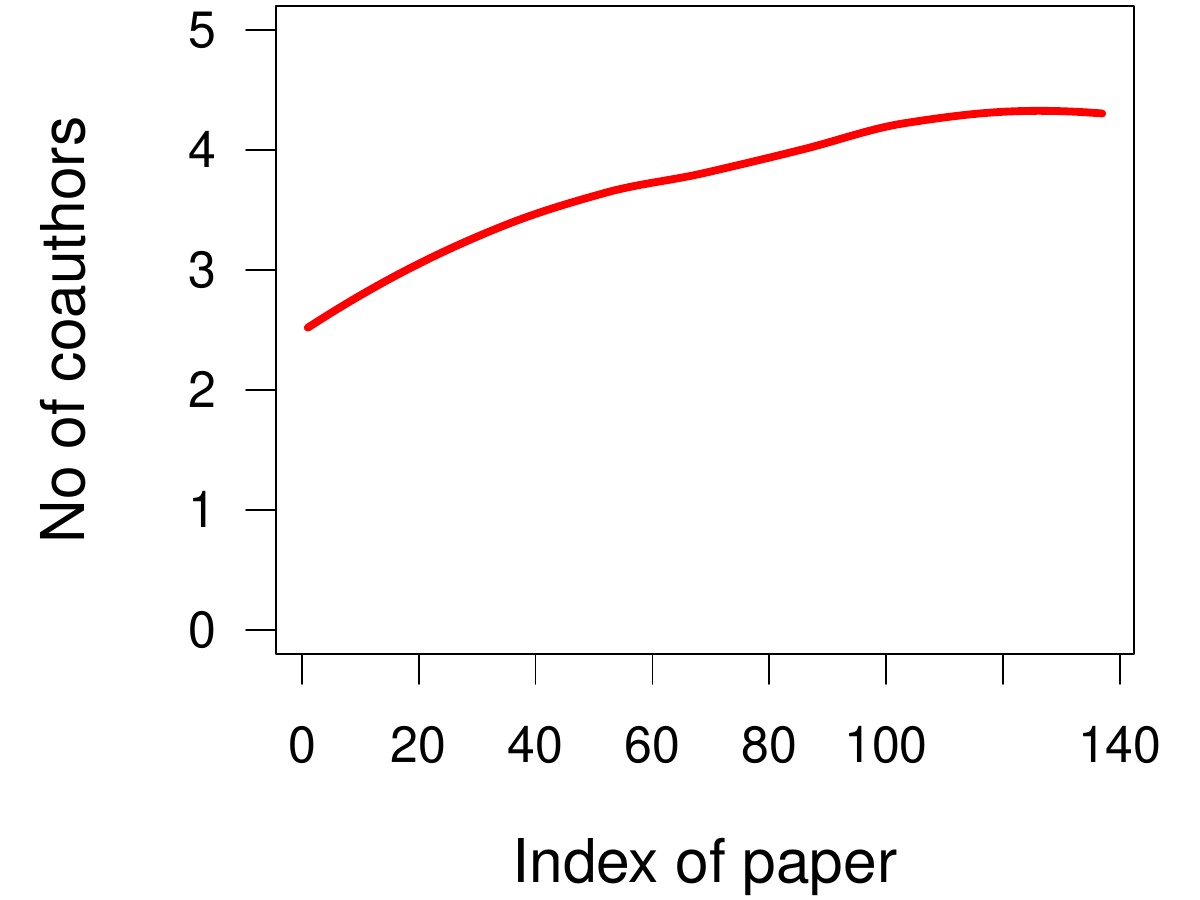} & \includegraphics[width = 0.3\textwidth]{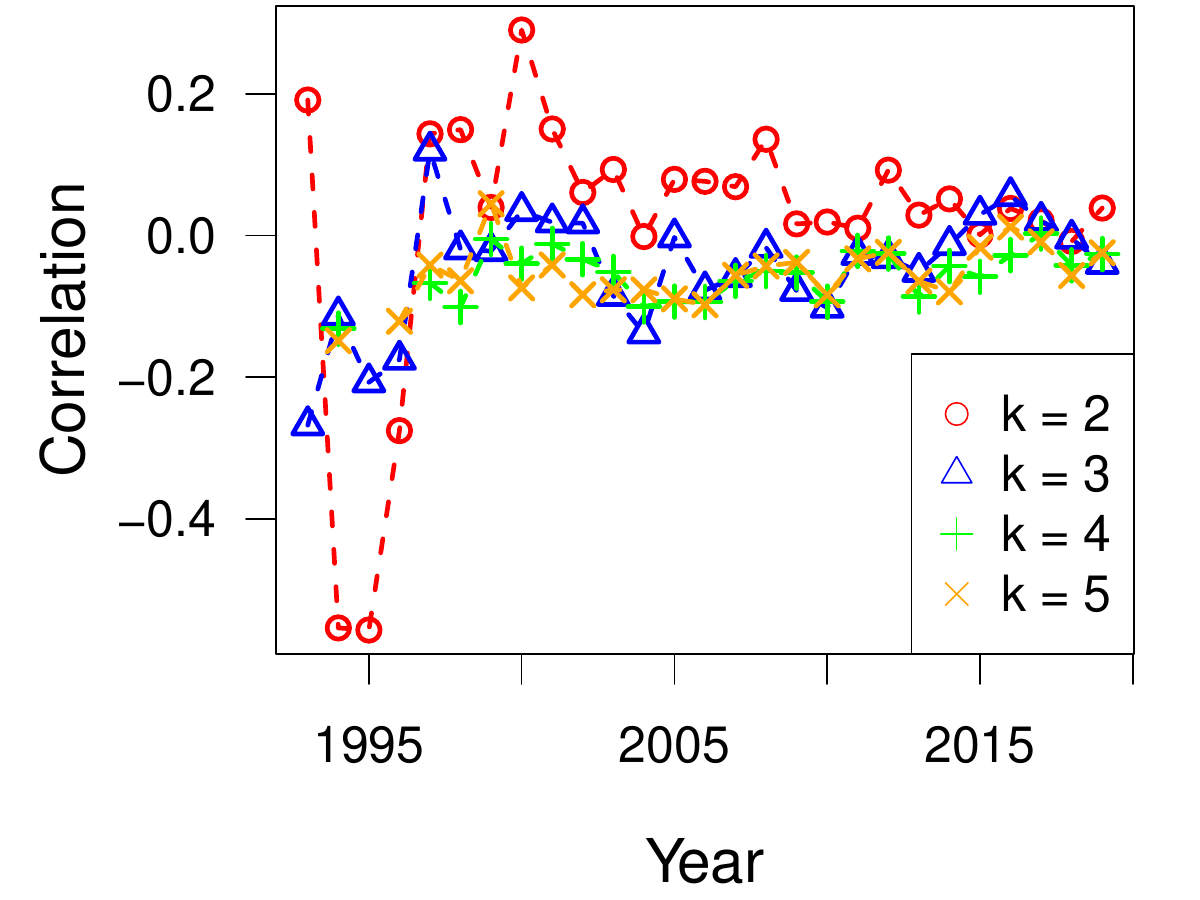} & \includegraphics[width = 0.3\textwidth]{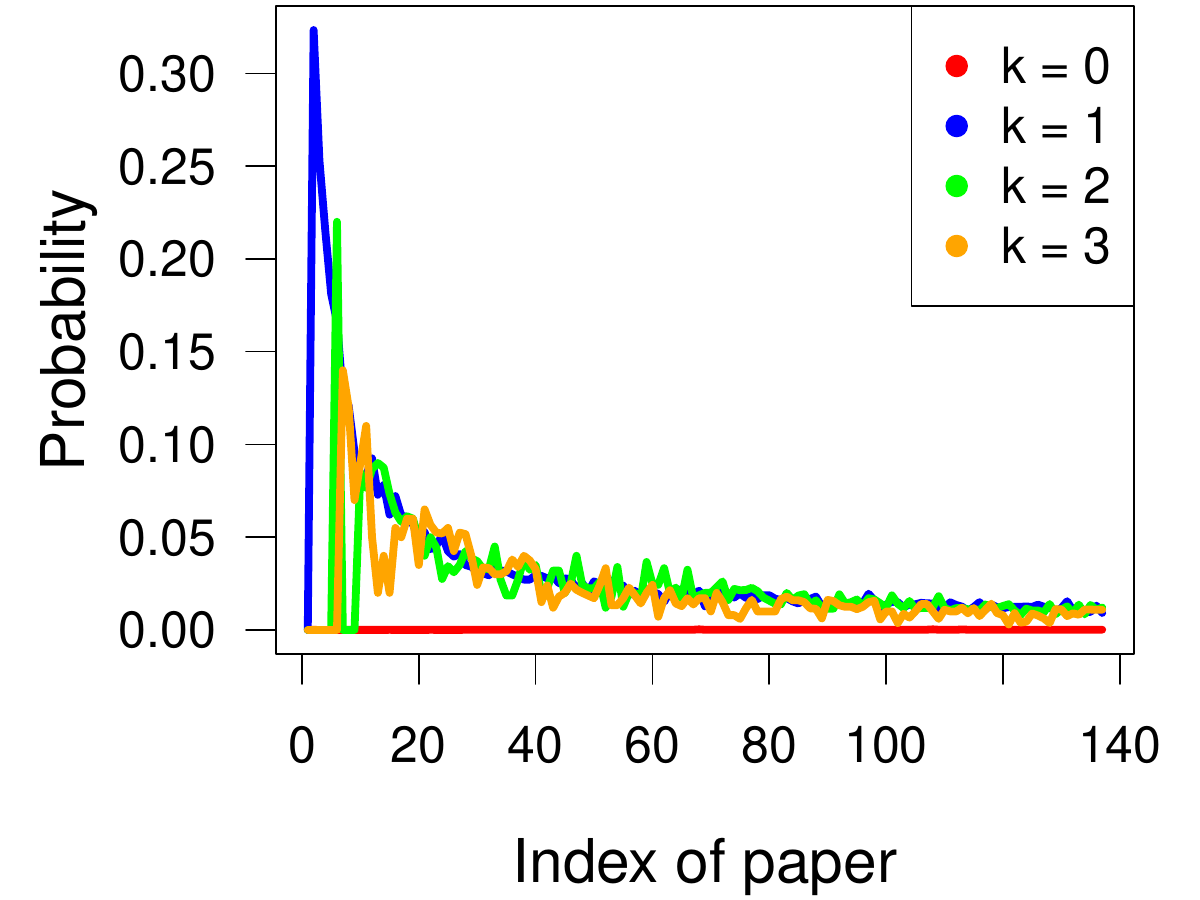} \\
     (a) & (b) & (c)
     \end{tabular}
     \caption{ArXiv data for CS: (a) Plot of the no of co-authors in the $k$-th paper of an author versus $k$ (based on the top 100 most productive authors); (b) plot of $\corr(X_1[t, t + \delta], X_k[t, t + \delta])$, for $\delta = 1$ year; (c) plot of $\hat{F}_n(k)$ versus $n$.}
     \label{fig:csWRITE}
\end{figure}

\begin{figure}[!htbp]
     \centering
     \begin{tabular}{ccc}
     \includegraphics[width = 0.3\textwidth]{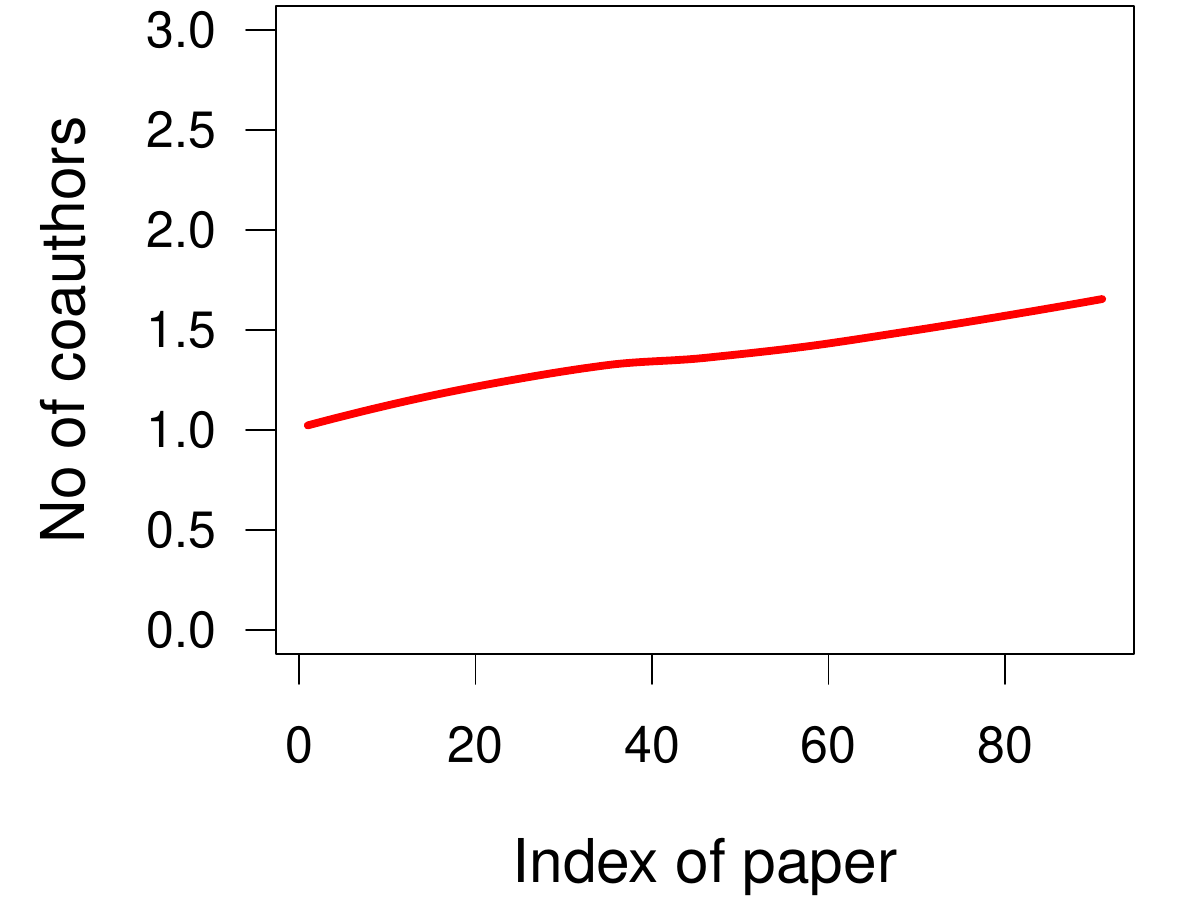} & \includegraphics[width = 0.3\textwidth]{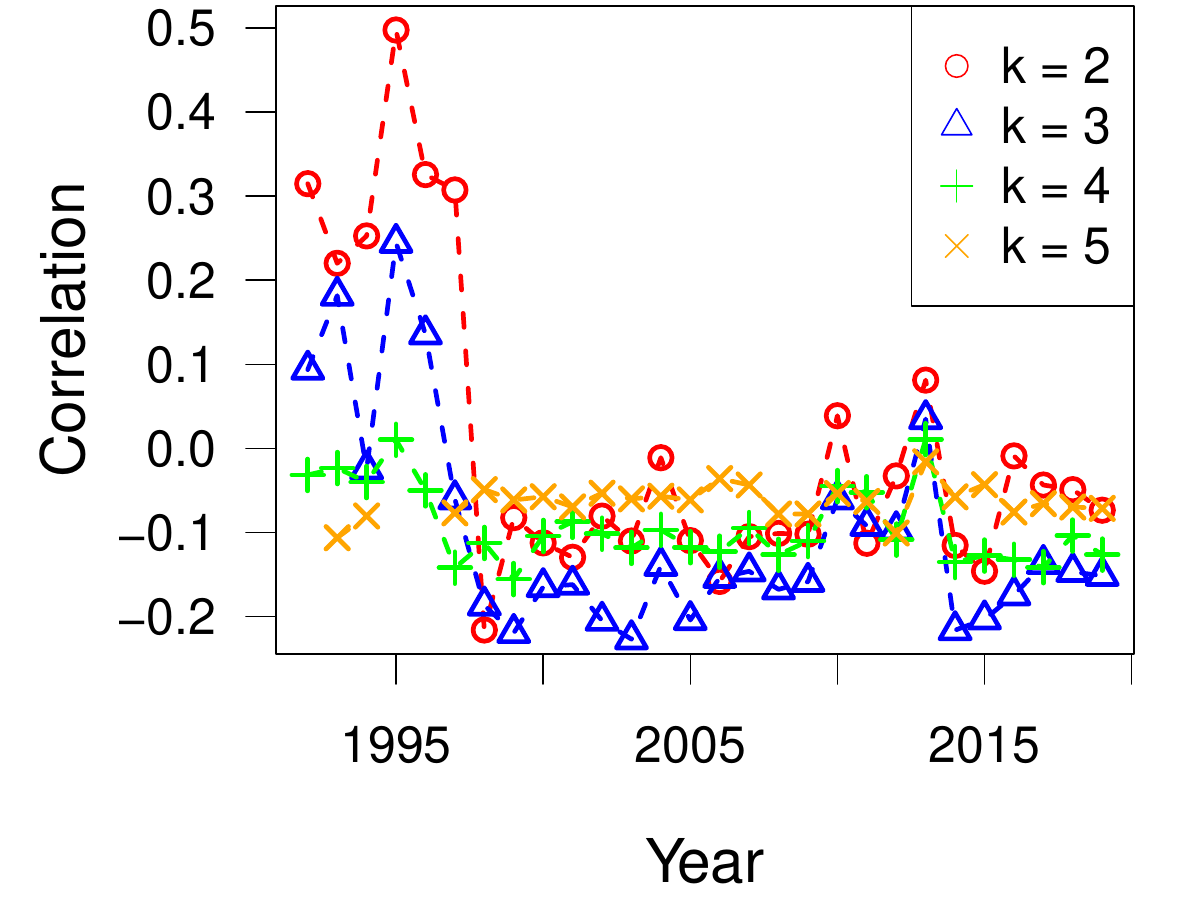} & \includegraphics[width = 0.3\textwidth]{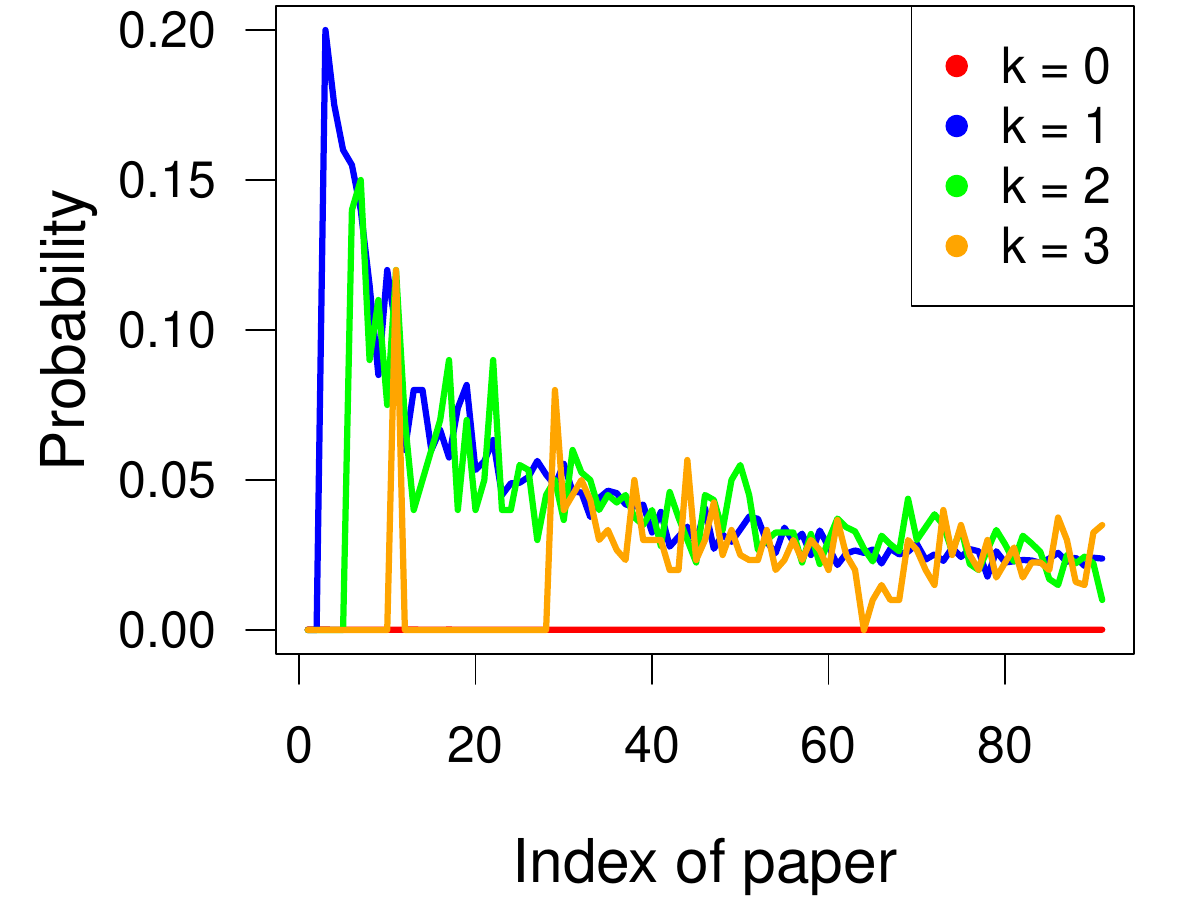} \\
     (a) & (b) & (c)
     \end{tabular}
     \caption{ArXiv data for Mathematics: (a) Plot of the no of co-authors in the $k$-th paper of an author versus $k$ (based on the top 100 most productive authors); (b) plot of $\corr(X_1[t, t + \delta], X_k[t, t + \delta])$, for $\delta = 1$ year; (c) plot of $\hat{F}_n(k)$ versus $n$.}
     \label{fig:mathWRITE}
\end{figure}

\begin{figure}[!htbp]
     \centering
     \begin{tabular}{ccc}
     \includegraphics[width = 0.3\textwidth]{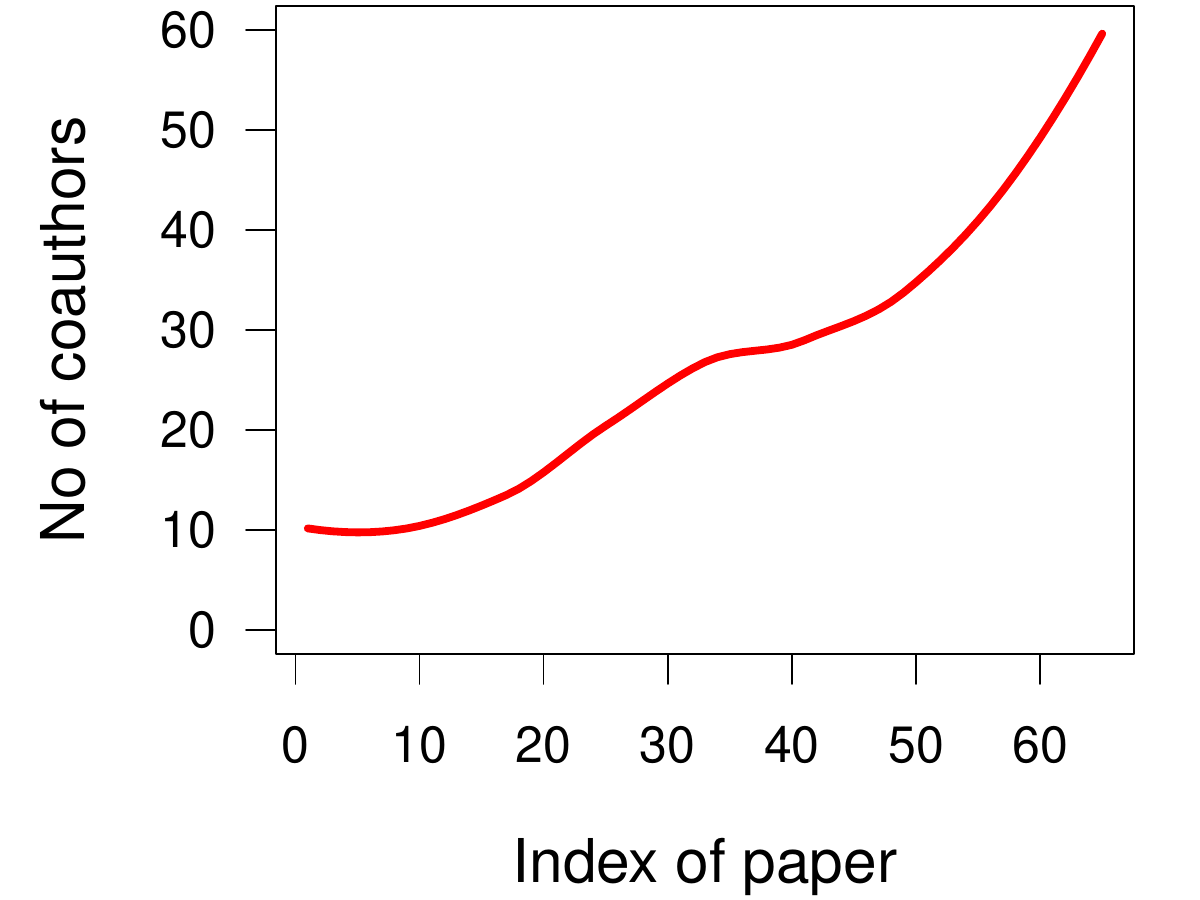} & \includegraphics[width = 0.3\textwidth]{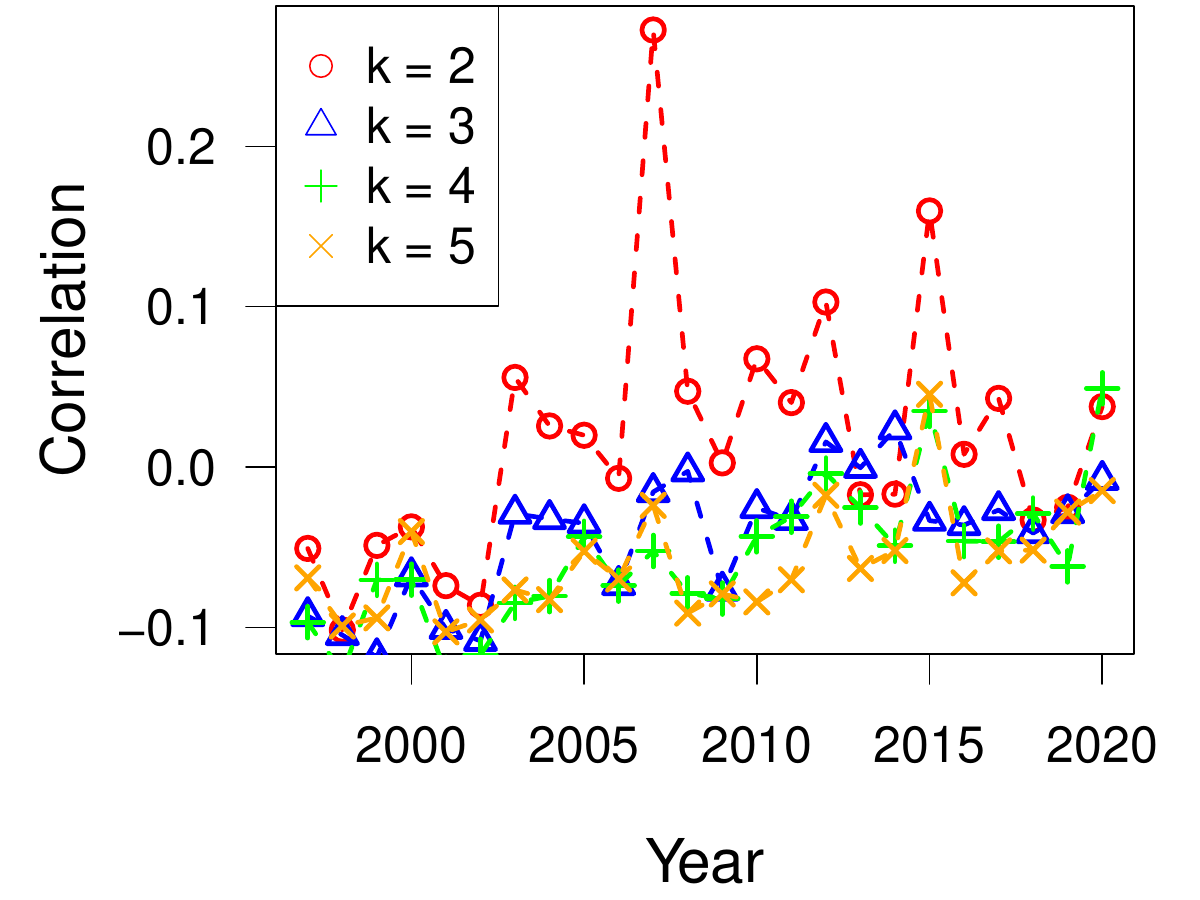} & \includegraphics[width = 0.3\textwidth]{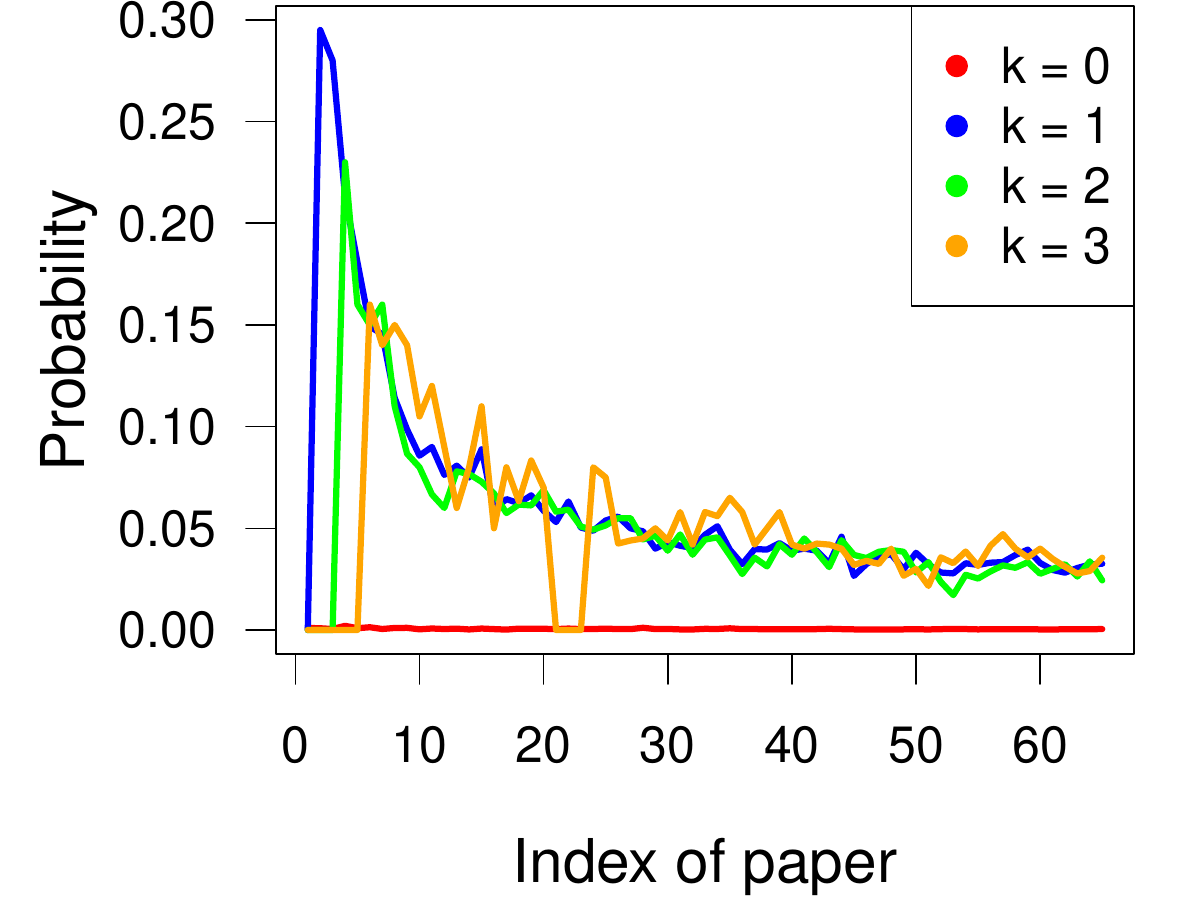} \\
     (a) & (b) & (c)
     \end{tabular}
     \caption{ArXiv data for Physics: (a) Plot of the no of co-authors in the $k$-th paper of an author versus $k$ (based on the top 100 most productive authors); (b) plot of $\corr(X_1[t, t + \delta], X_k[t, t + \delta])$, for $\delta = 1$ year; (c) plot of $\hat{F}_n(k)$ versus $n$.}
     \label{fig:phyWRITE}
\end{figure}

\section{Conclusion}\label{sec:conc}
In this work, we introduce a novel mean-field dynamic model for academic collaboration, by combining a non-homogeneous Poisson process modeling paper-writing events and a mechanism of choosing co-authors depending on past collaboration history. We also develop methods for estimating the (infinite dimensional) parameters of the model. We derive analytical expressions for various indices of collaboration under the proposed model in terms of the model parameters. The proposed model is highly flexible in that it can explain a variety of empirically observed features found in real-world collaboration networks. For example, the behaviours of various indices of collaboration over time, or the effect that people tend to collaborate more with people they have already collaborated with in the past, but this effect seems to reverse as the number of joint collaborations increase over time.

Although the proposed model reflects many complexities of real-world collaboration networks, there are some limitations. For example, the mean-field nature of the model overlooks the fact that the intensity of producing research vary across authors. This can potentially be tackled by considering the collaboration dynamics of all authors simultaneously, allocating each a separate intensity functional. Also, in the proposed model, an author does not consider the relative impacts (measured in terms of the impact factor of the journal in which a work is published, or the number of times the work has been cited) of various existing research works while choosing future collaborators. Nor does she take into account how successful the potential collaborators are. We believe that the proposed model can be suitably tweaked to incorporate these effects and leave the exploration of such extensions to future work. Furthermore, the proposed framework can also be adopted to study evolution of interdisciplinary research collaboration by assigning each author a discipline or specific field of research based on their publication record. We will pursue this direction in a future work.

\section*{Funding}
SSM is partially supported by an INSPIRE research grant (DST/INSPIRE/04/2018/002193) from the Department of Science and Technology, Government of India.

\bibliographystyle{plainnat}
\bibliography{collab_model.bib}

\clearpage
\appendix

\section{Proofs of our main results}\label{sec:proofs}
We begin with a technical lemma which will be used repeatedly later in various proofs.
\begin{lemma}\label{lem:est_En}
Let $N$ be an inhomogeneous Poisson process with intensity functional $\lambda(\cdot)$, whose event times are $E_n, n\ge 1$. Let $U = N[0,t]$ and $\lambda_{t, h} = \int_t^{t + h} \lambda(u) \, du$. Then, for all $t, h > 0$ satisfying $\lambda_{t, h} < 1$ and any functions $g: \NN \rightarrow [0, 1]$ and $h : \NN \times \NN \rightarrow [0, 1]$,   the following inequalities hold: 
\begin{align}\label{eq:bound_En}
    &\bigg|\sum_{n\ge 1}g(n)\P(E_n \in [t, t + h]) - e^{-\lambda_{t, h}}\lambda_{t, h}\sum_{n \ge 1}g(n) \P(U = n - 1)\bigg| \le \lambda_{t, h}^{2} + \frac{\lambda_{t, h}^{2}}{1 - \lambda_{t, h}}; \\ \nonumber
    &\bigg|\sum_{\substack{n_1 \ge 1 \\ n_2 > n_1}}h(n_1, n_2)\P(E_{n_1}, E_{n_2} \in [t, t + h]) - e^{-\lambda_{t, h}}\frac{\lambda_{t, h}^{2}}{2}\sum_{n \ge 1} h(n, n + 1) \P(U = n - 1) \bigg| \\  \label{eq:bound_En1_En2}
    &\hspace{22em} \le \lambda_{t,h}^{3}+ \frac{\lambda_{t, h}^{3}}{1-\lambda_{t, h}} + \frac{\lambda_{t, h}^{3}}{\big(1 - \lambda_{t, h}\big)^2}.
\end{align}
\end{lemma}
\begin{proof}
Let $\Lambda(t) = \int_{0}^t\lambda(u)\, du$ and $V = N[t, t + h]$. Note that $U$ and $V$ are independent, $U \sim \pois(\Lambda(t))$ and $V \sim \pois(\lambda_{t,h})$. Now if $W \sim \pois(\mu)$, then we have the estimate 
\begin{equation}\label{eq:poi_tail_trivial}
    \P(W \ge k) = e^{-\mu} \sum_{\ell \ge k} \frac{\mu^\ell}{\ell!} = e^{-\mu} \mu^k \sum_{\ell \ge k} \frac{\mu^{\ell - k}}{\ell!} \le e^{-\mu} \mu^k \sum_{\ell \ge k} \frac{\mu^{\ell - k}}{(\ell - k)!} = \mu^k.
\end{equation}
\textit{Proof of \eqref{eq:bound_En}.}
We begin by noting that
\begin{align*}
     \P(E_n \in [t, t + h]) &= \P(U \le n-1, U+V \ge n)=\sum_{k = 0}^{n - 1} \P(U = k) \P(V \ge n-k).
\end{align*}
Now
\begin{align*}
    \sum_{n \ge 1} &g(n) \P(E_n \in [t, t + h])
    \\
    &= \sum_{n \ge 1}g(n)\sum_{k = 0}^{n - 1} \P(U = k) \P(V \ge n-k)
    \\
    &= \sum_{n \ge 1} g(n)\P(U = n - 1) \P(V \ge 1) + \sum_{n\ge 2}g(n)\sum_{k = 0}^{n - 2} \P(U = k) \P(V \ge n - k).
\end{align*}
Therefore
\begin{align*}
    0 \le \sum_{n \ge 1} g(n) &\P(E_n \in [t, t + h]) - e^{-\lambda_{t, h}} \lambda_{t, h} \sum_{n \ge 1} g(n)\P(U = n - 1) \\
    &=\sum_{n \ge 1}g(n)\P(U = n - 1) \P(V \ge 2) + \sum_{n \ge 2}g(n)\sum_{k = 0}^{n - 2} \P(U = k) \P(V \ge n - k)
    \\
    &\le \P(V \ge 2) \sum_{n \ge 1}\P(U = n - 1) + \sum_{k \ge 0}\P(U = k)\sum_{n \ge k + 2} \P(V \ge n - k)
    \\
    &\le \lambda_{t, h}^{2} + \sum_{k \ge 0} \P(U = k) \sum_{\ell \ge 2} \lambda_{t, h}^{\ell} \quad (\text{using \eqref{eq:poi_tail_trivial}})
    \\
    &= \lambda_{t, h}^{2} + \frac{\lambda_{t, h}^{2}}{1 - \lambda_{t, h}}.
\end{align*}
\textit{Proof of \eqref{eq:bound_En1_En2}.}
Note that for $n_2 > n_1$,
\[
    \P(E_{n_1}, E_{n_2} \in [t, t + h]) = \P(U \le n_1-1, U+V \ge n_2) =\sum_{k = 0}^{n_1 - 1} \P(U = k) \P(V \ge n_2-k).
\]
Now
\begin{align*}
\sum_{\substack{n_1 \ge 1 \\  n_2 > n_1}} &h(n_1, n_2) \P(E_{n_1}, E_{n_2} \in [t, t + h]) \\
&= \sum_{n \ge 1}h(n, n + 1)\P(E_{n}, E_{n+1} \in [t, t + h]) + \sum_{\substack{n_1 \ge 1 \\ n_2 > n_1 + 1}}h(n_1, n_2) \P(E_{n_1}, E_{n_2} \in [t, t + h]) \\
&= \sum_{n \ge 1}h(n, n + 1)\sum_{k = 0}^{n - 1} \P(U = k) \P(V \ge n + 1 - k) \\
&\hspace{15em}+ \sum_{\substack{n_1 \ge 1 \\ n_2 > n_1 + 1}} h(n_1, n_2) \sum_{k = 0}^{n_1 - 1}\P(U = k) \P(V \ge n_2 - k) \\
&= \sum_{n \ge 1}h(n, n + 1)\P(U = n - 1) \P(V \ge 2) + \sum_{n \ge 2}h(n, n + 1)\sum_{k = 0}^{n - 2} \P(U = k) \P(V \ge n+1-k) \\
&\hspace{15em}+ \sum_{\substack{n_1 \ge 1 \\ n_2 > n_1 + 1}} h(n_1, n_2) \sum_{k = 0}^{n_1 - 1}\P(U = k) \P(V \ge n_2 - k).
\end{align*}
Therefore
\begin{align*}
0 \le &\sum_{\substack{n_1 \ge 1 \\  n_2 > n_1}} h(n_1, n_2) \P(E_{n_1}, E_{n_2} \in [t, t + h]) - e^{-\lambda_{t, h}}\frac{\lambda_{t, h}^{2}}{2}\sum_{n \ge 1} h(n, n + 1) \P(U = n - 1) \\
&= \sum_{n \ge 1}h(n, n + 1)\P(U = n - 1) \P(V \ge 3) + \sum_{n \ge 2}h(n, n + 1) \sum_{k = 0}^{n - 2} \P(U = k) \P(V \ge n+1-k) \\
&\hspace{15em}+ \sum_{\substack{n_1 \ge 1 \\ n_2 > n_1 + 1}} h(n_1, n_2) \sum_{k = 0}^{n_1 - 1}\P(U = k) \P(V \ge n_2 - k) \\
&\le \P(V \ge 3) \sum_{n \ge 1}\P(U = n - 1) + \sum_{k \ge 0} \P(U = k) \sum_{n \ge k + 2} \P(V \ge n + 1 - k) \\
&\hspace{15em}+ \sum_{k \ge 0} \P(U = k) \sum_{n_1 \ge k + 1} \sum_{\ell \ge 2} \P(V \ge n_1 + \ell - k) \\
&\le \lambda_{t, h}^3 \sum_{n \ge 1}\P(U = n - 1) + \sum_{k \ge 0} \P(U = k) \sum_{\ell \ge 3} \lambda_{t, h}^{\ell} \\
&\hspace{15em}+ \sum_{k \ge 0} \P(U = k) \sum_{n_1 \ge k + 1} \sum_{\ell \ge 2} \lambda_{t, h}^{n_1 + \ell - k} \quad (\text{using \eqref{eq:poi_tail_trivial}}) \\
&= \lambda_{t,h}^{3}+ \frac{\lambda_{t, h}^{3}}{1-\lambda_{t, h}} + \sum_{k \ge 0} \P(U = k) \sum_{n_1 \ge k + 1} \frac{\lambda_{t, h}^{n_1 + 2 - k}}{1-\lambda_{t, h}} \\
&= \lambda_{t,h}^{3}+ \frac{\lambda_{t, h}^{3}}{1-\lambda_{t, h}} + \frac{\lambda_{t, h}^{3}}{(1-\lambda_{t, h})^2}.
\end{align*}
This completes the proof.
\end{proof}
We now use Lemma~\ref{lem:est_En} to estimate various moments of $X_k[t, t + h]$.
\begin{lemma}\label{lem:est_xk-moments}
We have
\begin{align}
    \E X_{k + 1}[t, t + h] &= e^{-\lambda_{t, h}} \lambda_{t, h} H_t(k) + R_1(t, h), \\ 
    \E (X_{k + 1}[t, t + h])^2 &= e^{-\lambda_{t, h}} (\lambda_{t, h} H_t(k) + \lambda_{t, h}^2 G_t(k, k)) + R_2(t, h), \\ 
    \E (X_{k + 1}[t, t + h] X_{k' + 1}[t, t + h]) &= e^{-\lambda_{t, h}} \frac{\lambda_{t, h}^2}{2} (G_t(k, k') + G_t(k', k)) + R_3(t, h),
\end{align}
where
\begin{align*}
    |R_1(t, h)| &\le \lambda_{t, h}^2 \bigg(1 + \frac{1}{1 - \lambda_{t, h}}\bigg), \\
    |R_2(t, h)| &\le \lambda_{t, h} \bigg(1 + \frac{1}{1 - \lambda_{t, h}}\bigg) + 2\lambda_{t, h}^2 \bigg(1 + \frac{1}{1 - \lambda_{t, h}} + \frac{1}{(1 - \lambda_{t, h})^2}\bigg), \\
    |R_3(t, h)| &\le 2\lambda_{t, h}^3 \bigg(1 + \frac{1}{1 - \lambda_{t, h}} + \frac{1}{(1 - \lambda_{t, h})^2}\bigg).
\end{align*}
\end{lemma}

\begin{proof}
We have
\begin{align*}
    \E X_{k + 1}[t, t + h] &= \E \sum_{n \ge 1} \I(\#\C_{n} = k)\bone_{[t, t+h]}(E_{n}) \\
                           &= \sum_{n\ge 1} \P(\#\C_{n} = k) \P(E_n \in [t, t + h]) \\
                           &= e^{-\lambda_{t, h}}\lambda_{t, h}\sum_{n \ge 1}\P(\#\C_{n} = k)\P(U = n - 1) + R_1(t, h) \\
                           &= e^{-\lambda_{t, h}}\lambda_{t, h}H_t(k) + R_1(t, h),
\end{align*}
where $|R_1(t, h)| \le \lambda_{t, h}^2 (1 + \frac{1}{1 - \lambda_{t, h}})$ by \eqref{eq:bound_En}.

On the other hand,
\begin{align*}
    \E(X_{k + 1}[t, &t + h])^2 \\
    &= \E \bigg(\sum_{n_1 \ge 1} \I(\#\C_{n_1} = k) \bone_{[t, t + h]}(E_{n_1})\bigg) \bigg(\sum_{n_2 \ge 1} \I(\#\C_{n_2} = k) \bone_{[t, t + h]}(E_{n_2}) \bigg) \\
    &= \sum_{n_1} \P(\#\C_{n_1} = k) \P(E_{n_1} \in [t, t + h]) \\
    &\qquad\qquad + 2\sum_{\substack{n_1 \ge 1 \\ n_2 > n_1}} \P(\#\C_{n_1} = k, \#\C_{n_2} = k) \P(E_{n_1}, E_{n_2} \in [t, t + h]) \\
    &= e^{-\lambda_{t, h}}\lambda_{t, h}\sum_{n \ge 1}\P \bigg(\#\C_{n} = k\bigg)\P(U = n - 1) \\
    &\qquad\qquad + 2 e^{-\lambda_{t, h}}\frac{\lambda_{t, h}^{2}}{2}\sum_{n \ge 1} \P(\#\C_{n} = k, \#\C_{n + 1} = k') \P(U = n - 1) + R_2(t, h) \\ 
    &= e^{-\lambda_{t, h}}\bigg(\lambda_{t, h} H_t(k) + \frac{\lambda_{t, h}^{2}}{2} G_t(k, k)\bigg) + R_2(t, h),
\end{align*}
where $|R_2(t, h)| \le \lambda_{t, h} (1 + \frac{1}{1 - \lambda_{t, h}}) + 2\lambda_{t, h}^3 (1 + \frac{1}{1 - \lambda_{t, h}} + \frac{1}{(1 - \lambda_{t, h})^2})$ by \eqref{eq:bound_En} and \eqref{eq:bound_En1_En2}.

Finally,
\begin{align*}
    \E (X_{k + 1}[t, t + h] &X_{k' + 1}[t, t + h]) \\
    &= \E \bigg(\sum_{n_1 \ge 1} \I(\#\C_{n_1} = k) \bone_{[t, t + h]}(E_{n_1})\bigg) \bigg(\sum_{n_2 \ge 1} \I(\#\C_{n_2} = k') \bone_{[t, t + h]}(E_{n_2}) \bigg) \\
    &= \sum_{\substack{n_1 \ge 1 \\ n_2 > n_1}} \P(\#\C_{n_1} = k, \#\C_{n_2} = k') \P(E_{n_1}, E_{n_2} \in [t, t + h]) \\
    &\qquad\qquad + \sum_{\substack{n_2 \ge 1 \\ n_1 > n_2}}\P(\#\C_{n_1} = k, \#\C_{n_2} = k') \P(E_{n_1}, E_{n_2} \in [t, t + h]) \\
    &=e^{-\lambda_{t, h}}\frac{\lambda_{t, h}^{2}}{2}\sum_{n \ge 1} \P(\#\C_{n} = k, \#\C_{n + 1} = k') \P(U = n - 1) \\ 
    &\qquad\qquad + e^{-\lambda_{t, h}}\frac{\lambda_{t, h}^{2}}{2}\sum_{n \ge 1} \P(\#\C_{n + 1} = k, \#\C_{n} = k') \P(U = n - 1) + R_3(t, h)\\
    &= e^{-\lambda_{t, h}}\frac{\lambda_{t, h}^{2}}{2} (G_t(k, k') + G_t(k', k)) + R_3(t, h), \\
\end{align*}
where $|R_3(t, h)| \le 2\lambda_{t, h}^3 (1 + \frac{1}{1 - \lambda_{t, h}} + \frac{1}{(1 - \lambda_{t, h})^2})$ by \eqref{eq:bound_En1_En2}.
\end{proof}

\begin{proof}[Proof of Theorem~\ref{thm:xk-moments}]
Follows from the estimates in Lemma~\ref{lem:est_xk-moments} and the fact that $\lim_{h \downarrow 0}\frac{\lambda_{t, h}}{h} = \lambda(t)$.
\end{proof}

\begin{proof}[Proof of Lemma~\ref{lem:ind-expression}]
Note that
\begin{align*}
    I_\varphi[s, t] &= \frac{\sum_{k \ge 1}\varphi(k)X_k[s, t]}{N[s, t]} \\
    &= \frac{\sum_{k \ge 1}\varphi(k)\sum_{n \ge 1} \I(\#\C_n = k - 1) \bone_{[s, t]}(E_n)}{N[s, t]} \\
    &= \sum_{n \ge 1} \frac{\bone_{[s, t]}(E_n)}{N[s, t]} \sum_{k \ge 1} \varphi(k) \I(\#\C_n = k - 1) \\
    &= \sum_{n \ge 1} \frac{\bone_{[s, t]}(E_n)}{N[s, t]} \varphi(\# \C_n + 1).
\end{align*}
The desired result now follows from the independence of $N$ and $\C_n$.
\end{proof}

\begin{proof}[Proof of Lemma~\ref{lem:UV}]
Note that $\bone_{[s, t]}(E_n) = 1$ if and only if $U \le n - 1$ and $U + V \ge n$. 
\end{proof}

\begin{proof}[Proof of Corollary~\ref{cor:const_F}]
That $\C_n$'s are independent when $F_n(k) = p_n$ for all $0 \le k \le n - 1$ is clear. In fact, in this case $\#\C_n \sim \mathrm{Binomial}(L, p_n)$.

If $F_n(k) = p$ for all $0 \le k \le n - 1, n \ge 1$, then $\#C_n \overset{d}{=} B \sim \mathrm{Binomial}(L, p)$. Therefore, using the fact that the sets $\C_n$, $n \ge 1$, are independent across themselves and also of $U = N[0, t]$, we have
\begin{align*}
    G_t(k, k') &= \P(\#\C_{U + 1} = k, \#\C_{U + 2} = k') \\
    &= \E_U \P(\#\C_{U + 1} = k, \#\C_{U + 2} = k' \mid U) \\
    &= \E_U \P(\#\C_{U + 1} = k) \P(\#\C_{U + 2} = k') \\
    &= \P(B = k) \P(B = k').
\end{align*}
Similarly, $G_t(k', k) = \P(B = k) \P(B = k')$, and $H_t(k) = \P(B = k)$.
Part (i) now follows Theorem~\ref{thm:xk-moments}.

Since $\E \varphi(\#\C_n + 1) = \E \varphi(B + 1)$, we have from Lemma~\ref{lem:ind-expression} that
\begin{align*}
    \E I_{\varphi}[t, t + h] &= \sum_{n \ge 1} \E \bigg[ \frac{\bone_{[t, t+h]}(E_n)}{N[t, t+h]} \bigg] \E \varphi(\# \C_n+1) \\
    &= \E \varphi(B + 1) \sum_{n \ge 1} \E \bigg[ \frac{\bone_{[t, t+h]}(E_n)}{N[t, t+h]} \bigg]\\
    &= \E \varphi(B + 1) \E \sum_{n \ge 1} \frac{\bone_{[t, t+h]}(E_n)}{N[t, t+h]} \\
    &= \E \varphi(B + 1).
\end{align*}
This completes the proof of part (ii).
\end{proof}

\begin{proof}[Proof of Lemma~\ref{lem:linear_F}]
We have
\begin{align*}
    \E [\#\C_{n} \mid \C_1, \ldots, \C_{n - 1}] &= \sum_{i = 1}^L (a_{n} m_{n - 1, i} + b_{n}) \\
    &= a_n \sum_{i = 1}^L m_{n - 1, i} + L b_n \\ 
    &= a_n \sum_{i = 1}^L \sum_{\ell = 1}^{n - 1} \bone_{\C_{\ell}}(i) + L b_n \\ 
    &= a_n \sum_{\ell = 1}^{n - 1} \sum_{i = 1}^L \bone_{\C_{\ell}}(i) + L b_n \\ 
    &= a_n \sum_{\ell = 1}^{n - 1} \#\C_{\ell} + L b_n.
\end{align*}
Now the desired result follows by taking expectation with respect to $\C_1, \ldots, \C_{n - 1}$.
\end{proof}

\begin{proof}[Proof of Lemma~\ref{lem:linear_F_soln}]
Write $A_n = \E\#\C_n$. Define $S_n = \sum_{\ell = 1}^n A_{\ell}$, with $S_0 := 0$. Then we have for any $n \ge 1$,
\[
    A_n = a_n S_{n - 1} + L b_n
\]
and
\[
    S_n = S_{n - 1} + A_n = (1 + a_n) S_{n - 1} + L b_n.
\]
By iterating backwards, it is easy to see that
\begin{align*}
    S_n &= (1 + a_n) S_{n - 1} + L b_n \\
        &= (1 + a_n) (1 + a_{n - 1}) S_{n - 2} + L(b_n + b_{n - 1} (1 + a_n)) \\
        &= (1 + a_n) (1 + a_{n - 1}) (1 + a_{n - 2}) S_{n - 3} + L(b_n + b_{n - 1}(1 + a_n) + b_{n - 2} (1 + a_{n - 1})(1 + a_n)) \\
        &\,\,\,\vdots \\
        &= \prod_{j = 2}^{n} (1 + a_j) S_0 + L \bigg(b_n + \sum_{j = 1}^{n - 1} b_j \prod_{\ell = j + 1}^n (1 + a_{\ell})\bigg) \\
        &= L \bigg(b_n + \sum_{j = 1}^{n - 1} b_j \prod_{\ell = j + 1}^n (1 + a_{\ell})\bigg).
\end{align*}
Thus
\begin{align*}
    A_n &= S_n - S_{n - 1} \\
        &= L\bigg(b_n + \sum_{j = 1}^{n - 1} b_j \prod_{\ell = j + 1}^n (1 + a_{\ell}) - b_{n - 1} - \sum_{j = 1}^{n - 2} b_j \prod_{\ell = j + 1}^{n - 1} (1 + a_{\ell})\bigg) \\
        &= L\bigg(b_n + b_{n - 1} a_n + \sum_{j = 1}^{n - 2} b_j a_n \prod_{\ell = j + 1}^{n - 1} (1 + a_{\ell})\bigg).
\end{align*}
This completes the proof.
\end{proof}

\begin{proof}[Proof of Theorem~\ref{thm:ind_rate}]
Recall that
\begin{align*}
    \E I_\varphi[t, t + h] = \sum_{n \ge 1} \E \bigg[ \frac{\bone_{[t, t + h]}(E_n)}{N[t, t + h]} \bigg] \E \varphi(\#\C_n + 1).
\end{align*}
Since
\[
    \E \bigg[ \frac{\bone_{[t, t + h]}(E_n)}{N[t, t + h]} \bigg] = \sum_{k = 0}^{n - 1} \P(U = k) \E\bigg[\frac{\I(V \ge n - k)}{V}\bigg]
    \ge \P(U = n - 1) \P(V = 1),
\]
we have
\begin{align*}
    \E I_\varphi[t, t + h] &\ge \sum_{n \ge 1} \P(U = n - 1) \P(V = 1) \E\varphi(\#\C_n + 1) \\
    &= \P(V = 1) \sum_{n \ge 1} \P(U = n - 1) \E\varphi(\#\C_n + 1) \\
    &= e^{-\lambda_{t, h}} \lambda_{t, h} \E \varphi(\#\C_{U + 1} + 1).
\end{align*}
On the other hand, since $\E \big[ \frac{\bone_{[t, t + h]}(E_n)}{N[t, t + h]} \big] \le \P(E_n \in [t, t + h])$, we have
\begin{align*}
    \E I_\varphi[t, t + h] &\le \sum_{n \ge 1} \E \varphi(\#\C_n + 1) \P(E_n \in [t, t + h]) \\
    &= \varphi(L + 1) \sum_{n \ge 1} \frac{\E \varphi(\#\C_n + 1)}{\varphi(L + 1)} \P(E_n \in [t, t + h]) \\
    &= \varphi(L + 1) \bigg[e^{-\lambda_{t, h}} \lambda_{t, h} \sum_{n \ge 1} \frac{\E \varphi(\#\C_n + 1)}{\varphi(L + 1)} \P(U = n - 1) + R_4(t, h)\bigg] \\
    &= e^{-\lambda_{t, h}} \lambda_{t, h} \sum_{n \ge 1} \E \varphi(\#\C_n + 1) \P(U = n - 1) + \varphi(L + 1) R_4(t, h) \\
    &= e^{-\lambda_{t, h}} \lambda_{t, h} \E \varphi(\#\C_{U + 1} + 1) + \varphi(L + 1) R_4(t, h),
\end{align*}
where $|R_4(t, h)| \le \lambda_{t, h}^2(1 + \frac{1}{1 - \lambda_{t, h}})$.

We conclude that
\[
    |\E I_\varphi[t, t + h] - e^{-\lambda_{t, h}} \lambda_{t, h} \E \varphi(\#\C_{U + 1} + 1)| \le \varphi(L + 1) R_4(t, h).
\]
The desired result now follows from the above estimate. 
\end{proof}

Let $\mu_{n, r}$ be the $r$-th moment of $m_{n, 1}$, given by
\[
    \mu_{n, r} := \sum_{k = 1}^{n} k^r p_{n, k}.
\]

\begin{proof}[Proof of Theorem~\ref{thm:ests-well-defined}]
    (a) Let us use the convention that for $n \ge 0$, $F_n(-1) = 0, p_{0, 0} = p_{n, n + 1} = 0$. Note that $p_{1, 0} = 1 - F_1(0) > 0$, and for $k \in [n]$, we have 
\begin{align*}
    p_{n, k} &= \P(m_{n, 1} = k) \\
    &= \P(m_{n, 1} = k, m_{n - 1, 1} = k - 1) + \P(m_{n, 1} = k, m_{n - 1, 1} = k) \\
    &= \P(\bone_{C_n}(1) = 1 \mid m_{n - 1, 1} = k - 1) \P(m_{n - 1, 1} = k - 1) \\
    &\qquad\qquad\qquad + \P(\bone_{C_n}(1)  = 0 \mid m_{n - 1, 1} = k) \P(m_{n - 1, 1} = k) \\
    &= F_n(k - 1) p_{n - 1, k - 1} + (1 - F_n(k)) p_{n - 1, k}.
\end{align*} 
   
By induction, we get that $p_{n, k} > 0$ for all $n \ge 1$, $k = 0, 1, \ldots, n$. By the strong law of large numbers, as $L \to \infty$, we have that
\begin{equation}\label{eq:slln}
    \frac{1}{L}\sum_{i \in [L]}\I\{m_{n, i} = k\} \xrightarrow{\mathrm{a.s.}} p_{n, k} > 0.
\end{equation}
This completes the proof.

(b) Fix $n$. As $p_{n,k} > 0$ for $k \in [n]$, we have $\var(m_{n,1}) := \mu_{n, 2} - \mu^2_{n, 1} > 0$. By the strong law of large numbers, as $L \to \infty$, we have that 
\[
    \frac{1}{L}\sum_{i \in [L]}m^2_{n,i} \xrightarrow{\mathrm{a.s.}} \mu_{n, 2}, \quad \frac{1}{L}\sum_{i \in [L]}m_{n,i} \xrightarrow{\mathrm{a.s.}} \mu_{n, 1}.
\]
Therefore, as $L \to \infty$,
\[
    \frac{1}{L}\sum_{i \in [L]} m^2_{n, i} - \bigg(\frac{1}{L}\sum_{i \in [L]}m_{n, i}\bigg)^2 \xrightarrow{\mathrm{a.s.}} \mu_{n, 2} - \mu^2_{n, 1} > 0.
\]
This completes the proof.
\end{proof}

\begin{proof}[Proof of Theorem~\ref{thm:consistency-nonparam}]
Define the random variables 
\[
    X_i = \I(m_{n - 1, i} = k) \bone_{\C_{n}}(i), \quad Y_i = \I(m_{n - 1, i} = k), \quad i \in [L].
\]
so that $X_i \stackrel{\mathrm{i.i.d.}}{\sim} X_1, Y_i \stackrel{\mathrm{i.i.d.}}{\sim} Y_1$. Also, $X_1, Y_1$ are Bernoulli random variables with
\[
   \P(Y_1 = 1) = \P(m_{n - 1, 1} = k) = p_{n - 1, k}
\]
and
\begin{align*}
    \P(X_1 = 1) &= \P(X_1 = 1, Y_1 = 1) \\
                        &= \P(m_{n - 1, 1} = k, \bone_{\C_{n}}(1) = 1) \\
                        &= \P(\bone_{\C_{n}}(1) = 1 \mid m_{n - 1, 1} = k) \P(m_{n - 1, 1} = k) \\
                        &= F_{n}(k) p_{n - 1, k}.
\end{align*}
Let 
\begin{align*}
    \mu_x &:= \E(X_1) = F_n(k) p_{n - 1, k}, \\
    \mu_y &:= \E(Y_1) = p_{n - 1, k}, \\
    \sigma_x^2 &:= \var(X_1) = F_n(k) p_{n - 1, k} (1 - F_n(k) p_{n - 1, k}), \\
    \sigma_y^2 &:= \var(Y_1) = p_{n - 1, k}(1 - p_{n - 1, k}), \\
    \rho\sigma_{x}\sigma_{y} &:= \cov(X_1, Y_1) = F_n(k) p_{n - 1, k} - F_n(k) p^2_{n - 1, k}.
\end{align*}
Using Lemma \ref{lem:wlln-clt}, we immediately get that
\[
    \hat{F}_n(k) = \frac{\sum_{i \in [L]}\I(m_{n - 1, i} = k) \bone_{\C_{n}}(i)}{\sum_{i \in [L]}\I(m_{n - 1, i} = k)} \xrightarrow{p} \frac{F_n(k) p_{n - 1, k}}{p_{n - 1, k}} = F_{n}(k)
\]
and
\[
    \sqrt{L}\bigg(\hat{F}_n(k) - F_n(k)\bigg) = \sqrt{L}\bigg(\frac{\sum_{i \in [L]} \I(m_{n - 1, i} = k) \bone_{\C_{n}}(i)}{\sum_{i \in [L]}\I(m_{n - 1, i} = k)} - \frac{F_n(k) p_{n - 1, k}}{p_{n - 1, k}}\bigg)
\]
converges in distribution to a normal random variable with mean $0$ and variance
\begin{align*}
    &\frac{p^3_{n - 1, k} F_n(k) (1 - F_n(k) p_{n - 1, k}) + F^2_n(k) p^3_{n - 1, k} (1 - p_{n - 1, k}) - 2 F^2_n(k) p^3_{n - 1, k}(1 - p_{n - 1, k})}{p^4_{n - 1, k}} \\
    &= \frac{F_n(k) (1 - F_n(k) p_{n - 1, k}) + F^2_n(k) (1 - p_{n - 1, k}) - 2 F^2_n(k) (1 - p_{n - 1, k})}{p_{n - 1, k}} \\
    &= \frac{F_n(k) (1 - F_n(k))}{p_{n - 1, k}}.
\end{align*}
This completes the proof.
\end{proof}

\begin{proof}[Proof of Corollary~\ref{cor:confint}]
From \eqref{eq:slln}, \eqref{eq:asconv}, \eqref{eq:av} and \eqref{eq:ev}, we get that
\begin{align}\label{eq:vconv}
   \hat{\sigma}^2 \xrightarrow{p} {\sigma}^2.
\end{align}
Then \eqref{eq:wconv}, \eqref{eq:vconv} and Slutsky's theorem implies that
\[
    \frac{\sqrt{L}\big(\hat{F}_n(k) - F_n(k)\big)}{\hat{\sigma}} \xrightarrow{d} N(0, 1).
\]
In other words,
\begin{align*}
    1 - q &= \lim_{L \to \infty} \P\bigg(-z_{q/2} \le  \frac{\sqrt{L}\big(\hat{F}_n(k) - F_n(k)\big)}{\hat{\sigma}} \le z_{q/2}\bigg) \\
          &= \lim_{L \to \infty} \P\bigg({F}_n(k) \in \bigg(\hat{F}_n(k) - \frac{\hat{\sigma}}{\sqrt{L}} z_{q/2}, \hat{F}_n(k) + \frac{\hat{\sigma}}{\sqrt{L}} z_{q/2}\bigg)\bigg).
\end{align*}
This completes the proof.
\end{proof}

\begin{proof}[Proof of Theorem~\ref{thm:consistency-param}]
Define the random variables
\[
    X_i := \bone_{\C_{n}}(i), \quad Y_i := m_{n - 1, i}, \quad i \in [L]
\]
so that $X_i \stackrel{\mathrm{i.i.d.}}{\sim} X_1, Y_i \stackrel{\mathrm{i.i.d.}}{\sim} Y_1$. Then $X_1$ is a Bernoulli random variable and $Y_1$ is a discrete random variable taking values in $\{0\} \cup [n - 1]$ with
\begin{align*}
    \P(Y_1 = k) = \P(m_{n - 1, 1} = k) = p_{n - 1, k}.
\end{align*}
For $\lambda > 0$, let
\[
    s_{0, b} := \E(Y_1^{\lambda}) = \mu_{n - 1, \lambda},
\]
and for $\gamma > 0, \lambda \ge 0$, let 
\begin{align*}
    s_{\gamma, \lambda} := \E(X_1^{\gamma} Y_1^{\lambda}) &= \sum_{k = 1}^{n - 1} k^{\lambda} \P(X_1 = 1, Y_1 = k) \\
                                                                  &= \sum_{k = 1}^{n - 1} k^{\lambda} \P(\bone_{\C_{n}}(1) = 1 \mid m_{n - 1, 1} = k) \P(m_{n - 1, 1} = k) \\
                                                                  &= \sum_{k = 1}^{n - 1} k^{\lambda} F_n(k) p_{n - 1, k} \\
                                                                  &= a_n \sum_{k = 1}^{n - 1} k^{\lambda + 1} p_{n - 1, k} + b_n \sum_{k = 1}^{n - 1} k^{\lambda} p_{n - 1, k} \\
                                                                  &= a_n \mu_{n - 1, \lambda + 1} + b_n \mu_{n - 1, \lambda}.
\end{align*}
Therefore
\begin{align*}
    \Sigma = ((\sigma_{ij})) &:= \cov(X_1, Y_1, Y_1^2, X_1 Y_1) \\
                             &= \begin{bmatrix} 
                                    s_{1, 0} - s^2_{1, 0} & s_{1, 1} - s_{1, 0} s_{0, 1} & s_{1, 2} - s_{1, 0} s_{0, 2} & s_{2, 1} - s_{1, 0} s_{1, 1} \\
                                    s_{1, 1} - s_{1, 0} s_{0, 1} & s_{0, 2} - s^2_{0, 1} & s_{0, 3} - s_{0, 1} s_{0, 2} & s_{1, 2} - s_{1, 0} s_{1, 1} \\
                                    s_{1, 2} - s_{1, 0} s_{0, 2} & s_{0, 3} - s_{0, 1} s_{0, 2} & s_{0, 4} - s_{0, 2}^2 & s_{1, 3} - s_{0, 2} s_{1, 1} \\
                                    s_{2, 1} - s_{1, 0} s_{1, 1} & s_{1, 2} - s_{1, 0} s_{1, 1} & s_{1, 3} - s_{0, 2} s_{1, 1} & s_{2, 2} - s^2_{1, 1}
                                \end{bmatrix}
\end{align*}
and
\begin{align*}
    \mu_x &:= \E(X_1) = a_n\mu_{n - 1, 1} + b_n, \\
    \mu_y &:= \E(Y_1) = \mu_{n - 1, 1}, \\ 
    \sigma_x^2 &:= \sigma_{11} = (a_n \mu_{n - 1, 1} + b_n) (1 - a_n\mu_{n - 1, 1} - b_n), \\
    \sigma_y^2 &:= \sigma_{22} = \mu_{n - 1, 2} - \mu_{n - 1, 1}^2, \\ 
    \beta &:= \frac{\rho \sigma_x \sigma_y}{\sigma_y^2} = \frac{\sigma_{12}}{\sigma_{22}} = \frac{a_n \mu_{n - 1, 2} + b_n \mu_{n - 1, 1} - (a_n \mu_{n - 1, 1} + b_n) \mu_{n - 1, 1}}{\mu_{n - 1, 2} - \mu_{n - 1, 1}^2} = a_n, \\
    \alpha &:= \mu_x - \beta \mu_y = a_n \mu_{n - 1, 1} + b_n - a_n \mu_{n - 1, 1} = b_n.
\end{align*}
Thus $\hat{\beta} = \hat{a}_n$, $\hat{\alpha} = \hat{b}_n$. Then, parts (a) and (b) of Lemma \ref{lem:pest} immediately imply that 
\begin{align*}
    \hat{a}_n \xrightarrow{p} a_n, \quad \hat{b}_n \xrightarrow{p} b_n
\end{align*}
and
\begin{align*}
    \sqrt{L}(\hat{a}_n - a_n) \xrightarrow{d} N(0, \sigma_a^2), \quad \sqrt{L}(\hat{b}_n - b_n) \xrightarrow{d} N(0, \sigma_b^2),
\end{align*}
where
\begin{align*}
    \sigma_a^2 &= \frac{\mu_y^2\sigma_{11}}{\sigma_y^4} + \frac{2\alpha\mu_y\sigma_{12}}{\sigma_y^2} + \frac{2\beta\mu_y\sigma_{13}}{\sigma_y^4} - \frac{2\mu_y\sigma_{14}}{\sigma_y^4} + \frac{\alpha^2\sigma_{22}}{\sigma_y^4} \\
               &\qquad\qquad + \frac{2\alpha\beta\sigma_{23}}{\sigma_y^4} - \frac{2\alpha\sigma_{24}}{\sigma_y^4} + \frac{\beta^2\sigma_{33}}{\sigma_y^4} - \frac{2\beta\sigma_{34}}{\sigma_y^4} + \frac{\sigma_{44}}{\sigma_y^4}
\end{align*}
and
\begin{align*}
    \sigma_b^2 &= \bigg(1 + \frac{\mu_y^2}{\sigma_y^2}\bigg)^2 \sigma_{11} - 2\bigg(1 + \frac{\mu_y^2}{\sigma_y^2}\bigg) \bigg(\frac{\beta\sigma_y^2 + 2\beta\mu_y^2 - \mu_x\mu_y}{\sigma_y^2}\bigg) \sigma_{12} + 2\bigg(1 + \frac{\mu_y^2}{\sigma_y^2}\bigg) \frac{\beta\mu_y\sigma_{13}}{\sigma_y^2} \\
               &\qquad - 2\bigg(1 + \frac{\mu_y^2}{\sigma_y^2}\bigg) \frac{\mu_y\sigma_{14}}{\sigma_y^2} + \bigg(\frac{\beta\sigma^2_y + 2\beta\mu_y^2 - \mu_x\mu_y}{\sigma_y^2}\bigg)^2 \sigma_{22} \\
               &\qquad - 2\bigg(\frac{\beta\sigma_y^2 + 2\beta\mu_y^2 - \mu_x\mu_y}{\sigma_y^2}\bigg) \frac{\beta\mu_y\sigma_{23}}{\sigma_y^2} \\
               &\qquad + 2\bigg(\frac{\beta\sigma_y^2 + 2\beta\mu_y^2 - \mu_x\mu_y}{\sigma_y^2}\bigg)\frac{\mu_y\sigma_{24}}{\sigma_y^2} + \frac{\beta^2\mu_y^2\sigma_{33}}{\sigma_y^4} - \frac{2\beta\mu_y^2\sigma_{34}}{\sigma_y^4} + \frac{\mu_y^2\sigma_{44}}{\sigma_y^4}.
\end{align*}
This completes the proof.
\end{proof}

\begin{proof}[Proof of Corollary~\ref{cor:confint-param}]
Part (c) of Lemma \ref{lem:pest} implies the existence of consistent estimators $\hat{\sigma}^2_a$ and $\hat{\sigma}^2_b$ of $\sigma^2_a$ and $\sigma^2_b$, respectively.
Then \eqref{eq:anpest} and Slutsky's theorem gives that
\[
    \frac{\sqrt{L}(\hat{a}_n - a_n)}{\hat{\sigma}_a} \xrightarrow{d} N(0, 1)
\]
and
\[
    \frac{\sqrt{L}(\hat{b}_n - b_n)}{\hat{\sigma}_b} \xrightarrow{d} N(0, 1).
\]
In other words,
\begin{align*}
    1 - q &= \lim_{L \to \infty} \P\bigg(-z_{q/2} \le  \frac{\sqrt{L}(\hat{a}_n - a_n)}{\hat{\sigma}_a} \le z_{q/2}\bigg) \\
          &= \lim_{L \to \infty}\P\bigg(a_n \in \bigg(\hat{a}_n - \frac{\hat{\sigma}_a}{\sqrt{L}}z_{q/2}, \hat{a}_n + \frac{\hat{\sigma}_a}{\sqrt{L}}z_{q/2}\bigg)\bigg)
\end{align*}
and
\begin{align*}
   1 - q &= \lim_{L \to \infty}\P\bigg(-z_{q/2} \le  \frac{\sqrt{L}(\hat{b}_n - b_n)}{\hat{\sigma}_b} \le z_{q/2}\bigg) \\
         &= \lim_{L \to \infty}\P\bigg(b_n \in \bigg(\hat{b}_n - \frac{\hat{\sigma}_b}{\sqrt{L}}z_{q/2}, \hat{b}_n + \frac{\hat{\sigma}_b}{\sqrt{L}}z_{q/2}\bigg)\bigg).
\end{align*}
This completes the proof.
\end{proof}

\section{Technical lemmas}\label{sec:tech-lem}
\begin{lemma}\label{lem:wlln-clt}
Suppose that $(X_i, Y_i), i \ge 1,$ are i.i.d. nonnegative bivariate random vectors with $\E(X_1) = \mu_x$, $\E(Y_1) = \mu_y \ne 0$, $\var(X_1) = \sigma_x^2$, $\var(Y_1) = \sigma_y^2$ and $\corr(X_1, Y_1) = \rho$. Then
\[
    \frac{\sum_{i = 1}^n X_i}{\sum_{i = 1}^n Y_i} \xrightarrow{p} \frac{\mu_x}{\mu_y} \quad \text{and} \quad \sqrt{n}\bigg(\frac{\sum_{i = 1}^n X_i}{\sum_{i = 1}^n Y_i} - \frac{\mu_x}{\mu_y}\bigg) \xrightarrow{d} N(0, \sigma_0^2),
\]
where
\[
    \sigma_0^2 = \frac{\mu_y^2 \sigma_x^2 + \mu_x^2 \sigma_y^2 - 2 \rho \mu_x \mu_y \sigma_x \sigma_y}{\mu_y^4}.
\]
\end{lemma}
\begin{proof}
By the weak law of large numbers,
\[
    (\bar{X}_n, \bar{Y}_n) := \bigg(\frac{1}{n}\sum_{i = 1}^n X_i, \frac{1}{n}\sum_{i = 1}^n Y_i\bigg) \xrightarrow{p} (\mu_x, \mu_y).
\]
Applying the continuous mapping theorem with the map
\[
    (x,y) \mapsto \frac{x}{y}\bone\{y\ne0\},
\]
which has only finitely many discontinuity points, we get 
\[
    \frac{\sum_{i = 1}^n X_i}{\sum_{i = 1}^n Y_i} \bone\bigg\{\sum_{i = 1}^n Y_i \ne 0\bigg\} = \frac{\bar{X}_n}{\bar{Y}_n}\bone\bigg\{\bar{Y}_n \ne 0\bigg\} \xrightarrow{p} \frac{\mu_{x}}{\mu_{y}}.
\]
Let $Z_i = X_i - \frac{\mu_x}{\mu_y}Y_i$. Note that $\{Z_i\}_{i \ge 1}$ are i.i.d. as $\{(X_{i}, Y_{i})\}_{i \ge 1}$ are so.
Now 
\[
    \E(Z_i) = \E(X_i) - \frac{\mu_x}{\mu_y}\E(Y_i) = \mu_x - \frac{\mu_x}{\mu_y}\mu_y = 0
\]
and
\begin{align*}
    V := \var(Z_i) &= \var(X_i) + \bigg(\frac{\mu_x}{\mu_y}\bigg)^2 \var(Y_i) - 2\frac{\mu_x}{\mu_y} \cov(X_i, Y_i) \\
                    &= \sigma_x^2 + \bigg(\frac{\mu_x}{\mu_y}\bigg)^2\sigma_y^2 - 2\frac{\mu_x}{\mu_y}\rho\sigma_x\sigma_y \\
                    &= \frac{1}{\mu_y^2}\big(\mu_y^2 \sigma_x^2 + \mu_x^2 \sigma_y^2 - 2 \rho \mu_x \mu_y \sigma_x \sigma_y\big).
\end{align*}
Therefore, using the central limit theorem, we get that
\begin{align*}
    \sqrt{n}\bar{Z}_n \xrightarrow{d} Z \sim N(0, V).
\end{align*}
By Slutsky's Theorem and using the fact that $\mu_y \ne 0$, we get that
\[
    \sqrt{n}\bigg(\frac{\sum_{i = 1}^n X_i}{\sum_{i = 1}^n Y_i} - \frac{\mu_x}{\mu_y}\bigg) = \frac{\sqrt{n}}{n\bar{Y}_n}\bigg(\sum_{i = 1}^n\bigg(X_i - \frac{\mu_x}{\mu_y}Y_i\bigg)\bigg) = \sqrt{n}\frac{\bar{Z}_n}{\bar{Y}_n} \xrightarrow{d} \frac{Z}{\mu_y} \sim N(0, V / \mu_y^2).
\]
This completes the proof.
\end{proof}

\begin{lemma}\label{lem:pest}
Suppose that $(X_i, Y_i), i \ge 1$, are i.i.d. bivariate random vectors with $\E(X_1) = \mu_x$, $\E(Y_1) = \mu_y$, $\var(X_1) = \sigma_x^2$, $\var(Y_1) = \sigma_y^2 > 0$, $\corr(X_1, Y_1) = \rho$ and 
\[
    \Sigma = ((\sigma_{ij})) := \cov(X_1, Y_1, Y_1^2, X_1 Y_1).
\]
Then we have the following:
\begin{itemize}
    \item[(a)] Let $\beta = \frac{\rho\sigma_x}{\sigma_y}$ and
    \[
        \hat{\beta} = \frac{\frac{1}{n}\sum_{i = 1}^n X_i Y_i - (\frac{1}{n}\sum_{i = 1}^n X_i) (\frac{1}{n}\sum_{i = 1}^n Y_i)}{\frac{1}{n}\sum_{i = 1}^n Y_i^2 - (\frac{1}{n}\sum_{i = 1}^n Y_i)^2}.
    \]
    Then 
    \[
        \hat{\beta} \xrightarrow{p} \beta \quad \text{and} \quad \sqrt{n}(\hat{\beta} - \beta) \xrightarrow{d} N(0, \sigma_1^2),
    \]
    where
    \begin{align*}
        \sigma_1^2 &= \frac{\mu_y^2\sigma_{11}}{\sigma_y^4} + \frac{2\alpha\mu_y\sigma_{12}}{\sigma_y^4} + \frac{2\beta\mu_y\sigma_{13}}{\sigma_y^4} - \frac{2\mu_y\sigma_{14}}{\sigma_y^4} + \frac{\alpha^2\sigma_{22}}{\sigma_y^4} \\
                   &\qquad\qquad + \frac{2\alpha\beta\sigma_{23}}{\sigma_y^4} - \frac{2\alpha\sigma_{24}}{\sigma_y^4} + \frac{\beta^2\sigma_{33}}{\sigma_y^4} - 2\frac{2\beta\sigma_{34}}{\sigma_y^4} + \frac{\sigma_{44}}{\sigma_y^4}.
    \end{align*}

    \item[(b)] Let $\alpha = \mu_x - \beta\mu_y$ and 
    \[
        \hat{\alpha} = \frac{1}{n}\sum_{i = 1}^n X_i - \frac{\hat{\beta}}{n}\sum_{i = 1}^n Y_i.
    \]
    Then 
    \[
        \hat{\alpha} \xrightarrow{p} \alpha \quad \text{and} \quad \sqrt{n}(\hat{\alpha} - \alpha) \xrightarrow{d} N(0, \sigma_2^2),
    \]
    where
    \begin{align*}
        \sigma_2^2 &= \bigg(1 + \frac{\mu_y^2}{\sigma_y^2}\bigg)^2 \sigma_{11} - 2\bigg(1 + \frac{\mu_y^2}{\sigma_y^2}\bigg)\bigg(\frac{\beta\sigma_y^2 + 2\beta\mu_y^2 - \mu_x\mu_y}{\sigma_y^2}\bigg) \sigma_{12} \\
                   &\qquad + 2\bigg(1 + \frac{\mu_y^2}{\sigma_y^2}\bigg)\frac{\beta\mu_y\sigma_{13}}{\sigma_y^2} - 2\bigg(1 + \frac{\mu_y^2}{\sigma_y^2}\bigg)\frac{\mu_y\sigma_{14}}{\sigma_y^2} + \bigg(\frac{\beta\sigma_y^2 + 2\beta\mu_y^2 - \mu_x\mu_y}{\sigma_y^2}\bigg)^2 \sigma_{22} \\
                   &\qquad - 2\bigg(\frac{\beta\sigma_y^2 + 2\beta\mu_y^2 - \mu_x\mu_y}{\sigma_y^2}\bigg)\frac{\beta\mu_y\sigma_{23}}{\sigma_y^2} \\
                   &\qquad + 2\bigg(\frac{\beta\sigma_y^2 + 2\beta\mu_y^2 - \mu_x\mu_y}{\sigma_y^2}\bigg)\frac{\mu_y\sigma_{24}}{\sigma_y^2} + \frac{\beta^2\mu_y^2\sigma_{33}}{\sigma_y^4} - \frac{2\beta\mu_y^2\sigma_{34}}{\sigma_y^4} + \frac{\mu_y^2\sigma_{44}}{\sigma_y^4}.
    \end{align*}

    \item[(c)] Let $i \in [2]$. There exists an estimator $\hat{\sigma}^2_i$ of $\sigma^2_i$ such that $\hat{\sigma}^2_i \xrightarrow{p} \sigma^2_i$.
\end{itemize}
\end{lemma}

\begin{proof}
Let
\[
    \bZ^{\top}_i = (X_i, Y_i, Y_i^2, X_i Y_i) \quad \text{and} \quad \btheta^{\top} = (\mu_x, \mu_y, \sigma_y^2, \beta\sigma_y^2).
\] 
By the weak law of large numbers and the multivariate central limit theorem,
\begin{align*}
    \bar{\bZ}_n \xrightarrow{p} \btheta, \quad \sqrt{n}(\bar{\bZ}_n - \btheta) \xrightarrow{d} N(0, \Sigma),
\end{align*}
where $\Sigma = ((\sigma_{ij})) := \cov(\bZ_1)$.
\begin{itemize}
    \item[(a)] Consider the function $g_1 : \RR^4 \rightarrow \RR$ given by
    \[
        g_1(x_1, x_2, x_3, x_4) = \begin{cases}
            \frac{x_4 - x_1 x_2}{x_3 - x_2^2} & \text{if } (x_1, x_2, x_3, x_4) \in A, \\
            0 & \text{otherwise},
        \end{cases}
    \]
    where $A = \{(x_1, x_2, x_3, x_4) : x_3 - x_2^2 > 0\}$. As $\sigma_y^2 > 0$, $g_1(\bar{\bZ}_n) = \hat{\beta}$ with probability $1$ for all sufficiently large $n$ and $g_1(\btheta) = \beta$. As $g_1$ is continuous at $\btheta$, the continuous mapping theorem gives that
    \[
        \hat{\beta} = g_1(\bar{\bZ}_n) \xrightarrow{p} g_1(\btheta) = \beta.
    \]
    Also, as $g_1$ is differentiable at $\btheta$, an application of the delta method yields 
    \[
        \sqrt{n}(\hat{\beta} - \beta) = \sqrt{n}(g_1(\bar{\bZ}_n) - g_1(\btheta)) \xrightarrow{d} N(0, \nabla g_1(\btheta)^\top \Sigma \nabla g_1(\btheta)).
    \]
    Now, since
    \[
        \nabla g_1(\btheta) = \bigg(-\frac{\mu_y}{\sigma_y^2}, -\frac{\alpha}{\sigma_y^2}, -\frac{\beta}{\sigma_y^2}, \frac{1}{\sigma_y^2}\bigg),
    \]
    an explicit computation shows that $\nabla g_1(\btheta)^\top \Sigma \nabla g_1(\btheta) = \sigma_1^2$. 

    \vskip5pt
    \item[(b)] Consider the function $g_2 : \RR^4 \rightarrow \RR$ given by:
    \[
        g_2(x_1, x_2, x_3, x_4) = \begin{cases}
            x_1 - \frac{x_2(x_4 - x_1 x_2)}{x_3 - x_2^2} & \text{if } (x_1, x_2, x_3, x_4) \in A, \\
            0 & \text{otherwise},
        \end{cases}
    \]
    where $A = \{(x_1, x_2, x_3, x_4) : x_3 - x_2^2 > 0\}$. As $\sigma_y^2 > 0$, $g_2(\bar{\bZ}_n) = \hat{\alpha}$ with probability $1$ for all sufficiently large $n$ and $g_2(\btheta) = \alpha$. As $g_2$ is continuous at $\btheta$, the continuous mapping theorem implies that
    \[
        \hat{\alpha} = g_2(\bar{\bZ}_n) \xrightarrow{p} g_2(\btheta) = \alpha.
    \]
    Also, as $g_2$ is differentiable at $\btheta$, an application of the delta method yields
    \[
        \sqrt{n}(\hat{\alpha} - \alpha) = \sqrt{n}(g_2(\bar{\bZ}_n) - g_2(\btheta)) \xrightarrow{d} N(0, \nabla g_2(\btheta)^\top \Sigma \nabla g_2(\btheta)).
    \]
    Now, since
    \[
        \nabla g_2(\btheta) = \bigg(1 + \frac{\mu_y^2}{\sigma_y^2}, -\frac{\beta\sigma_y^2 + 2\beta\mu_y^2 - \mu_x\mu_y}{\sigma_y^2}, \frac{\beta\mu_y}{\sigma_y^2}, -\frac{\mu_y}{\sigma_y^2}\bigg),
    \]
    an explicit computation shows that $\nabla g_2 (\btheta)^\top \Sigma\nabla g_2(\btheta) = \sigma_2^2$.

    \vskip5pt
    \item[(c)] Fix $i \in [2]$. As $\nabla g_i$ is continuous at $\btheta$, the continuous mapping theorem implies that
    \[
        \nabla g_i(\bar{\bZ}_n) \xrightarrow{p} \nabla g_i(\btheta).
    \]
    By the weak law of large numbers,
    \[
        \hat{\Sigma} := \frac{1}{n}\sum_{i = 1}^n \bZ_i \bZ_i^\top - \bar{\bZ}_n \bar{\bZ}_n^\top \xrightarrow{p} \Sigma.
    \]
    Hence
    \[
        \hat{\sigma}_i^2 := \nabla g_i(\bar{\bZ}_n)^\top \hat{\Sigma} \nabla g_i(\bar{\bZ}_n) \xrightarrow{p} \nabla g_i(\btheta)^\top \Sigma \nabla g_i(\btheta) = \sigma^2_i.
    \]
\end{itemize}
\end{proof}
\end{document}